\documentclass[final,5p,times,sort&compress,twocolumn]{elsarticle}

\newtheorem{theorem}{Theorem}
\newtheorem{proof}{Proof}
\newtheorem{lemma}{Lemma}
\newtheorem{corollary}{Corollary}

\usepackage{amssymb}
\usepackage{amsmath,amsfonts}
\usepackage{graphicx}
\usepackage{subfigure,epsfig,color}
\usepackage{textgreek}
\usepackage{multirow}
\usepackage{array}
\usepackage{epstopdf}
\usepackage{natbib}
\usepackage{steinmetz}
\usepackage{breqn}
\usepackage{booktabs}
\usepackage{cleveref}

\usepackage{fancyhdr}
\usepackage{color}
\usepackage{secdot}
\usepackage{lipsum}
\usepackage{multirow}
\usepackage{booktabs}
\usepackage{adjustbox}
\usepackage{epstopdf}
\usepackage{epsfig}
\usepackage[linesnumbered,lined,ruled]{algorithm2e}
\usepackage{amssymb}
\usepackage{amsmath}
\usepackage{tabu}
\usepackage{float}
\usepackage{mathtools,url}
\usepackage{amssymb,bm}
\usepackage{mathtools}

\usepackage{steinmetz}
\usepackage{breqn}
\usepackage{booktabs}
\usepackage{cleveref}
\usepackage{suffix}

\journal{Physical Communication Journal}

\makeatletter
\renewcommand{\fnum@figure}{Fig. \thefigure}
\makeatother

\DeclarePairedDelimiterX \MeijerM[3]{\lparen}{\rparen}%
{\,#3\delimsize\vert\begin{smallmatrix}#1 \\ #2\end{smallmatrix}}

\newcommand\MeijerG[8][]{%
	G^{\,#2,#3}_{#4,#5}\MeijerM[#1]{#6}{#7}{#8}}

\WithSuffix\newcommand\MeijerG*[7]{%
	G^{\,#1,#2}_{#3,#4}\MeijerM*{#5}{#6}{#7}}

\begin{document}

\begin{frontmatter}


\title{Capacity Analysis for Joint Radar-Communication Capable Coherent MIMO Radars}
\author[label1]{Muharrem Arik}
\ead{marik@ku.edu.tr}
\ead[url]{https://nwcl.ku.edu.tr/marik.html}

\author[label1,label2]{Ozgur B. Akan}
\ead{akan@ku.edu.tr, oba21@cam.ac.uk}
\ead[url]{https://ioe.eng.cam.ac.uk/directory/akan}

\address[label1]{Next-generation and Wireless Communications Laboratory (NWCL)\\ 
Department of Electrical and Electronics Engineering\\
Koc University, Istanbul, 34450, Turkey}
\address[label2]{Internet of Everything (IoE) Group\\
Electrical Engineering Division, Department of Engineering\\
University of Cambridge, Cambridge, CB3 0FA, UK}




\begin{abstract}
Recently, huge attention is attracted to the concept of integrating communication and radar missions within the same platform. Joint Radar-Communications (JRC) system gives an important opportunity to reduce spectrum usage and product cost while doing concurrent operation, as target sensing via radar processing and establishing communication links. A JRC-capable coherent MIMO radar system have been proposed recently in the literature. Several methods are introduced to reach dual goal as a notable null level towards the direction of interest of the radar and MIMO radar waveform orthogonality. Due to the limitations originated form the JRC operation, communication channel may encounter unwanted amplitude variations. This unwanted modulation normally affects the communication performance by its nature, due to the fades on radiated signal amplitude towards the direction of communication. However, the effect of this unintentional modulation on communication channel is yet to be investigated. In this paper, the communication channel for JRC capable phase-coded coherent MIMO radars is analyzed and investigated under additive white Gaussian noise and Rayleigh\textcolor{blue}{/Rician} fading conditions. Communication capacity is evaluated for each channel condition. The results reveal that, using the single-side limited null direction fixed waveform generation method displays the best capacity performance under all channel conditions.
\end{abstract}

\end{frontmatter}



\section{Introduction}
\label{intro}
Several techniques are proposed over the years to enable JRC concept with different types of radar and communication systems and different topologies \cite{surveyRFCommRadar}. Almost all types of radars and all types of modulation are becoming the subject to this research domain. Mostly, efforts are concentrated on enhancements on data rates, modulation types and transparency over radar sensing operation. 

Up to now, JRC systems are investigated mostly on radar perspective and the majority of the researchers have spent more efforts on finding out how much performance the radar system loses when communication is enabled. However, due to the limited resources and being a secondary mission on radar system, communication operation has often remained behind the scenes for a JRC system. To assess the success of a communication channel, the main issue to consider is channel capacity. Basically, the channel capacity is the information theoretic performance metric for a communication channel. There are several researches on communication channel capacity or performance limits of a JRC system. In \cite{innerBoundsRadarComm}, radar estimation rate concept was proposed to investigate JRC information limits. This rate is a measure of the amount of information collected from target in terms of information which might be exchanged uncooperatively between target and radar as a function of time. This type of metrics are used to optimize collaborative operation between radar and communication systems. There are many researches which consider communication capable radar waveform design. Each work tries to find the best waveform that provides invisible communication when flawless radar sensing operation is conducted. In \cite{highSNRchannelCap}, authors investigate high-SNR capacity of a communications channel based on some coding alphabet in the presence of AWGN. This alphabet presents a hybrid approach which constitutes a dual optimize JRC operation.
 
\begin{figure}[htb]
	\centering
	\includegraphics[width=0.50\textwidth]{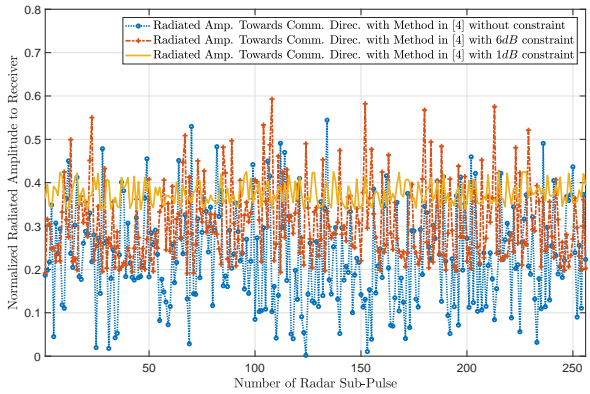}
	\caption{Unintentional amplitude modulation on radiated signal towards communication direction for JRC capable coherent MIMO radars.}
	\label{fig:unintenionalModCohRadar}
\end{figure}

Recently, a JRC-capable monostatic coherent MIMO radar system have been proposed in \cite{RealcMIMO}. Coherent MIMO radars offer fully coherent signal processing and coherent transmit beamforming which presents radiation patterns to reduce the probability of intercept. In \cite{RealcMIMO}, MIMO radar waveform coding technique in \cite{MIMOradarWaveformOrth} is modified for enabling communication without disturbing coherent MIMO radar structure. JRC ability is accomplished for coherent MIMO radar without disturbing the orthogonality and transmit beamforming requirements. Three different novel communication methods are proposed in \cite{RealcMIMO}. First method uses each sub-pulse in the radar pulse for communication. This method is directional and it may provide thousands of bits per radar pulse. Second method exploits coding mechanism through the radar pulse which is only visible at communication direction. The phase of the information vector is modulated for each radar pulse. Last communication method utilizes the progressive phase difference between orthogonal waveforms radiated from the first and the last antenna element caused by small degrees of steering on transmit beampattern. \textcolor{blue}{ Specifically, null direction fixed phase-coded coherent MIMO radar waveform generation method in \cite{RealcMIMO} produces an unintentional amplitude modulation during the transmission of the radar pulse towards the communication direction. In \cite{RealcMIMO}, in order to provide amplitude stable channel towards receiver, some constraints are introduced during the waveform generation phase. Fig.\ref{fig:unintenionalModCohRadar} displays the fading like amplitude variations during the radar pulse with $3dB$ and $6dB$ constraints. Blue curve in Fig.\ref{fig:unintenionalModCohRadar} shows the amplitude modulation without any constraint applied during waveform generation.} This modulation affects the communication performance by its nature due to possible fades on signal amplitude through the radar pulse.

In this paper, first, we give channel models for different cases, then, we analyze the distribution characteristics of the unintentional modulation on signal amplitude. After that, using this distribution, we try to reach the capacity expressions for different communication methods under additive white Gaussian noise (AWGN) and Rayleigh\textcolor{blue}{/Rician} fading conditions. Lastly, we evaluate the probability density function (pdf) of the distributions and capacity expressions for various communication methods proposed for JRC capable coherent MIMO radar. While \textcolor{blue}{analyzing} the unintentional modulation under Rayleigh\textcolor{blue}{/Rician} fading conditions, pdf of the distribution of product of a single or double truncated Rayleigh r.v. and a Rayleigh\textcolor{blue}{/Rician} r.v. is firstly given in the literature.
\begin{table}[t]
	\caption{List of major symbols and notations}
	\label{tab:notation}
	\resizebox{.9\columnwidth}{!}{
		\begin{tabular}{c l}
			\hline\\
			$M$ & Number of MIMO radar antenna element\\
			$L$ & Number of sub-pulse in a radar pulse\\
			$\textbf{s}$ & $M\times1$ transmit signal vector\\
			$P_t$ & Total transmit power budget\\
			$P_c(\theta_c)$ & Average transmit power towards communication direction\\
			$f_c$ & Carrier signal frequency\\
			$\textbf{w}_l$ & $1\times M$ transmit beamforming vector of the sub-pulse $l$\\
			$\boldsymbol{\Phi}$ & $L\times M$ \textcolor{blue}{space-time phase coding matrix}\\
			$\theta_{c}$ & Communication receiver direction\\
			$\theta_{n}$ & Null direction towards target or clutter\\
			$\textbf{a}_{tr}(\theta)$ & $M\times1$ steering vector of the transmit array towards $\theta$\\
			$G_l(\theta)$ & Radiated signal towards the direction of $\theta$\\
			$A_l$ & Magnitude of the radiated signal towards communication direction $\theta_{c}$\\
			$H_t$ & complex channel gain with amplitude $H$ and phase $\phi_H$ at time $t$\\
			$\hat{s}_{r}$ & transmitted signal towards $\theta_{c}$, $G_l(\theta_{c})s(t)$\\
			$\hat{s}^{comm}_{r}$ & $\hat{s}_{r}$ when communication direction fixed method in \cite{RealcMIMO} is used, $G_1(\theta_{c})s(t)$\\
			$\hat{s}^{null}_{r}$ & $\hat{s}_{r}$ when null direction fixed method in \cite{RealcMIMO} is used, $G_l(\theta_{c})s(t)$\\
			$\Delta$ & Positive number that defines communication pre-distortion level\\
			$\sigma_{N}^{2}$ & Additive White Gaussian Noise (AWGN) variance\\
			$\sigma_{H}^{2}$ & variance of the both real/imaginary part of the complex r.v. $H_t$\\
			$(.)^{*}$ & Complex conjugate operator\\
			$(.)^{\dagger}$ & Complex conjugate transpose operator\\
			\hline
		\end{tabular}
	}
\end{table}
The remainder is organized as follows. Firstly, the channel model is given in Sec. II. Then, amplitude distribution of the unintentional modulation is investigated in Sec. III and capacity expressions for JRC capable coherent MIMO radars are detailed in Sec. IV. Sec. V provides the numerical evaluation results. Lastly, conclusions are discussed in Sec. VI. For ease of exposition, the major notations used throughout this paper are listed in Table \ref{tab:notation}.

\section{Channel Model}
\label{SigMod}

\begin{figure}[t]
	\centering
	\includegraphics[width=0.5\textwidth]{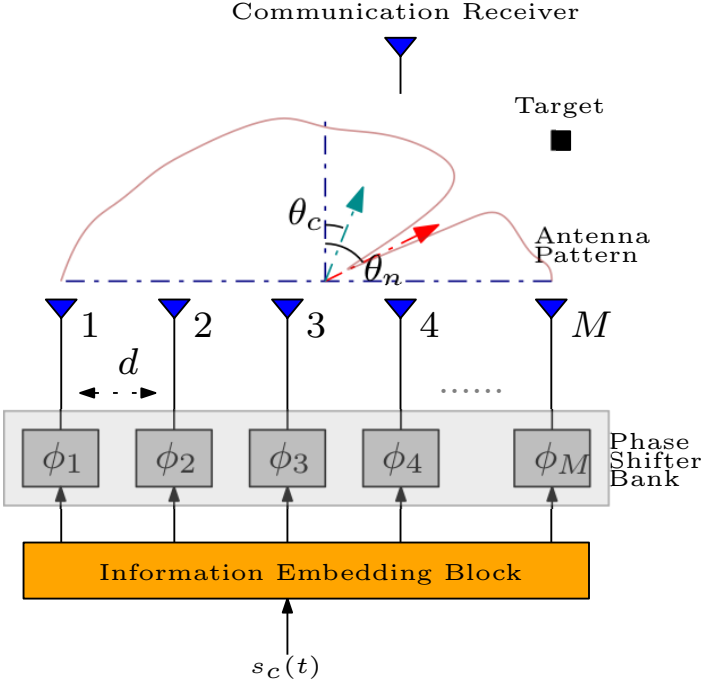}
	\caption{JRC enabled Coherent MIMO radar transmitter architecture.}
	\label{fig:jrcCoherentMIMO}
\end{figure}
We utilize the same channel model as proposed in our previous work in \cite{RealcMIMO}. \textcolor{blue}{Corresponding} JRC system is equipped with $M$ \textcolor{blue}{element} MIMO radar antenna aligned as a uniform linear array (ULA) as in Fig. \ref{fig:jrcCoherentMIMO}. In the array, each antenna element has an isotropic radiation pattern with spacing distance of $d$ in terms of a wavelength. Moreover, each antenna radiates distinct phase coded quasi-orthogonal waveforms at the same carrier frequency. $M\times1$ transmit signal vector during the $l$-th sub-pulse of the $r$-th radar pulse is given as,
\begin{equation}
\label{st}
\textbf{s}(t;r;l)=\textbf{w}_{l}^{\dagger}(r)s_{c}(t), l=1,2,...,L,
\end{equation}where $r$ and $l$ are the radar pulse and sub-pulse index, respectively and $t$ is the fast time index, $(.)^{\dagger}$ stands for the complex conjugate transpose and carrier signal can be given as $s_{c}(t)=\sqrt{P_{t}/M}e^{j2\pi f_{c}t}$, where $P_{t}$ is the total transmit power budget and $f_{c}$ is the radar carrier \textcolor{blue}{ frequency.} Then, $\textbf{w}_{l}$ is the $1\times M$ transmit beamforming weight vector of the sub-pulse $l$ as, $\textbf{w}_{l}=[e^{j\phi_{1}(l)}, e^{j\phi_{2}(l)}, ..., e^{j\phi_{M}(l)}]$, $l=1,2,...,L$, where $\phi_{m}(l)$ is the phase of the sub-pulse $l$ from the antenna element $m$.
For each radar pulse, \textcolor{blue}{angle of the} $L\times M$ space-time phase coding matrix for $M$ antennas and $L$ sub-pulses can be given as,
\begin{equation}
\label{stMatrix}
\begin{split}
\boldsymbol{\varPhi} =&\angle{\boldsymbol{\Phi}}=\left[ 
\begin{matrix}
\angle\textbf{w}_{1}\\
\angle\textbf{w}_{2}\\
\vdots\\
\angle\textbf{w}_{L}
\end{matrix}\right]
=\left[ 
\begin{matrix}
\phi_{1}(1)&\phi_{2}(1)& ... &\phi_{M}(1)\\
\phi_{1}(2)&\phi_{2}(2)& ... &\phi_{M}(2)\\
\vdots&\vdots&\vdots&\vdots\\
\phi_{1}(L)&\phi_{2}(L)& ... &\phi_{M}(L)\\
\end{matrix}\right],
\end{split}
\end{equation}where $\angle$ is the phasor angle function.

At the receiver side, single omni-directional antenna element is connected to a communication receiver located in direction $\theta_{c}$, which is assumed to be known from the receiver. Besides, JRC enabled communication receiver is equipped with a matched filter. The received signal at the communication receiver for sub-pulse $l$ can be expressed as,
\begin{equation}
\label{yt}
y(t;r;l)=h(t)\textbf{a}_{tr}^{T}(\theta_{c})\textbf{s}(t;r;l)+\eta(t),
\end{equation}where $\textbf{a}_{tr}(\theta)$ is the $M\times1$ steering vector of the transmit array toward the spatial angle $\theta$,
\begin{equation}
\textbf{a}_{tr}(\theta)=\left[1,e^{-j2\pi d\sin(\theta)},...,e^{-j2\pi (M-1) d\sin(\theta)}\right]^{T}
\end{equation}and $h(t)$ is the complex channel coefficient which reflects the propagation gain between the transmit coherent MIMO array and the communication receiver. Then, $\eta(t)$ is the additive white Gaussian noise with zero mean and variance $\sigma_{N}^{2}$. By combining (\ref{st}) and (\ref{yt}), the baseband received signal for sub-pulse $l$ can be written as,
\begin{equation}
\label{yt2}
y_{b}(t;r;l)=h(t)\sqrt{P_{t}/M}\textbf{a}^T_{tr}(\theta_{c})\textbf{w}_{l}^{\dagger}(r)+\eta(t),
\end{equation}After this point, we have dropped the $r$ index in the equations for the simplicity. Using information embedding strategies via manipulating transmit signal vector $\textbf{s}(t;l)$, communication symbols can be buried into the radar signal.
\begin{figure}[t]
	\centering
	\includegraphics[width=0.5\textwidth]{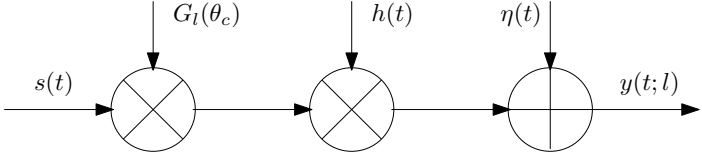}
	\caption{Channel Model for JRC capable coherent MIMO radar}
	\label{fig:cMIMOChannelModel_NullF}
\end{figure}

For the coherent MIMO radar case, the channel model for JRC capable coherent MIMO radar is expressed by,
\begin{equation}
\label{ytch}
\begin{split}
y(t;l)=&h(t)\hat{s}_{r}(t;l)+\eta(t)\\
=&h(t)\textbf{a}_{tr}^{T}(\theta_{c})\textbf{s}(t;l)+\eta(t)\\
=&h(t)G_l(\theta_{c})s(t)+\eta(t),
\end{split}
\end{equation}where $G_{l}(\theta_{c})=\textbf{a}^T_{tr}(\theta_{c})\textbf{w}_{l}^{\dagger}(r)$, $s(t)$ is the channel input \textcolor{blue}{and can be given as information modulated carrier signal, $s(t)=e^{\chi_k}s_c(t)$, where $e^{\chi_k}$ is the complex information symbol.} $\hat{s}_{r}(t;l)$ can be defined as transmitted signal towards communication direction from the MIMO array. When communication fixed waveform generation method in \cite{RealcMIMO} is applied, channel model can be simplified as,
\textcolor{blue}{
\begin{equation}
\label{ytch_com}
\begin{split}
y(t;l)=&h(t)\hat{s}^{comm}_{r}(t;l)+\eta(t)\\
=&h(t)G_1(\theta_{c})s(t)+\eta(t),
\end{split}
\end{equation}}where $\hat{s}^{comm}_{r}(t;l)$ is defined as transmitted signal towards communication direction from the MIMO array when communication direction fixed waveform generation method in \cite{RealcMIMO} is used. $G_{1}(\theta_{c})$ is defined as the radiated signal towards communication direction $\theta_{c}$ for the first sub-pulse. 

Null direction fixed method in \cite{RealcMIMO} aims to fix the radiation signal at null direction using the technique in \cite{MIMOradarWaveformOrth, MIMOradarWaveformOrthRadarConf}, while satisfying higher signal levels for the communication direction. This method slightly changes the channel model for the communication due to the undesired variations on the amplitude level expressed in Eq.(\ref{ytch_null}). This unintentional modulation shows random behavior and for the sake of simplicity, we have assumed that this unintentional modulation is the part of the channel effect. Therefore, we consider the channel model as in Fig.\ref{fig:cMIMOChannelModel_NullF} \textcolor{blue}{and it is} given by,
\begin{equation}
\label{ytch_null}
\begin{split}
y(t;l)=&h(t)e^{j\varsigma(l)}\textbf{a}^T_{tr}(\theta_{c})\textbf{w}_{l}^{\dagger}(r)s(t)+\eta(t)\\
=&h(t)e^{j\varsigma(l)}G_l(\theta_{c})s(t)+\eta(t)\\
=&h(t)e^{j\varsigma(l)}A_le^{-j\bar{\phi_l}}s(t)+\eta(t)\\
=&h(t)A_le^{-j\bar{\phi_1}}s(t)+\eta(t),
\end{split}
\end{equation}where $\hat{s}^{null}_{r}(t;l)$ is the channel input when null direction fixed waveform generation method in \cite{RealcMIMO} is used and $A_l=|G_{l}(\theta_{c})|$ is defined as amplitude of the radiated signal towards communication direction $\theta_{c}$. Although null direction fixed method introduces unintentional modulation, a phase correction procedure is proposed to fix the phase toward receiver. In \cite{RealcMIMO}, the correction factor $\varsigma(l)$ is given as the phase angle differences between radiated signal at the first and the $l$th sub-pulse, $\angle{G_{1}(\theta_{c})}-\angle{G_{l}(\theta_{c})}$ for $l=2,...,L$. After applying this procedure, the phase component is fixed \textcolor{blue}{to $e^{-j\bar{\phi_1}}$ along} the radar pulse. This phase correction leads phase stable radiated signal. However, radiated signal amplitude towards communication direction still contains unintentional variations which may cause a degradation on communication performance.
\section{Amplitude Distribution of the Unintentional Modulation}
\label{AmpDistribution}
In order to reach the channel capacity of the null direction \textcolor{blue}{fixed} method, it is necessary to characterize the distribution of the unintentional modulation. First, we \textcolor{blue}{must} rewrite the radiated signal $G_l(\theta_{c})$. When null direction fixed method is applied, the radiated signal $G_l$ becomes,
\begin{equation}
\label{transmitSignal}
\begin{split}
G_{l}(\theta_{c})=&\textbf{a}^T_{tr}(\theta_{c})\textbf{w}_{l}^{\dagger}(r)\\
=&\left[1,e^{-j2\pi d\sin(\theta_{c})},...,e^{-j2\pi (M-1) d\sin(\theta_{c})}\right]\\
&\times\left[
\begin{matrix}
e^{-j[\phi_{I_{1}(l)}+2\pi (I_{1}(l)-1)d\sin{\theta_{n}}]}\\
\vdots\\ 
e^{-j[\phi_{I_{M}(l)}+2\pi (I_{M}(l)-M)d\sin{\theta_{n}}]}
\end{matrix}
\right]\\
=&\sum_{m=1}^{M}e^{-j\left[\phi_{I_m(l)}+2\pi d\left[m(\sin(\theta_{c})-\sin(\theta_{n}))+I_m(l)\sin(\theta_{n})-\sin(\theta_{c})\right]\right]}
\end{split}
\end{equation}where $\{I_{1},I_{2},...,I_{M}\}$ are the indexes of a random permutation sequence of $\{1,2,...,M\}$ and $\phi_{I_{m}}(l)$ is the new phase component of the $m$-th antenna element for the $l$-th sub-pulse \cite{RealcMIMO}. Then, this unintentional modulation caused by $G_l(\theta_{c})$ can be rewritten as a sum of individual exponentials as,
\begin{equation}
\label{sumofexponentials}
G_l(\theta_{c})=\sum_{m=1}^{M}e^{-j\phi_m(l)}=A_le^{-j\bar{\phi_l}}.
\end{equation}For different $l$ values, $\bar{\phi}$ is a phase variable in the range $(-\pi,\pi]$ and $A_l$ is the radiated signal amplitude variable in the range $[0,M]$.
\begin{theorem}
	\label{theoremRay}
	Amplitude of the radiated signal towards communication direction for JRC enable\textcolor{blue}{d} coherent MIMO radar when null direction fixed method \textcolor{blue}{is} applied presents a Rayleigh distribution with the following probability density function, 
	\begin{equation}
	\begin{split}
	p_{A_l}(\varrho)=\frac{2\varrho}{M}e^{-\varrho^2/M},\quad \varrho>0, \iff & L>M\gg1 \text{ , }\\
	&\theta_{c}\ne\theta_{n} \text{ and } L,M\in\mathbb{Z}^+,
	\end{split}
	\end{equation}where $M$ is the number of transmit antenna elements in the MIMO array, $l=1,...,L$ is the sub-pulse index and $L$ is the total number of sub-pulse in a radar pulse.
\end{theorem}
\begin{proof}
	See the \ref{ProofUnintAmpProof}.
\end{proof}
\section{Capacity Analysis}
\label{capacity}
In this part, \textcolor{blue}{channel} capacity \textcolor{blue}{is analyzed with the assumption} that channel distribution information is known by the receiver. Then, we assumed that, the receiver perfectly knows the Channel State Information (CSI). For two waveform generation methods in \cite{RealcMIMO}, the channel capacity is expressed in detail. Then, capacity expressions are given for AWGN and Rayleigh\textcolor{blue}{Rician} channels.
\subsection{Communication Direction Fixed Method - Stable Channel}
This method guarantees the radiated signal is kept constant during the radar pulse. Therefore, the channel becomes an AWGN channel since $h(t)=1$ and $G_1(\theta_{c})$ is constant during the radar pulse. Then, the channel model in Eq.(\ref{ytch_com}) can be rewritten as,
\begin{equation}
y(t;l)=G_1(\theta_{c})s(t)+\eta(t)=A_1e^{-j\bar{\phi_1}}s(t)+\eta(t).
\end{equation}The channel capacity for the communication direction fixed method with AWGN channel can be defined in nats as, $C_s=B\ln(1+\gamma)$, where $B$ is the channel bandwidth and $\gamma$ is the channel Signal-Noise Ratio (SNR). SNR can be given as, $\gamma = P_c(\theta_c)/(\sigma^{2}_{N}B)$, where $P_c(\theta_c)$ is the average transmit power towards communication direction, $\sigma^{2}_{N}$ is the variance of the zero mean AWGN. For the communication direction fixed method, $P_c$ is equal to $(A_0^2P_t)/M$, where $A_0$ is the radiated signal amplitude through communication direction of the reference transmit beampattern, $|G_0(\theta_{c})|$. Then, the capacity \textcolor{blue}{expression} becomes,
\begin{equation}
\label{CsnrComm2}
C_s=B\ln\left(1+\frac{A_0^2P_t}{\sigma^{2}_{N}BM}\right)=B\ln\left(1+\frac{A_0^2c_\gamma}{M}\right).
\end{equation}where $c_\gamma=P_t/(\sigma^{2}_{N}B)$. $A_0$ can be selected as, the probability of $A_l$ is greater than $A_0$ is equal to $p_0$, $Pr[A_l\ge A_0]=p_0$. Since $A_l$ has Rayleigh distribution with the pdf in Theorem \ref{theoremRay}, $p_0$ can be given as,
\begin{equation}
\label{CsnrComm3}
p_0=\int_{A_0}^{\infty}\frac{2\varrho}{M}e^{-\varrho^2/M}d\varrho=-e^{\frac{-\varrho^2}{M}}\bigg\rvert_{A_0}^{\infty}=e^{\frac{-A_0^2}{M}}.
\end{equation}Then, $A_0$ \textcolor{blue}{becomes} $\sqrt{-M\ln(p_0)}$. Hence, the capacity formula can be rewritten as,
\begin{equation}
\label{CsnrComm4}
C_s=B\ln\left(1+\frac{(\sqrt{-M\ln(p_0)})^2c_\gamma}{M}\right)=B\ln\left(1-c_\gamma\ln(p_0)\right).
\end{equation}
\subsection{Null Direction Fixed Method - Unintentionally Modulated Channel}
Since the amplitude distribution of $A_l$, \textcolor{blue}{which is equal to} $|G_l(\theta_{c})|$, shows Rayleigh distribution for the null direction fixed method, the unintentional modulation affects the channel as \textcolor{blue}{the} Rayleigh fading does. Then, using Eq.(\ref{ytch_null}) the channel can be represented as,
\begin{equation}
\label{ndfmUnintChMod}
y(t;l)=h(t)A_le^{-j\bar{\phi_1}}s(t)+\eta(t),
\end{equation}for $l=1,...,L$. The channel has the Rayleigh distribution by the following pdf, if there is no constraint applied in the null direction method, i.e. $0\ge A_l\ge\infty$, 
\begin{equation}
p_{A_l}(\varrho)=\frac{2\varrho}{M}e^{-\varrho^2/M}, \varrho>0,
\end{equation}where $A_l=|G_l(\theta_{c})|$. In \cite{RealcMIMO}, null direction method requires applying amplitude constraint as \textcolor{blue}{$A^{min}\le A_l\le A^{max}$, where $A^{min}=A_1/\varDelta$ and $A^{max}=A_1\varDelta$.} Since, the distribution is modified to a double truncated Rayleigh distribution, the probability density distribution is changed to a double truncated Rayleigh pdf as,
\textcolor{blue}{
\begin{equation}
\label{amppdf}
p_{A_l}(\varrho)=\frac{\frac{2\varrho}{M}e^{-\varrho^2/M}}{\int_{A^{min}}^{A^{max}}\frac{2\varrho}{M}e^{-\varrho^2/M}d\varrho}=\frac{2\varrho e^{-\varrho^2/M}}{\alpha(\tilde{A}_{1},\tilde{A}_{2}) M},\quad A^{min}\le \varrho\le A^{max},
\end{equation}}where \textcolor{blue}{$\alpha(\tilde{A}_{1},\tilde{A}_{2})=e^{-(A^{min})^2/M}-e^{-(A^{max})^2/M}$, $\tilde{A}_{1}=(A^{min})^2/M$ and $\tilde{A}_{2}=(A^{max})^2/M$.} By the transformation theorem for single random variables, the channel power gain ${A_l}^2$ has an exponential distribution with the mean $M$. The pdf of ${A_l}^2$ is
\textcolor{blue}{
\begin{equation}
\label{snrpdf1}
p_{{A_l}^2}(\bar{a})=\frac{1}{\alpha(\tilde{A}_{1},\tilde{A}_{2}) M}e^{-\bar{a}/M},\quad (A^{min})^2\le \bar{a}\le (A^{max})^2.
\end{equation}}Using the (\ref{snrpdf1}), pdf of the SNR, $\gamma$ becomes,
\textcolor{blue}{
\begin{equation}
\label{snrpdf2}
p_{\gamma}(a)=\frac{1}{\alpha(\tilde{A}_{1},\tilde{A}_{2}) c_\gamma}e^{-a/c_\gamma},\quad c_\gamma\tilde{A}_{1}\le a\le c_\gamma\tilde{A}_{2},
\end{equation}}where $c_\gamma=P_t/(\sigma^{2}_{N}B)$. 
\subsubsection{Capacity for AWGN channels}
For this case, $h(t)$ is selected as $1$ in Eq.(\ref{ndfmUnintChMod}). Then, fast fading (ergodic) conditions can be applied, since the channel has encountered all possible fades during the radar pulse when the radar pulse is long enough. The ergodic channel capacity can be expressed as,
\begin{equation}
\label{CsnrNull}
\begin{split}
C_{rp}(A_{1},A_{2})&=\int_{\frac{c_\gamma {A_{1}}^2}{M}}^{\frac{c_\gamma {A_{2}}^2}{M}}B\ln(1+a)p_{\gamma}(a)da\\
&=\int_{\frac{c_\gamma {A_{1}}^2}{M}}^{\frac{c_\gamma {A_{2}}^2}{M}}B\ln(1+a)\left[\frac{\frac{1}{c_\gamma}e^{-a/c_\gamma}}{(e^{-A_{1}^2/M}-e^{-A_{2}^2/M})}\right]da.
\end{split}
\end{equation}For the numerator integral, let $u=\log_e(1+a)$ and $dv=\exp\left(\frac{-a}{c_\gamma}\right)$. By using integral by parts we get,
\begin{equation}
\label{CsnrNull2}
\begin{split}
C_{rp}=&\frac{B}{\alpha(\tilde{A_1},\tilde{A_2})}\left[-\ln(1+a)e^{\frac{-a}{c_\gamma}}\bigg\rvert_{c_\gamma \tilde{A_1}}^{c_\gamma \tilde{A_2}}-\int_{c_\gamma \tilde{A_1}}^{c_\gamma \tilde{A_2}} \frac{-e^{\frac{-a}{c_\gamma}}}{1+a}da\right]\\
=&\frac{B}{\alpha(\tilde{A_1},\tilde{A_2})}\bigg[\left[\ln(1+c_\gamma \tilde{A_1})e^{-\tilde{A_1}}-\ln(1+c_\gamma\tilde{A_2}) e^{-\tilde{A_2}}\right]\\
&+e^{\frac{1}{c_\gamma}}\int_{\frac{M+c_\gamma {A_{1}}^2}{c_\gamma M}}^{\frac{M+c_\gamma {A_{2}}^2}{c_\gamma M}}\frac{1}{\bar{\gamma}}e^{-\bar{\gamma}}d\bar{\gamma}\bigg]\\
=&\frac{B}{\alpha(\tilde{A_1},\tilde{A_2})}\bigg[\left[\ln(1+c_\gamma \tilde{A_1})e^{-\tilde{A_1}}-\ln(1+c_\gamma \tilde{A_2})e^{-\tilde{A_2}}\right]\\
&+e^{\frac{1}{c_\gamma}}\left[E_1\left(\frac{1}{c_\gamma}+\tilde{A_1}\right)-E_1\left(\frac{1}{c_\gamma}+\tilde{A_2}\right)\right]\bigg]
\end{split}
\end{equation}where $E_1(x)=\int_{x}^{+\infty}\frac{e^{-t}}{t}dt,(x>0)$ is the exponential integral for $n=1$. 
The delta constraint in \cite{RealcMIMO} is given as $A_0/\varDelta\le A_l\le A_0\varDelta$ at the null direction method. Hence, the limits, $A_{1}$ and $A_{2}$, are given as $A_0/\varDelta$ and $A_0\varDelta$. Eq.(\ref{CsnrNull2}) can be rewritten for $A_{1}=A_0/\varDelta=\sqrt{-M\ln(p_0)}/\varDelta$ and $A_{2}=A_0\varDelta=\sqrt{-M\ln(p_0)}\varDelta$,
\begin{equation}
\label{CsnrNullspecial}
\begin{split}
C_{rp}(\frac{A_0}{\varDelta},A_0\varDelta)=&\frac{B}{\alpha(\frac{A_0}{\varDelta},A_0\varDelta)}\bigg[
\bigg(\ln\left(1-\frac{c_\gamma\ln(p_0)}{\varDelta^2}\right)p_0^{-\frac{1}{\varDelta^2}}\\
&-\ln\left(1-c_\gamma\ln(p_0)\varDelta^2\right)p_0^{-\varDelta^2}\bigg)\\
&+e^{\frac{1}{c_\gamma}}\bigg[E_1\left(\frac{1}{c_\gamma}-\frac{\ln(p_0)}{\varDelta^2}\right)-E_1\left(\frac{1}{c_\gamma}+\ln(p_0)\varDelta^2\right)\bigg]\bigg].
\end{split}
\end{equation}If there is no given constraint at the null direction method, the limits, $A_{1}$ and $A_{2}$, can be given as $0$ and $M$. Eq.(\ref{CsnrNull2}) can be rewritten for $A_{1}=0$ and $A_{2}=M$,
\begin{equation}
\label{CsnrNullspecialprop1P1}
\begin{split}
C_{rp}(0,M)=&\frac{B}{1-e^{-M}}\bigg[-\ln\left(1+c_\gamma M\right)e^{-M}\\
&+e^{\frac{1}{c_\gamma}}\bigg[E_1\left(\frac{1}{c_\gamma}\right)-E_1\left(\frac{1}{c_\gamma}+M\right)\bigg]\bigg],
\end{split}
\end{equation}and if $M\gg1$, the capacity can be approximated to the ergodic channel capacity of the Rayleigh fading channel with known CSI as,
\begin{equation}
\label{CsnrNullspecialprop1}
C_{rp}(0,M)\approx C_{RCSI}=-B\exp\left(\frac{1}{c_\gamma}\right)Ei\left(\frac{-1}{c_\gamma}\right),
\end{equation}when $M\gg1$, where $c_\gamma$ is the average SNR which is $c_\gamma=P_t/\sigma^{2}_{N}$ and $Ei(x)=-\int_{-x}^{+\infty}\frac{e^{-t}}{t}dt,(x>0)$. 
If the constraint is given as $A_l\ge A_0$ at the null direction method, the limits, $A_{1}$ and $A_{2}$, becomes $A_0$ and $M$. Eq.(\ref{CsnrNull2}) can be rewritten for $A_{1}=A_0=\sqrt{-M\ln(p_0)}$ and $A_{2}=M$,
\begin{equation}
\label{CsnrNullspecialprop2P1}
\begin{split}
C_{rp}(A_0,M)=&\frac{B}{p_0-e^{-M}}\bigg[\big(\ln\left(1-c_\gamma\ln(p_0)\right)p_0-\ln(1+c_\gamma M)e^{-M}\big)\\ &+e^{\frac{1}{c_\gamma}}\left[E_1\left(\frac{1}{c_\gamma}-\ln(p_0)\right)-E_1\left(\frac{1}{c_\gamma}+M\right)\right]\bigg]
\end{split}
\end{equation}and if $M\gg1$, the capacity can be approximated as,
\begin{equation}
\label{CsnrNullspecialprop2}
\begin{split}
\bar{C}_{rp}(A_0,M)\approx& B\bigg[\ln\left(1-c_\gamma\ln(p_0)\right)\\
&+\left(\frac{1}{p_0}\right)\exp\left(\frac{1}{c_\gamma}\right)E_1\left(\frac{1}{c_\gamma}-\ln(p_0)\right)\bigg],
\end{split}
\end{equation}where $c_\gamma$ is the average SNR \textcolor{blue}{and given as} $P_t/\sigma^{2}_{N}B$.

\subsection{Null Direction Fixed Method Under Fading Channel}
In this part, we consider the capacity of the JRC capable coherent MIMO radar when the null direction fixed method is used for waveform generation, where a single-user transmission is occurred over \textcolor{blue}{Rician and }Rayleigh fading channel. \textcolor{blue}{For directional communications as sidelobe based JRC systems, most of the cases, receiver has a major LoS component along with multi-path reflections. For these cases, the fading channel is assumed as Rician fading channel. However, for some applications, receiver may lost its LoS component due to the mobility. This cases is considered as Rayleigh fading channel in our derivations. Therefore, the fading channel is selected for two different conditions as Rayleigh and Rician fading channels.} \textcolor{blue}{Moreover,} slow and fast fading cases are investigated. Slow fading case occurs when fading channel coherence time $T_c$ is bigger than communication symbol duration $T_s$, i.e. $T_c\gg T_s$. Then, fast fading (ergodic) case is \textcolor{blue}{also} investigated. We assumed that CSI is known by the receiver only and the receiver is able to get the information of the fading states. The main channel model in Eq.(\ref{ytch_null}) can be rewritten as,
\begin{equation}
\label{NonChmodel}
Y_t=H_tG_l(\theta_{c})S_t+N_t=He^{-j\phi_H}A_le^{-j\bar{\phi_1}}S_t+N_t,
\end{equation}where $S_t$ is the channel input, $Y_t$ is the channel output, and $N_t$ independent AWGN random variable with mean zero and variance and $\sigma_N$ at time $t$. $H_t$ is a complex channel gain with amplitude $H$ and phase $\phi_H$ at time t. The phase $\phi_H$ is uniformly distributed in $[0, 2\pi)$, and the signal amplitude $H$ is a random variable with \textcolor{blue}{Rician} and Rayleigh pdf. The null direction fixed method also fixes the phase towards communication direction applying a correction vector. Hence, phase component of the $G_l(\theta_{c})$ becomes constant over the radar pulse as $e^{j\bar{\phi_1}}$. The channel model in Eq.(\ref{NonChmodel}) becomes,
\begin{equation}
\label{NonChmodel2}
Y_t=Ze^{-j\phi_Z}S_t+N_t,
\end{equation}where $Z$ is the product of two random variable $A_l$ and $H$ and the phase $\phi_Z$ is again uniformly distributed in $[0, 2\pi)$. Eq.(\ref{NonChmodel2}) becomes similar to a fading channel model. 
Before attempting to investigate capacity expressions for null direction fixed method under fading channel, pdf of the new random variable $Z$ needs to be introduced. $A_l$, is a random variable with truncated Rayleigh distribution which is given in Eq.(\ref{amppdf}). 

\textcolor{blue}{For Rayleigh fading cases, the pdf of $H$ can be given as,}
\begin{equation}
\label{amppdfH}
p_{H}(h)=\frac{h}{\sigma_H^2}e^{\frac{-h^2}{2\sigma_H^2}},\quad h\ge 0,
\end{equation}where $\sigma_H$ is the variance of the both real and imaginary part of the complex r.v. $H_t$. 
\begin{lemma}
	\label{LemmaDRayleigh}
	Any random variable $X$, which is a product of two random variable, i.e. $X=X_1X_2$, which follows double truncated Rayleigh ($X_1$) and Rayleigh ($X_2$) distribution, presents a distribution with the following probability density function (pdf), 
	\begin{equation}
	\label{pdftheorem}
	p_{X}(x)=\frac{x\left[\Gamma\left(0,\beta_1;\frac{x^2}{4{\sigma_X}^2}\right)-\Gamma\left(0,\beta_2;\frac{x^2}{4{\sigma_X}^2}\right)\right]}{\alpha(\beta_1,\beta_2)2{\sigma_X}^2},\quad x\ge 0,  
	\end{equation}where $\beta_1=\frac{(X^{min}_1)^2}{2\sigma_{X_1}^2}$, $\beta_2=\frac{(X^{max}_1)^2}{2\sigma_{X_1}^2}$ and $\sigma_{X}$ is equal to $\sigma_{X_1}\sigma_{X_2}$ and $\alpha(a,b)=e^{-a}-e^{-b}$. $\sigma_{X_1}$ and $\sigma_{X_2}$ are the variances of the both real and the imaginary part of the complex random variables $X_1$ and $X_2$, respectively. $X^{min}_1$ and $X^{max}_1$ are the lower and the higher limits of the double truncated Rayleigh distributed r.v. $X_1$, i.e. $X^{min}_1\le X_1 \le X^{max}_1$. $\Gamma(a,x;b)$ is the generalized incomplete gamma function \cite{GenIncGamma} as $\Gamma(a,x;b)=\int_{x}^{\infty}t^{a-1}e^{-t-bt^{-1}}dt$. 
\end{lemma}
\begin{proof}
	See the \ref{ProofProductDistTrunRay}.
\end{proof}Using the Lemma.\ref{LemmaDRayleigh}, pdf of the $Z$ in Eq.(\ref{NonChmodel2}) can be rewritten as,
\begin{equation}
p_{Z}(z)=\frac{z\left[\Gamma\left(0,\tilde{A_1};\frac{z^2}{M2\sigma_H^2}\right)-\Gamma\left(0,\tilde{A_2};\frac{z^2}{M2\sigma_H^2}\right)\right]}{\alpha(\tilde{A_1},\tilde{A_2})M\sigma_H^2},\quad z\ge 0, 
\end{equation}where $\tilde{A_1}=A_{1}^2/M$, $\tilde{A_2}=A_{2}^2/M$. By the transformation theorem for single random variables, the pdf of the channel power gain $Z^2$ becomes,
\begin{equation}
\label{snrpdf1F}
p_{Z^2}(\bar{z})=\frac{\left[\Gamma\left(0,\tilde{A_1};\frac{\bar{z}}{2M\sigma_H^2}\right)-\Gamma\left(0,\tilde{A_2};\frac{\bar{z}}{M2\sigma_H^2}\right)\right]}{\alpha(\tilde{A_1},\tilde{A_2})2M\sigma_H^2},\quad \bar{z}\ge 0. 
\end{equation}Using the (\ref{snrpdf1F}), pdf of the SNR, $\gamma$ becomes,
\begin{equation}
\label{snrpdf2F}
p_{\gamma}(a)=\frac{\left[\Gamma\left(0,\tilde{A_1};\frac{a}{c_\gamma2\sigma_H^2}\right)-\Gamma\left(0,\tilde{A_2};\frac{a}{c_\gamma2\sigma_H^2}\right)\right]}{\alpha(\tilde{A_1},\tilde{A_2})c_\gamma 2\sigma_H^2},\quad a\ge 0,
\end{equation}where $c_\gamma=P_t/(\sigma^{2}_{N}B)$. We have also investigated the pdf of the product of single side truncated Rayleigh and Rayleigh random variables.
\begin{corollary}
	Any random variable $X$, which is a product of two r.v., $X=X_1X_2$, which follows single side truncated Rayleigh ($X_1$), $X_1\ge X^{min}_1$, and Rayleigh ($X_2$) distribution, presents a distribution with the following probability density function (pdf),
	\begin{equation}
	\label{singleRaypdf}
	p_{X}(x)=\frac{x}{e^{-\beta_1}{\sigma_X}^2}\Gamma\left(0,\beta_1;\frac{x^2}{4{\sigma_X}^2}\right),\quad x\ge 0,  
	\end{equation}
\end{corollary}
\begin{proof}
	By replacing $\beta_2$ with $\infty$ in Eq.(\ref{pdftheorem}), $\alpha(\beta_1,\infty)$ and $\Gamma(a,\infty;b)$ becomes $e^{-\beta_1}$ and $0$, respectively. Then, related distribution pdf can be simply defined as in Eq.(\ref{singleRaypdf}).
\end{proof}Using the Eq.(\ref{singleRaypdf}), pdf of the $Z$ for single side truncated case in Eq.(\ref{NonChmodel2}) can be rewritten as,
\begin{equation}
p_{Z}(z)=\frac{z}{e^{\frac{A_0^2}{M}}M\sigma_H^2}\Gamma\left(0,\frac{A_0^2}{M};\frac{z^2}{M2\sigma_H^2}\right),\quad z\ge 0,
\end{equation}by doing the same transformation step in Eq.(\ref{snrpdf1F},\ref{snrpdf2F}) we get,
\begin{equation}
\label{snrpdf3F}
p_{\gamma}(a)=\frac{1}{e^{\frac{A_0^2}{M}}c_\gamma 2\sigma_H^2}\Gamma\left(0,\frac{A_0^2}{M};\frac{a}{c_\gamma2\sigma_H^2}\right),\quad a\ge 0.
\end{equation}Without constraint in null direction fixed method, random variable $Z$ becomes product of two Rayleigh variable and channel becomes double Rayleigh fading channel. Integral in Eq.(\ref{mellinconvint1}) with limits $\beta_1=0$ and $\beta_2=\infty$, can be expressed from \cite[p.~370]{BookTablesMath} as,
\begin{equation}
\label{pdfDoubleRay}
p_{Z}(z)=\frac{2z}{M\sigma_H^2}K_0\left(\frac{2z}{\sqrt{M2\sigma_H^2}}\right),\quad z\ge 0,
\end{equation}where $K_0(.)$ is the zeroth order modified Bessel function of the second kind. Then, pdf of the SNR becomes,
\begin{equation}
\label{pdfDoubleRaysnr}
p_{\gamma}(a)=\frac{1}{c_\gamma\sigma_H^2}K_0\left(\sqrt{\frac{2a}{c_\gamma\sigma_H^2}}\right),\quad a\ge 0.
\end{equation}

\textcolor{blue}{The same derivations must be done for Rician fading cases, the pdf of $H$ can be given as,}
\textcolor{blue}{
\begin{equation}
\label{amppdfH2}
p_{H}(h)=\frac{h}{\sigma_H^2}e^{\frac{-(h^2+\mu^2)}{2\sigma_H^2}}I_0\left(\frac{h\mu}{\sigma_H^2}\right),\quad h\ge 0,
\end{equation}where $I_0(.)$ is the modified Bessel function of the first kind with order zero and $\mu$ is the non-centrality parameter. The shape parameter $K$ can be given as, $\frac{\mu^2}{2\sigma_H^2}$ and the scale parameter $\Omega$ is written as $\mu^2+2\sigma_H^2$, which is defined as the total power received from multi-paths.}
\begin{lemma}
	\label{LemmaDRayleigh2}
	\textcolor{blue}{Any random variable $X$, which is a product of two random variable, i.e. $X=X_1X_2$, which follows double truncated Rayleigh ($X_1$) and Rician ($X_2$) distribution, presents a distribution with the following probability density function (pdf), }
	\textcolor{blue}{
	\begin{equation}
	\label{pdftheorem2}
	p_{X}(x)=\frac{xe^{\frac{-\mu^2}{2\sigma_{X_2}^2}}}{2{\sigma_X}^2}\sum_{i=0}^{\infty}\frac{\Gamma\left(-i,\beta_1;\frac{x^2}{4{\sigma_X}^2}\right)-\Gamma\left(-i,\beta_2;\frac{x^2}{4{\sigma_X}^2}\right)}{\alpha(\beta_1,\beta_2)i!i!(2\sigma_{X_2}^2\sigma_{X}^2)^{i}\left(\frac{x\mu}{2}\right)^{-2i}},\quad x\ge 0,  
	\end{equation}where $\beta_1=\frac{(X^{min}_1)^2}{2\sigma_{X_1}^2}$, $\beta_2=\frac{(X^{max}_1)^2}{2\sigma_{X_1}^2}$ and $\sigma_{X}$ is equal to $\sigma_{X_1}\sigma_{X_2}$ and $\alpha(a,b)=e^{-a}-e^{-b}$. $\sigma_{X_1}$ and $\sigma_{X_2}$ are the variances of the both real and the imaginary part of the complex random variables $X_1$ and $X_2$, respectively. $X^{min}_1$ and $X^{max}_1$ are the lower and the higher limits of the double truncated Rayleigh distributed r.v. $X_1$, i.e. $X^{min}_1\le X_1 \le X^{max}_1$. $\Gamma(a,x;b)$ is the generalized incomplete gamma function \cite{GenIncGamma} as $\Gamma(a,x;b)=\int_{x}^{\infty}t^{a-1}e^{-t-bt^{-1}}dt$.}
\end{lemma}
\begin{proof}
	\textcolor{blue}{See the \ref{ProofProductDistTrunRayAndRice}.}
\end{proof}\textcolor{blue}{Using the Lemma.\ref{LemmaDRayleigh2}, pdf of the $Z$ in Eq.(\ref{NonChmodel2}) can be rewritten as,}
\textcolor{blue}{
\begin{equation}
p_{Z}(z)=\frac{ze^{\frac{-\mu^2}{2\sigma_{H}^2}}}{M{\sigma_H}^2}\sum_{i=0}^{\infty}\frac{\Gamma\left(-i,\tilde{A_1};\kappa\right)-\Gamma\left(-i,\tilde{A_2};\kappa\right)}{\alpha(\tilde{A_1},\tilde{A_2})i!i!(M\sigma_{H}^4)^{i}\left(\frac{z\mu}{2}\right)^{-2i}},\quad z\ge 0,
\end{equation}where $\kappa=\frac{z^2}{2M{\sigma_H}^2}$. By the transformation theorem for single random variables, the pdf of the channel power gain $Z^2$ becomes,}
\textcolor{blue}{
\begin{equation}
\label{snrpdf1F2}
p_{Z^2}(\bar{z})=\frac{e^{\frac{-\mu^2}{2\sigma_{H}^2}}}{2M{\sigma_H}^2}\sum_{i=0}^{\infty}\frac{\Gamma\left(-i,\tilde{A_1};\bar{\kappa}\right)-\Gamma\left(-i,\tilde{A_2};\bar{\kappa}\right)}{\alpha(\tilde{A_1},\tilde{A_2})i!i!\left(\frac{\mu}{2\sigma_{H}^2}\right)^{-2i}\left(\frac{\bar{z}}{M}\right)^{-i}},\quad \bar{z}\ge 0.
\end{equation}Using the (\ref{snrpdf1F2}), pdf of the SNR, $\gamma$ becomes,}
\textcolor{blue}{
\begin{equation}
\label{snrpdf2F2}
p_{\gamma}(a)=\frac{e^{\frac{-\mu^2}{2\sigma_{H}^2}}}{c_\gamma 2{\sigma_H}^2}\sum_{i=0}^{\infty}\frac{\Gamma\left(-i,\tilde{A_1};\chi\right)-\Gamma\left(-i,\tilde{A_2};\chi\right)}{\alpha(\tilde{A_1},\tilde{A_2})i!i!\left(\frac{\mu}{2\sigma_{H}^2}\right)^{-2i}\left(\frac{a}{c_\gamma}\right)^{-i}},\quad a\ge 0,
\end{equation}where $c_\gamma=P_t/(\sigma^{2}_{N}B)$ and $\chi=\frac{a}{c_\gamma 2{\sigma_H}^2}$. We have also investigated the pdf of the product of single side truncated Rayleigh and Rayleigh random variables.}
\begin{corollary}
	\textcolor{blue}{Any random variable $X$, which is a product of two r.v., $X=X_1X_2$, which follows single side truncated Rayleigh ($X_1$), $X_1\ge X^{min}_1$, and Rayleigh ($X_2$) distribution, presents a distribution with the following probability density function (pdf),}
	\textcolor{blue}{
	\begin{equation}
	\label{singleRaypdf2}
	p_{X}(x)=\frac{xe^{\frac{-\mu^2}{2\sigma_{X_2}^2}}}{e^{-\beta_1}2{\sigma_X}^2}\sum_{i=0}^{\infty}\frac{\Gamma\left(-i,\beta_1;\frac{x^2}{4{\sigma_X}^2}\right)}{i!i!(2\sigma_{X_2}^2\sigma_{X}^2)^{i}}\left(\frac{x\mu}{2}\right)^{2i},\quad x\ge 0,   
	\end{equation}}
\end{corollary}
\begin{proof}
	\textcolor{blue}{By replacing $\beta_2$ with $\infty$ in Eq.(\ref{pdftheorem2}), $\alpha(\beta_1,\infty)$ and $\Gamma(a,\infty;b)$ becomes $e^{-\beta_1}$ and $0$, respectively. Then, related distribution pdf can be simply defined as in Eq.(\ref{singleRaypdf2}).}
\end{proof}\textcolor{blue}{Using the Eq.(\ref{singleRaypdf2}), pdf of the $Z$ for single side truncated case in Eq.(\ref{NonChmodel2}) can be rewritten as,}
\textcolor{blue}{
\begin{equation}
p_{Z}(z)=\frac{ze^{\frac{-\mu^2}{2\sigma_{H}^2}}}{e^{-\tilde{A_1}}M{\sigma_H}^2}\sum_{i=0}^{\infty}\frac{\Gamma\left(-i,\tilde{A_1};\frac{z^2}{2M{\sigma_H}^2}\right)}{i!i!(M\sigma_{H}^4)^{i}}\left(\frac{z\mu}{2}\right)^{2i},\quad z\ge 0,
\end{equation}by doing the same transformation step in Eq.(\ref{snrpdf1F},\ref{snrpdf2F}) we get,}
\textcolor{blue}{
\begin{equation}
\label{snrpdf3F2}
p_{\gamma}(a)=\frac{e^{\frac{-\mu^2}{2\sigma_{H}^2}}}{e^{-\tilde{A_1}}c_\gamma 2{\sigma_H}^2}\sum_{i=0}^{\infty}\frac{\Gamma\left(-i,\tilde{A_1};\frac{a}{c_\gamma 2{\sigma_H}^2}\right)}{i!i!\left(\frac{\mu}{2\sigma_{H}^2}\right)^{-2i}\left(\frac{a}{c_\gamma}\right)^{-i}},\quad a\ge 0,
\end{equation}Without constraint in null direction fixed method, random variable $Z$ becomes product of Rayleigh and Rician variable. Integral in Eq.(\ref{mellinconvint1}) with limits $\beta_1=0$ and $\beta_2=\infty$, can be expressed from \cite[p.~58]{BookProbGaussian} as,}
\textcolor{blue}{\begin{equation}
\label{pdfDoubleRay2}
p_{Z}(z)=\frac{2ze^{\frac{-\mu^2}{2\sigma_{H}^2}}}{M\sigma_{H}^2}\sum_{i=0}^{\infty}\frac{\left(\frac{\mu\sqrt{z}}{2\sigma_{H}^2}\right)^{2i}\left(\frac{\sigma_{H}}{\sqrt{M/2}}\right)^{i}}{i!i!}K_i\left(\frac{2z}{\sqrt{2M\sigma_{H}^2}}\right),\quad z\ge 0,
\end{equation}where $K_i(.)$ is the i-th order modified Bessel function of the second kind. Then, pdf of the SNR becomes,}
\textcolor{blue}{
\begin{equation}
\label{pdfDoubleRaysnr2}
p_{\gamma}(a)=\frac{e^{\frac{-\mu^2}{2\sigma_{H}^2}}}{c_\gamma\sigma_{H}^2}\sum_{i=0}^{\infty}\frac{\left(\frac{\mu a^{1/4}}{2\sigma_{H}^2}\right)^{2i}\left(\frac{\sigma_{H}}{\sqrt{c_\gamma/2}}\right)^{i}}{i!i!}K_i\left(\sqrt{\frac{2a}{c_\gamma\sigma_{H}^2}}\right),\quad a\ge 0,
\end{equation}}

\subsubsection{Capacity Under Slow fading, $T_c\gg T_s$}
For the slow fading case, communication symbol duration ($T_s$) is smaller than the fading channel coherence time ($T_c$). For slowly varying channels where the instantaneous SNR $\gamma_{i}$ is assumed to be constant for a large number of symbols, outage capacity needs to be presented. Over a big time period, channel data rate can reach up to capacity without fading case with negligible error. However, since the transmitter does not know the instantaneous SNR, a constant transmission data rate required, and it is independent from the instantaneous received SNR. 

Particularly, a design parameter $p_{out}$ has to be selected which expresses a probability that the system can be in outage. In other words, $p_{out}$ is the probability that the system cannot correctly decrypt the transmitted communication symbols. The cause of the outage is the slow fading conditions. Hence, it is independent from unintentional modulation parameter $A_l$. Then, $p_{out}$ can be defined as $Pr[\gamma_i<\gamma_{min}|A_l^2]$, where $\gamma_{min}$ is the minimum SNR level that the received symbols cannot be correctly decoded with probability $1$. At this level, the system declares an outage. The pdf of the $\gamma_{i}|A_l^2$ is,
\begin{equation}
\label{snrpdfinstan}
p_{\gamma_i|A_l^2}(a)=\frac{1}{\bar{\gamma}A_l^2}e^{-a/\bar{\gamma}A_l^2},\quad a>0,
\end{equation}where $\bar{\gamma}$ \textcolor{blue}{is} $(E[|h|^2]c_\gamma)/M$, and $E[X]$  shows the expected value of a random variable \textcolor{blue}{and given as $\int_{-\infty}^{\infty}xp_X(x)dx$.} $h$ is the instantaneous channel gain resulting from fading \textcolor{blue}{and} average channel power gain for Rayleigh fading \textcolor{blue}{$E[|h|^2]$ is} $2\sigma_H^2$ \textcolor{blue}{and for Rician fading, $E[|h|^2]=\mu^2+2\sigma_H^2$, where $\mu$ is the non-centrality parameter.} Then, $p_{out}$ can be defined as,
\begin{equation}
\label{poutage}
p_{out}=Pr[\gamma_i<\gamma_{min}|A_l]=1-\int_{\gamma_{min}}^{\infty}p_{\gamma_i|A_l}(a)da=1-e^{-\frac{\gamma_{min}}{\bar{\gamma}A_l^2}},
\end{equation}Taking out $\gamma_{min}$ from Eq.(\ref{poutage}),
\textcolor{blue}{
\begin{equation}
\label{snrminoutage}
\gamma_{min}=-\bar{\gamma}\ln(1-p_{out})A_l^2=-\frac{E[|h|^2]c_\gamma}{M}\ln(1-p_{out})A_l^2,
\end{equation}}Then, the capacity with outage expression can be written as,
\begin{equation}
\label{CapacitySlow1}
C^{out}_{sf} = \int B\ln(1+a)p_{\gamma_{min}}(a)da,
\end{equation}where the pdf of the $\gamma_{min}$ as,
\begin{equation}
\label{snrpdf2DF}
p_{\gamma_{min}}(a)=\frac{e^{-a/\tilde{\gamma}}}{\alpha(\tilde{A_1},\tilde{A_2})\tilde{\gamma}},\quad \tilde{\gamma}\tilde{A_1}\le a\le \tilde{\gamma}\tilde{A_2},
\end{equation}where \textcolor{blue}{$\tilde{\gamma}=-E[|h|^2]c_\gamma\ln(1-p_{out})$.} Also, average outage rate \cite{ergodicOutRate} can be defined as, 
\begin{equation}
\label{OutRateSlow1}
R^{out}_{sf} = (1-p_{out})C^{out}_{sf}.
\end{equation}Then, the ergodic capacity for coherent JRC system under slow fading can be expressed by replacing $c_\gamma$ with $\tilde{\gamma}$ in Eq.(\ref{CsnrNull}),
\begin{equation}
\label{CsnrNull2DF}
\begin{split}
C^{out}_{sf}=&\frac{B}{\alpha(\tilde{A_1},\tilde{A_2})}\bigg[\left[\ln(1+\tilde{\gamma}\tilde{A_1})e^{-\tilde{A_1}} -\ln(1+\tilde{\gamma}\tilde{A_2})e^{-\tilde{A_2}}\right]\\
&+e^{\frac{1}{\tilde{\gamma}}}\left[E_1\left(\frac{1}{\tilde{\gamma}}+\tilde{A_1}\right)-E_1\left(\frac{1}{\tilde{\gamma}}+\tilde{A_2}\right)\right]\bigg].
\end{split}
\end{equation}Capacity expressions in Eq.(\ref{CsnrNullspecial}) and Eq.(\ref{CsnrNullspecialprop1P1}) for coherent JRC methods with different constraints can be easily modified by replacing $c_\gamma$ with $\tilde{\gamma}$.

\begin{figure*}[htb]
	\centering
	\subfigure[]
	{
		\includegraphics[width=0.3\textwidth]{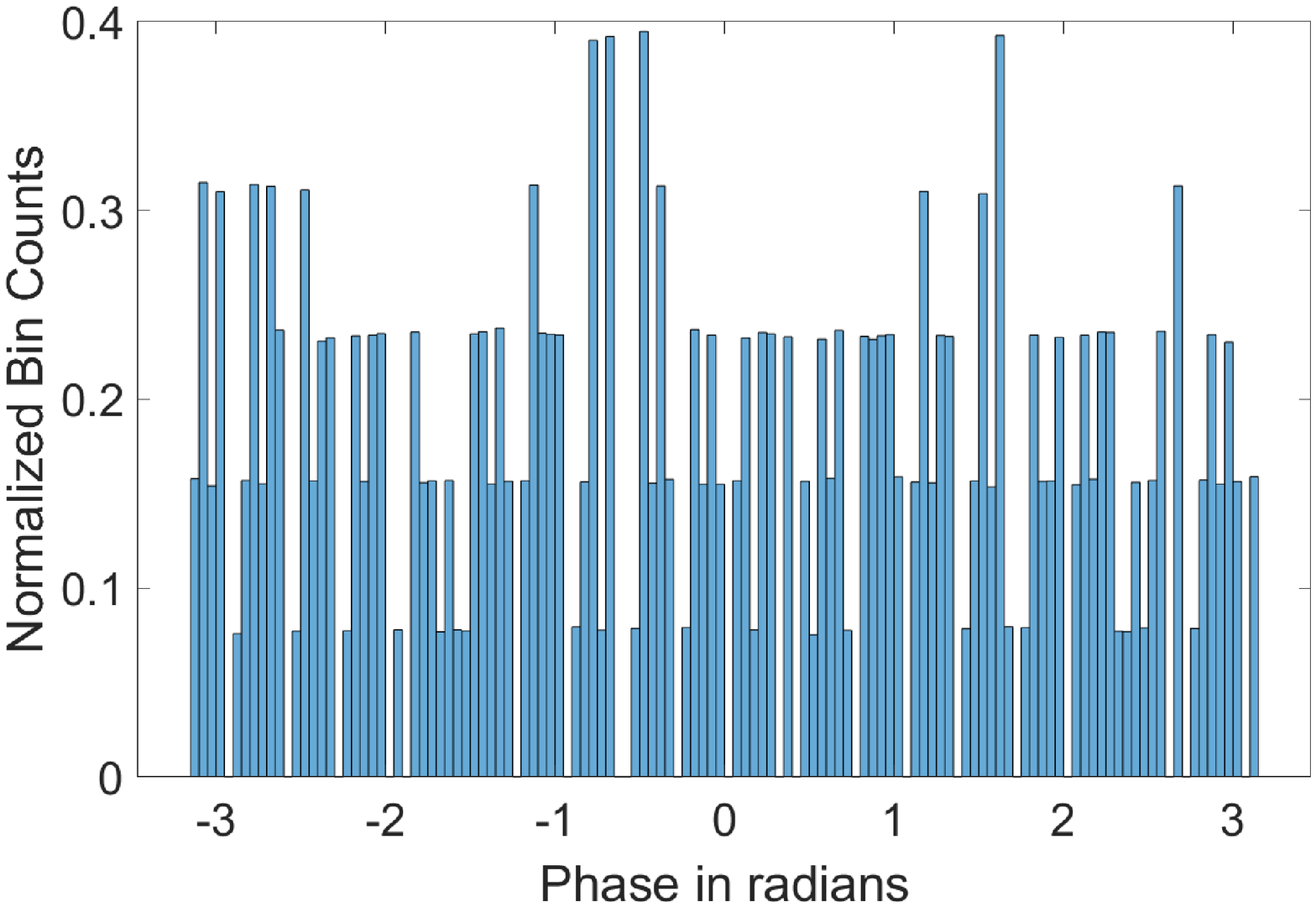}
		\label{fig:PhDistM16}
	}
	\subfigure[]
	{
		\includegraphics[width=0.3\textwidth]{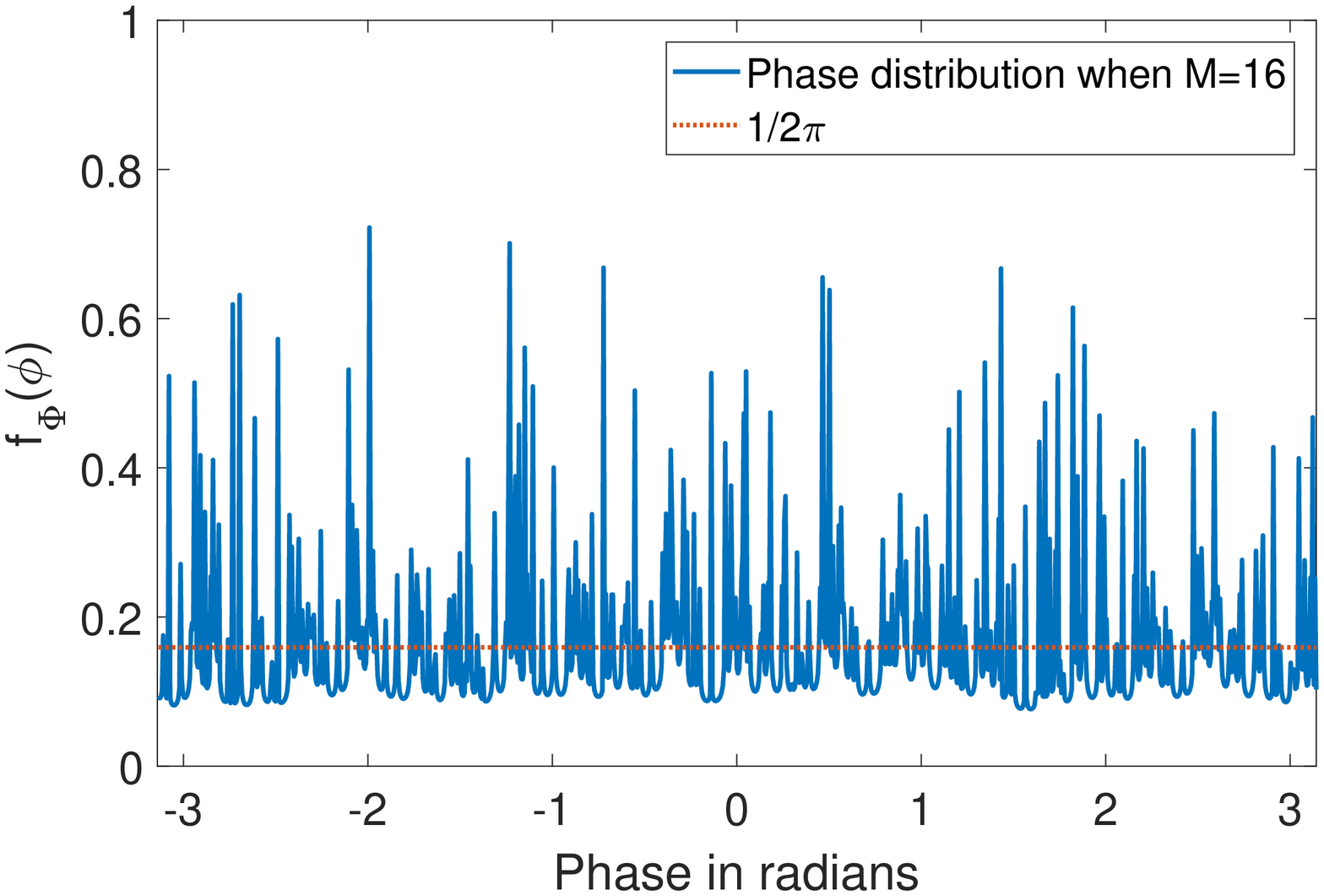}
		\label{fig:PhDistM16f}
	}
	\subfigure[]
	{
		\includegraphics[width=0.3\textwidth]{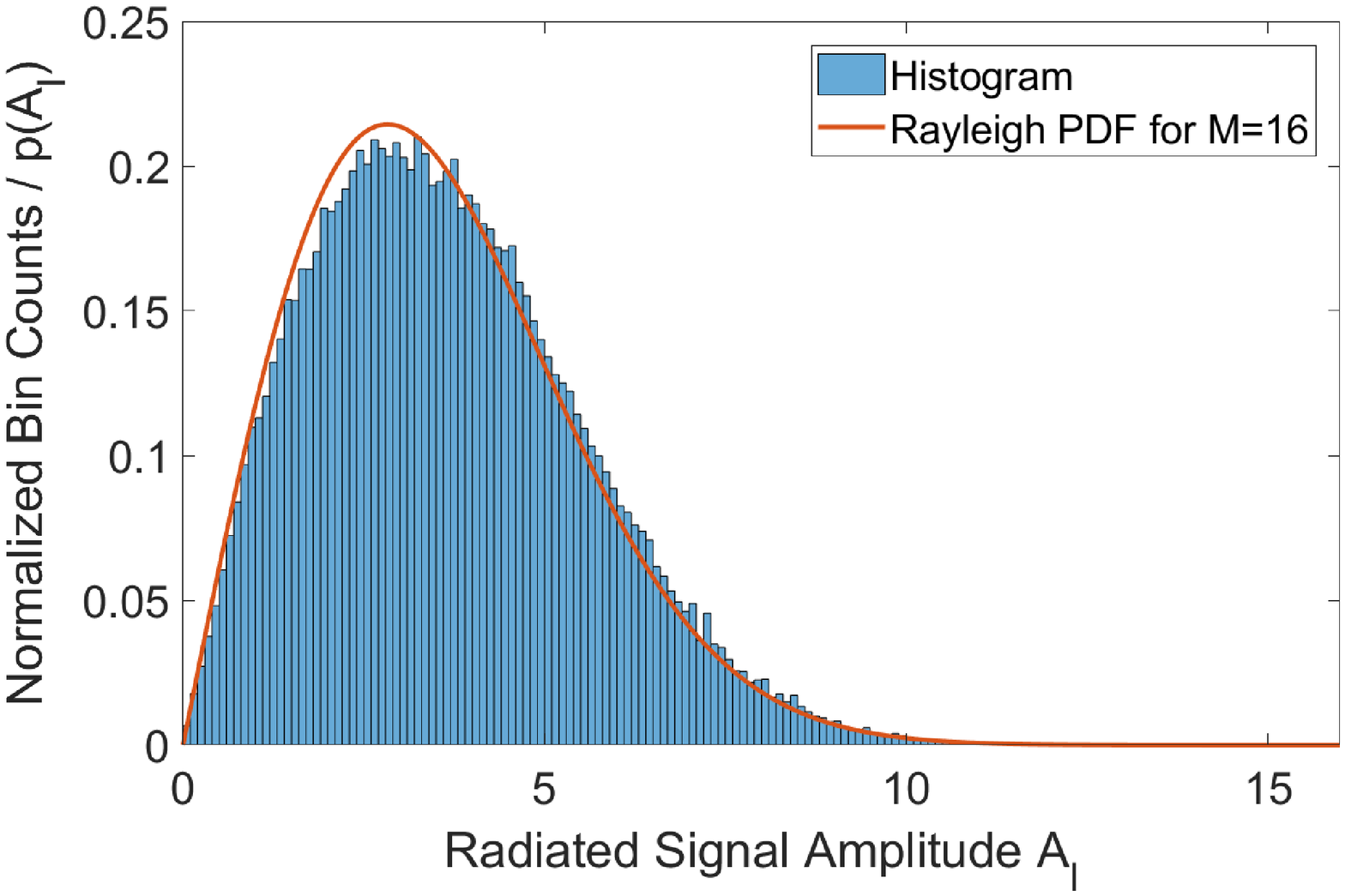}
		\label{fig:AmpDistM16}
	}
	\subfigure[]
	{
		\includegraphics[width=0.3\textwidth]{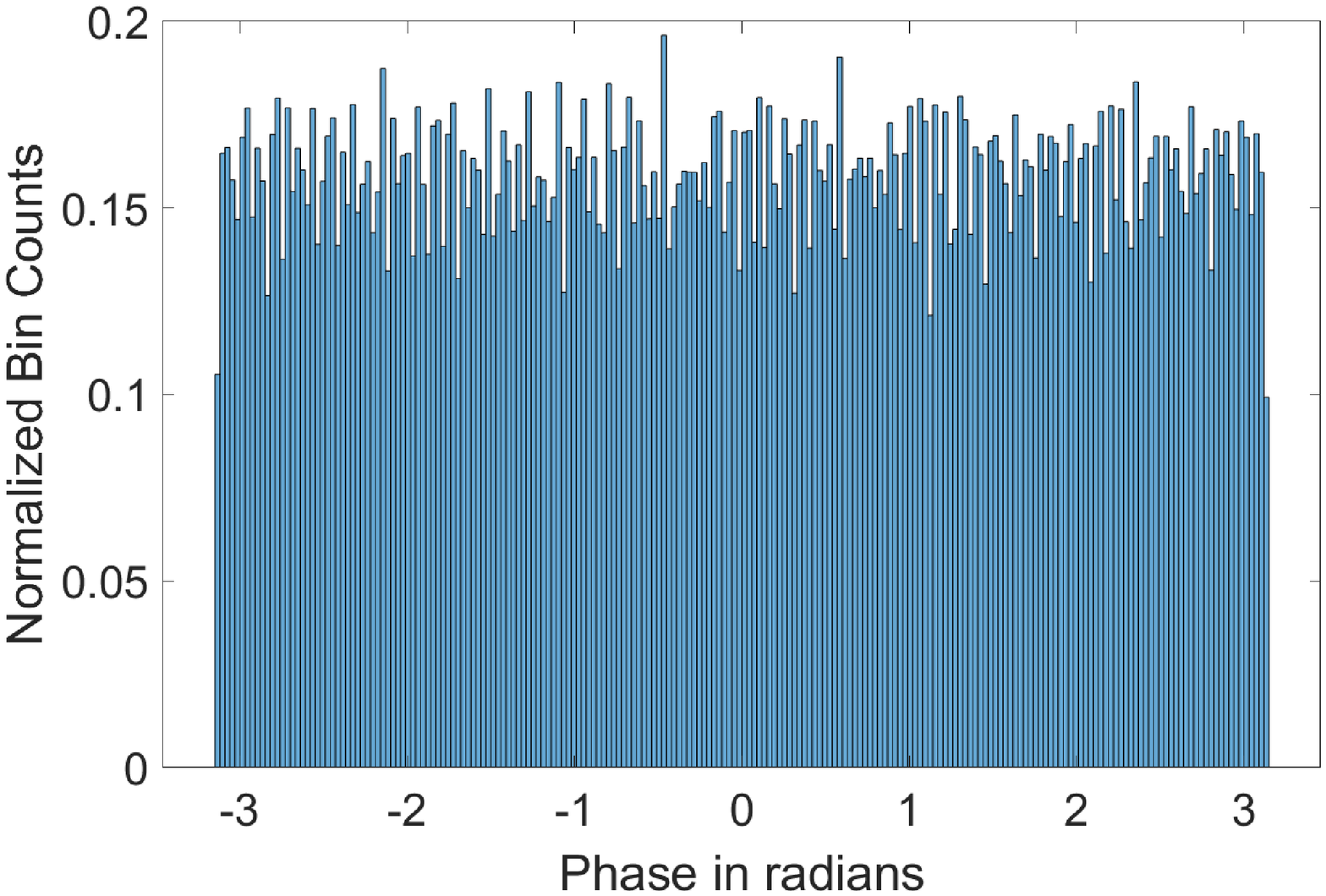}
		\label{fig:PhDistM100}
	}
	\subfigure[]
	{
		\includegraphics[width=0.3\textwidth]{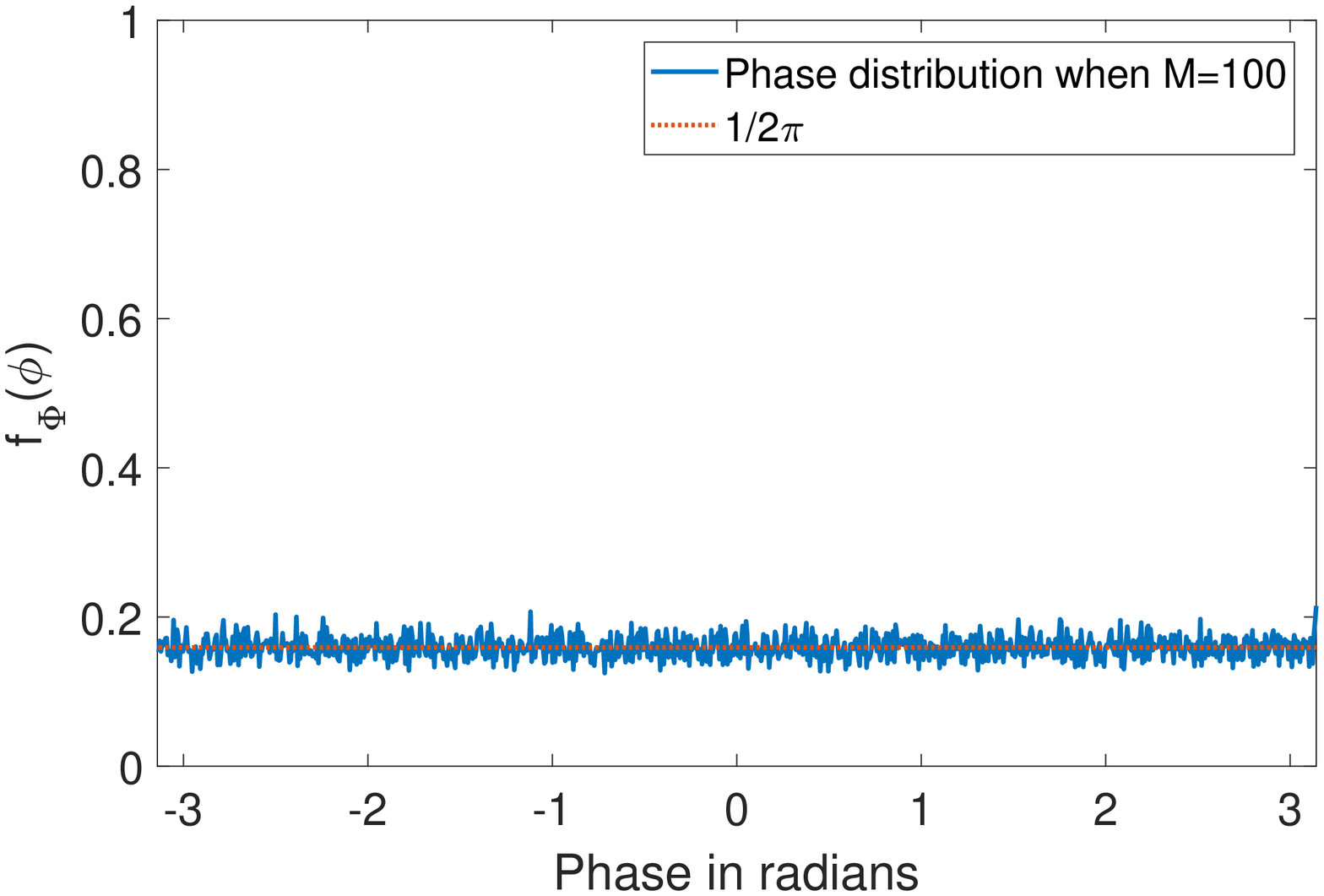}
		\label{fig:PhDistM100f}
	}
	\subfigure[]
	{
		\includegraphics[width=0.3\textwidth]{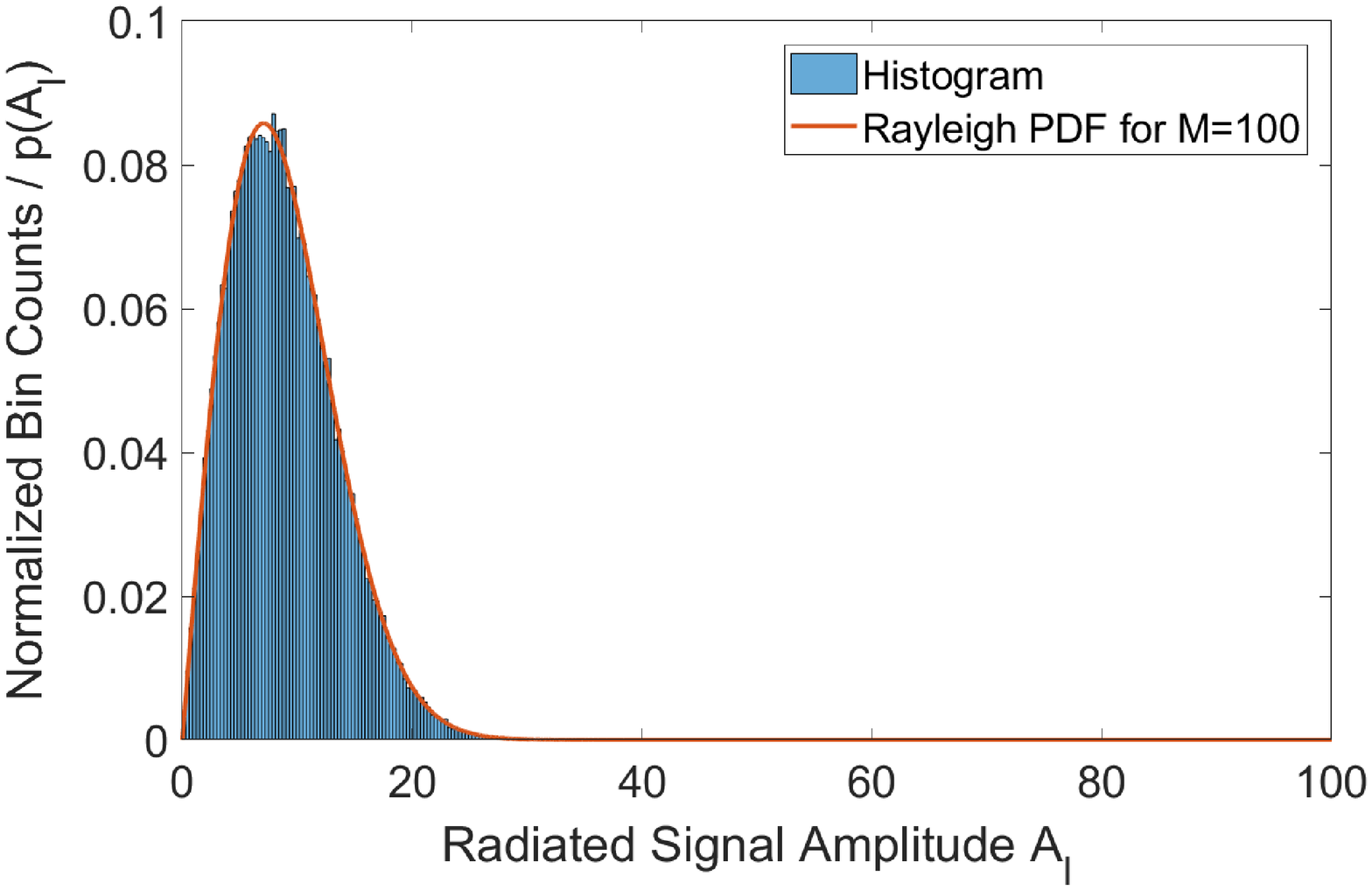}
		\label{fig:AmpDistM100}
	}	
	\caption{(a,d) Histogram, (b,e) probability density function in Eq.(\ref{pdfDelta11}) of phase distribution of all antennas, $\phi_m(l)$, and (c,f) histogram of the amplitude distribution of radiated signal $A_l$, for $M=16$ and $M=100$ and $l=1,2,...,L$, in the MIMO array for communication angle is selected as $\theta_{c}=-22$ degrees and Rayleigh pdf in Theorem.(\ref{theoremRay}) curve (c) for $M=16$, (f) for $M=100$.}
	\label{fig:resultsPhaseDistM16}
\end{figure*}

\subsubsection{Capacity Under Fast fading, $T_c\le T_s$}
For the fast fading case, communication symbol duration ($T_s$) is larger than the channel coherence time ($T_c$). The ergodic channel capacity can be reached \textcolor{blue}{for Rayleigh and Rician fading cases and different limit conditions using the related pdf expressions}, after the channel has encountered all possible fades and it can be expressed as,
\begin{equation}
\label{CtdrayFast}
\bar{C}_{rf}=\int_{0}^{\infty}B\ln(1+a)p_\gamma(a)da, \quad\quad a\ge 0
\end{equation}
Capacity expressions for the different communication methods can be easily obtained via modifying Eq.(\ref{CtdrayFast}) for $A_1$ and $A_2$ constraints. Moreover, \textcolor{blue}{the expression} in Eq.(\ref{CtdrayFast}) can be solved numerically \textcolor{blue}{ for all cases. With only exception,} capacity expression for null direction fixed method without constraint, i.e. double Rayleigh fading case, can be given using the pdf in Eq.(\ref{pdfDoubleRaysnr}) in a closed form equation as,
\begin{equation}
\label{CtdrayFast2}
\bar{C}_{rf}=\frac{1}{c_\gamma 2\sigma_H^2}\MeijerG[\Bigg]{3}{1}{1}{3}{-1}{-1,-1,0}{\frac{1}{c_\gamma 2\sigma_H^2}}
\end{equation}where $G_{m,n}^{p,q}$ is the Meijer G-function defined using a contour integral as,
\begin{equation}
\label{meijerGFunc}
\begin{split}
&\MeijerG[\Bigg]{m}{n}{p}{q}{a_1,\dots,a_p}{b_1,\dots,b_q}{z}=\frac{1}{j2\pi}\\
&\quad\quad\times\int\limits_{\mathrm{C}}\frac{\prod_{i=1}^{m}\Gamma(b_i+s)\prod_{i=1}^{n}\Gamma(1-a_i-s)}{\prod_{i=n+1}^{p}\Gamma(a_i+s)\prod_{i=m+1}^{q}\Gamma(1-b_i-s)}z^{-s}ds,
\end{split}
\end{equation}where $j=\sqrt{-1}$ and $\Gamma(.)$ is the gamma function \cite{BookHandMath}.

\section{Numerical Analysis}
\label{NumAnalysis}
In this section, first we investigate the amplitude distribution of the unintentional modulation given in Sec.\ref{AmpDistribution}. Then, we display and compare the results of the channel capacity for the JRC capable coherent MIMO radars when the actual channel distribution is AWGN or Rayleigh\textcolor{blue}{Rician}.

\subsection{Unintentional Modulation Amplitude Distribution Analysis}
\label{UnIntModAnalysis}
In Sec.(\ref{AmpDistribution}), we have proved that the probability distribution of the amplitude distribution of the unintentional modulation follows a Rayleigh distribution. In this proof, the main requirement to reach this distribution indicates that the phase distribution along the radar pulse must converge to uniform distribution when the number of antenna elements has relatively large values.

\begin{figure}[htb]
	\centering
	\subfigure[]
	{
		\includegraphics[width=0.22\textwidth]{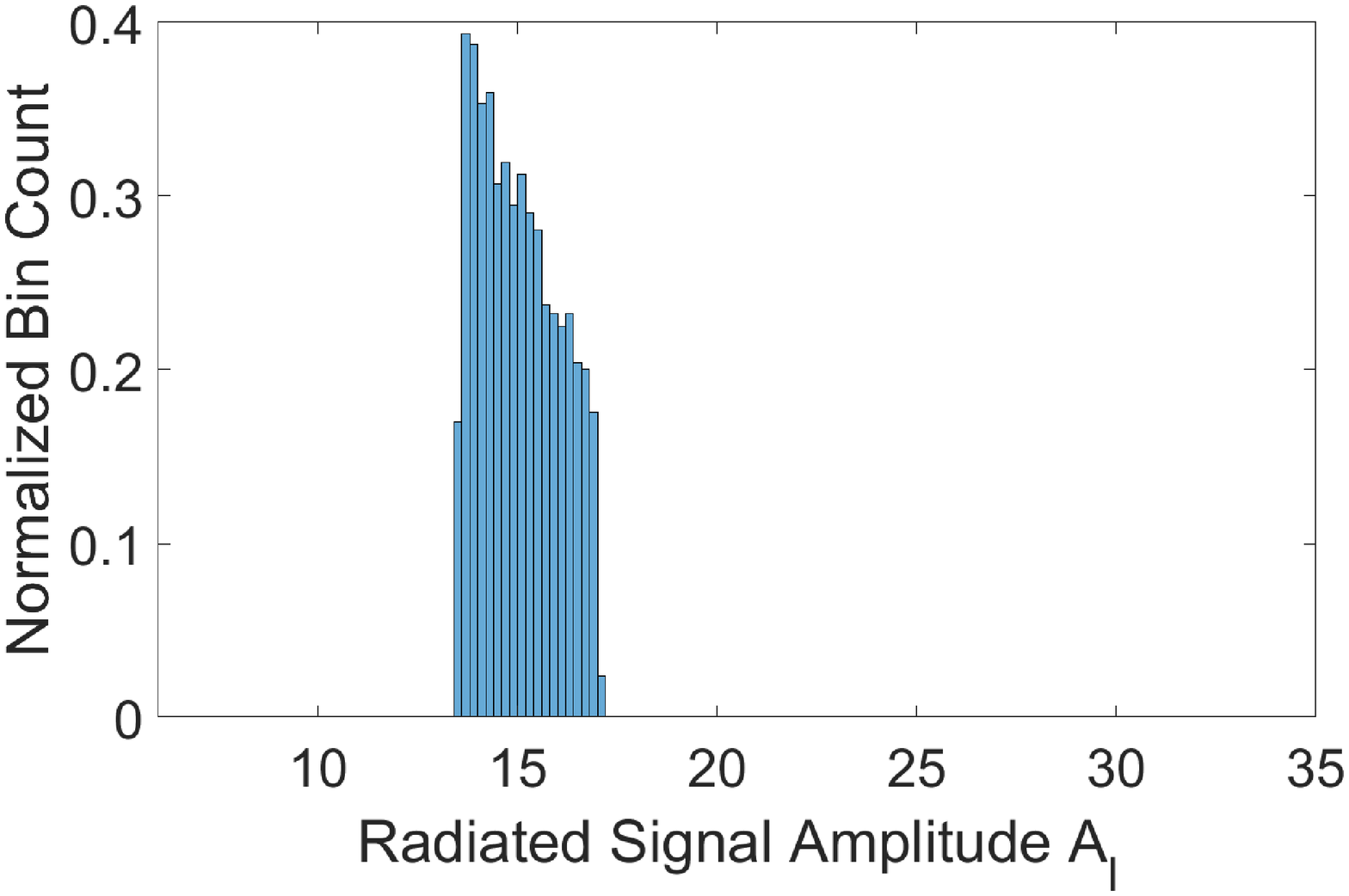}
		\label{fig:PhDistM100delta1dB}
	}
	\subfigure[]
	{
		\includegraphics[width=0.22\textwidth]{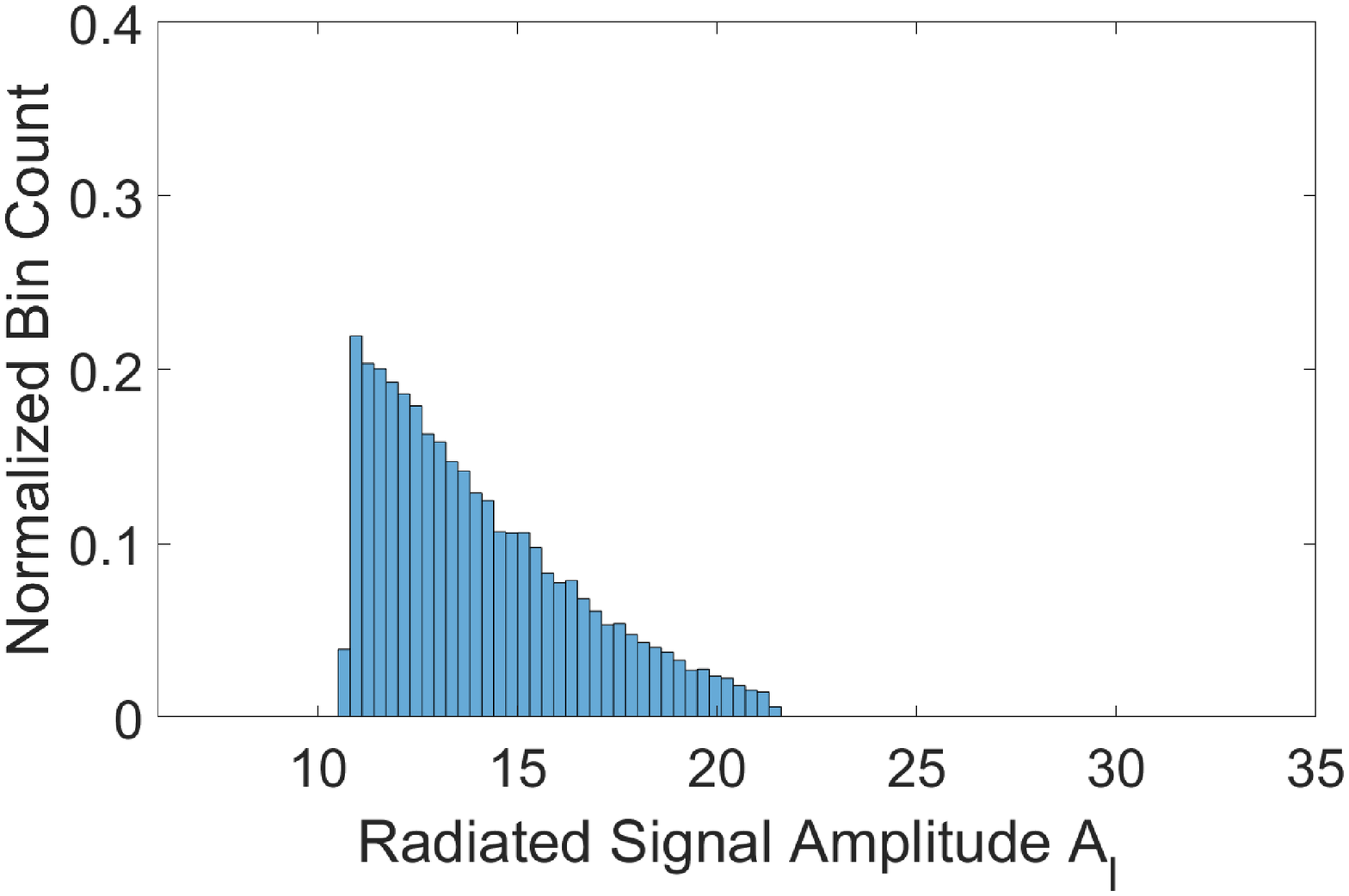}
		\label{fig:PhDistM100delta3dB}
	}
	\subfigure[]
	{
		\includegraphics[width=0.22\textwidth]{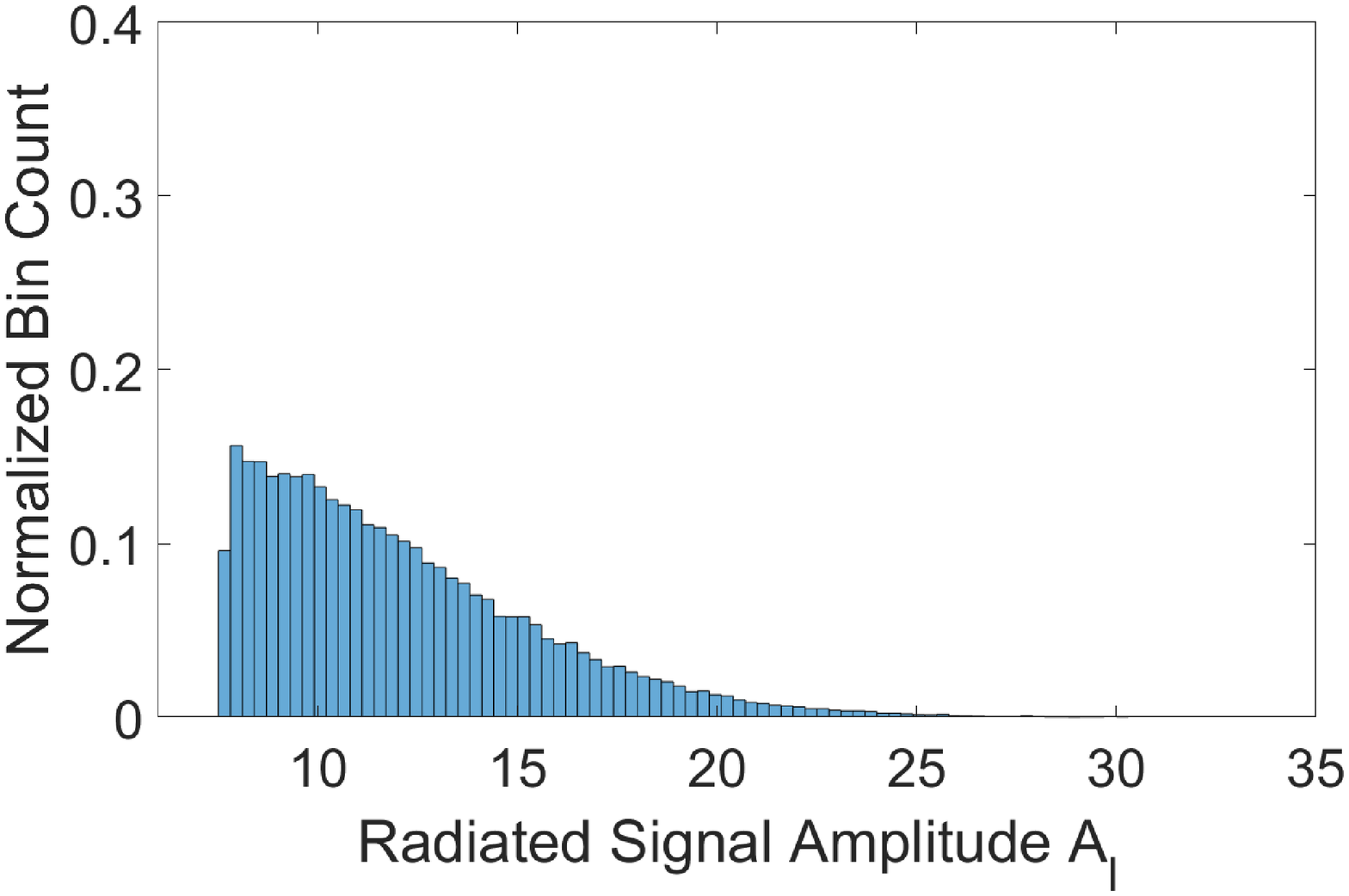}
		\label{fig:PhDistM100delta6dB}
	}
	\subfigure[]
	{
		\includegraphics[width=0.22\textwidth]{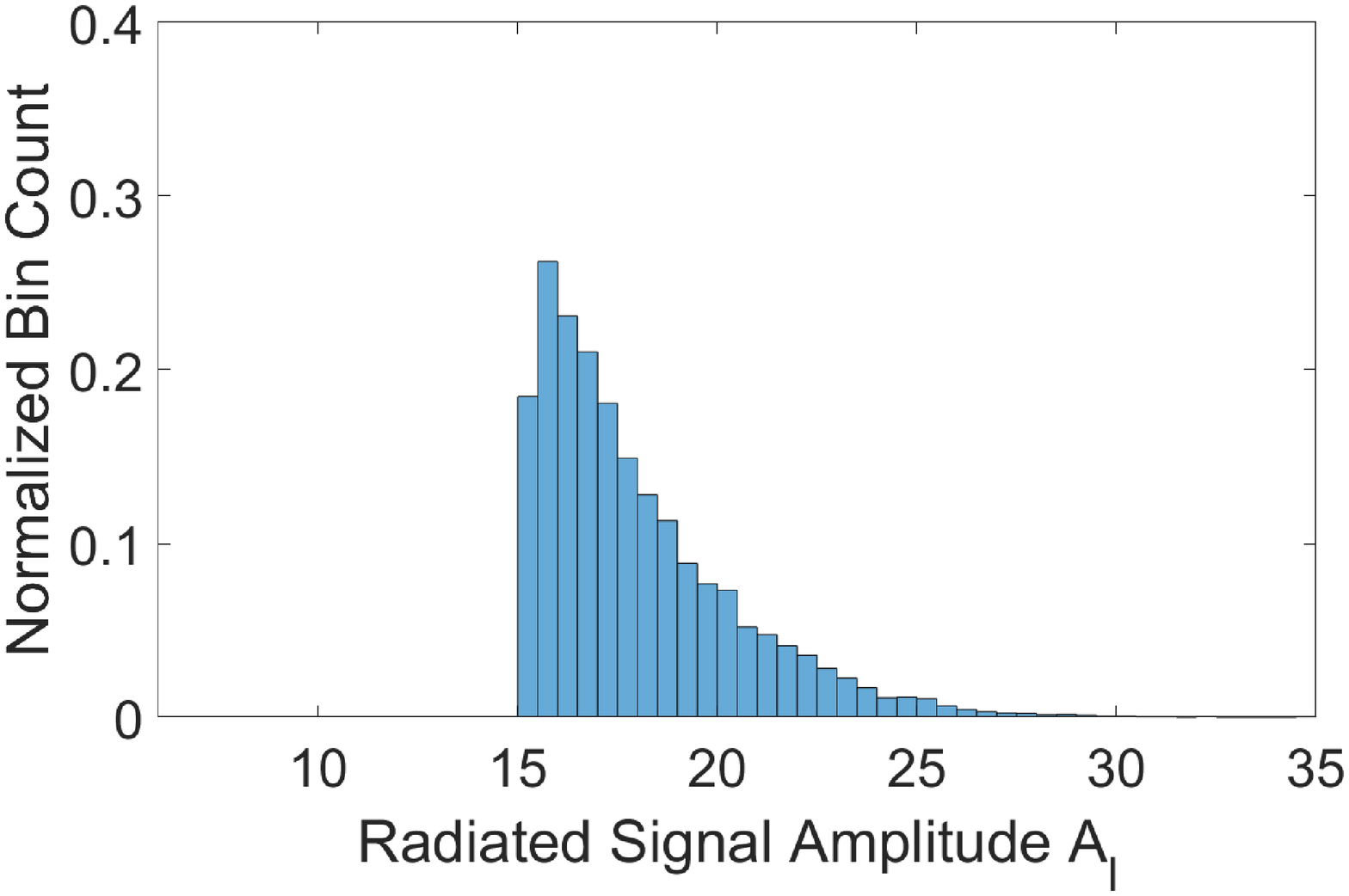}
		\label{fig:PhDistM100p0d1}
	}
	\subfigure[]
	{
		\includegraphics[width=0.22\textwidth]{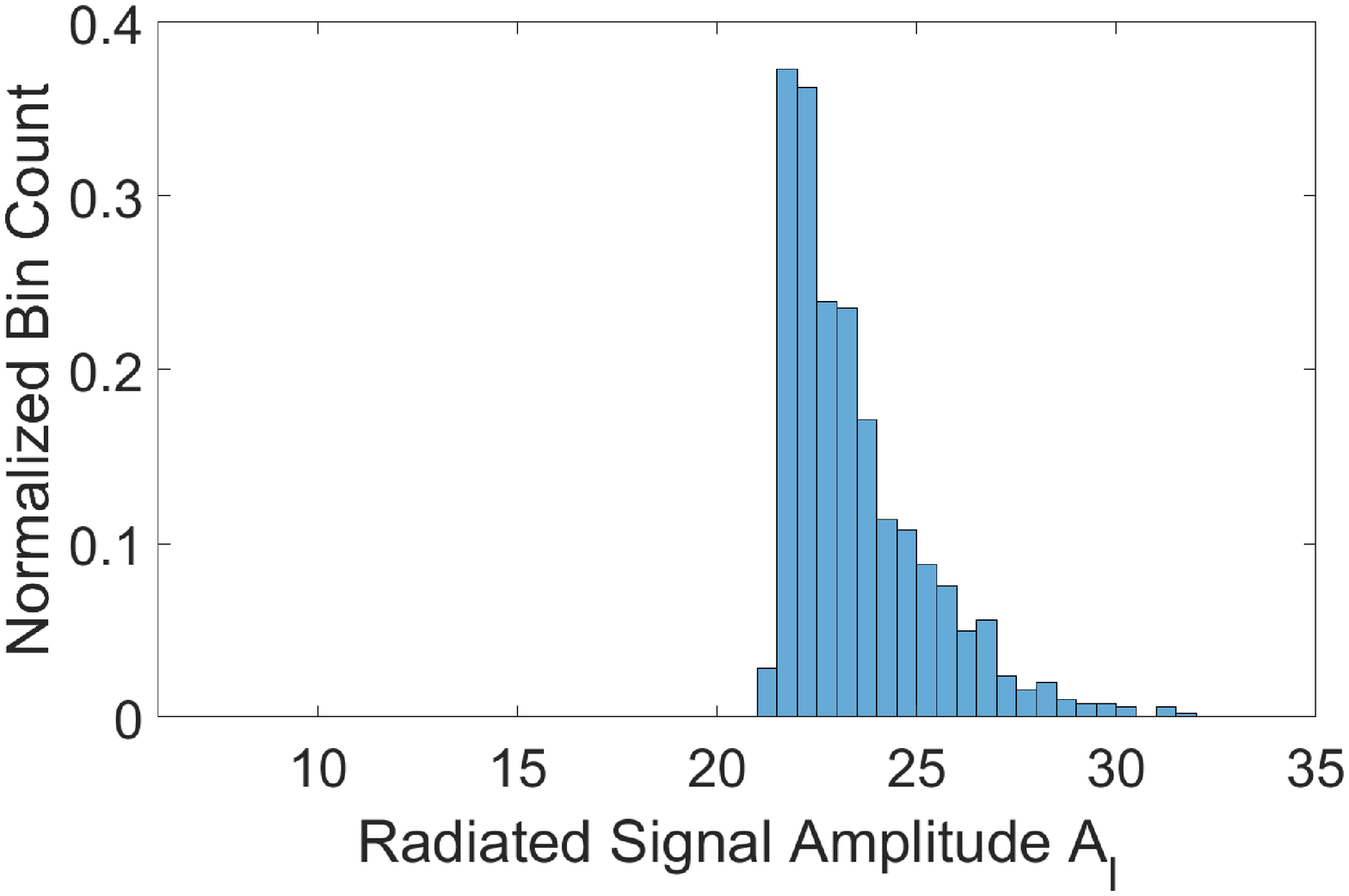}
		\label{fig:PhDistM100p0d01}
	}
	\caption{Histogram of amplitude of the radiated signal $A_l$, when the constraints are (a) $\hat{\varDelta}=1dB$ and $p_0=0.1$, (a) $\hat{\varDelta}=3dB$ and $p_0=0.1$, (c) $\hat{\varDelta}=6dB$ and $p_0=0.1$ in $A_0\varDelta\le A_l\le A_0\varDelta$, (d) $p_0=0.1$ and (e) $p_0=0.01$ in $A_l\ge A_0$, where $A_0=\sqrt{-M\ln(p_0)}$, $m=1,2,...,M=100$ and $l=1,2,...,L$.}
	\label{fig:resultsAmplitudeDistConstraint}
\end{figure}

We have analyzed the phase distribution, $\phi_m(l)$ in Eq.(\ref{sumofexponentials}), with two different number of antennas $16$ and $100$, which are spaced with $0.5$ of a wavelength. In Fig.\ref{fig:resultsPhaseDistM16}, phase distribution for each case spans homogeneously over $(-\pi,\pi]$ interval and the values are scattered around $1/2\pi$. Moreover, the expression in Eq.(\ref{pdfDelta2}) for the pdf of $\phi_{m}(l)$ is evaluated in Fig.\ref{fig:PhDistM16f}. The histogram in Fig.\ref{fig:PhDistM16} and the curve in Fig.\ref{fig:PhDistM16f} show similar behavior, so we can make sure that the Eq.(\ref{pdfDelta2}) almost represents the phase distribution. Calculations are also evaluated for $M=100$ and illustrated in Fig.\ref{fig:resultsPhaseDistM16}. Phase distribution of $\phi_{m}(l)$ shows uniform distributions when $M$ takes large number. The curve in Fig.\ref{fig:PhDistM100f} converges to $1/2\pi$ line. 

After making sure about the phase distribution which converges to uniform distribution by \textcolor{blue}{examining} the figures. We have also present\textcolor{blue}{ed} the histograms for amplitude distribution of the unintentional modulation for two different number of antenna element cases, $M$ is equal to $16$ and $100$. The histogram \textcolor{blue}{in Fig.\ref{fig:AmpDistM16}} almost follows the Rayleigh distribution pdf curve for $M=16$. However, for the case $M=100$, the histogram shows exactly same behavior as the curve of the pdf in Fig.\ref{fig:AmpDistM100}.

\begin{figure}[t]
	\centering
	\includegraphics[width=0.45\textwidth]{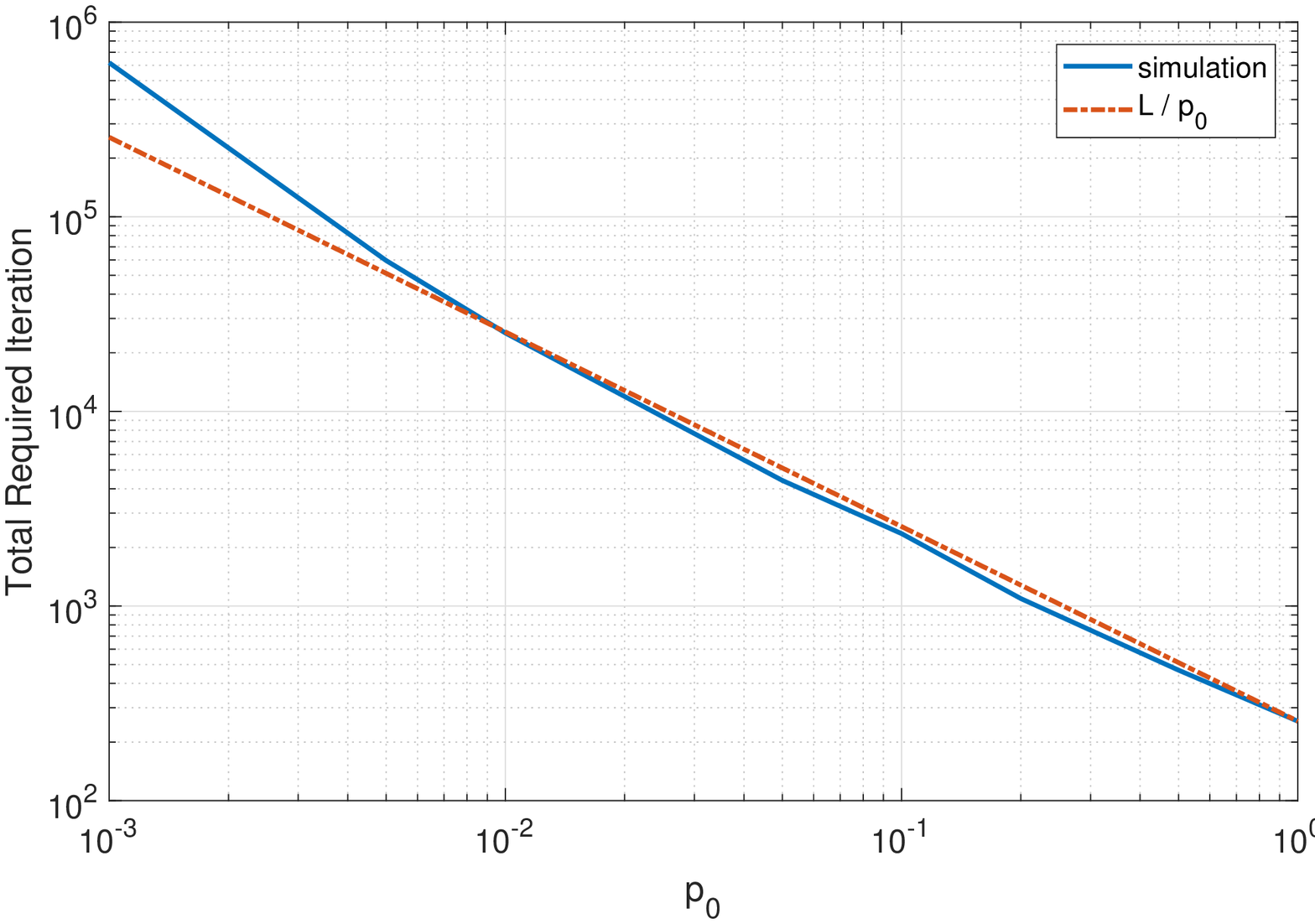}
	\caption{Total required iteration for generating MIMO waveforms with constraint $A_l\ge A_0$ where $A_0=\sqrt{-M\ln(p_0)}$ }
	\label{fig:TotalIteration}
\end{figure}

For different delta constraints, unintentional amplitude distributions are calculated and displayed in Fig.\ref{fig:resultsAmplitudeDistConstraint}. Each one of the figure include\textcolor{blue}{s} some part of the Rayleigh distribution. Null fixed waveform generation algorithm \textcolor{blue}{mainly contains} random permutations. Increase on the strictness of the constraints mean\textcolor{blue}{s} longer required waveform generation times. Let the radar waveform has $1024$ sub-pulses, $L=1024$ and $p_0$ is selected as $0.01$. Then, the required \textcolor{blue}{iterations} to reach $1024$ sub-pulses becomes approximated $L/p_0=102400$ cycles. In Fig.\ref{fig:TotalIteration}, total required random permutation iteration for null direction fixed waveform generation method is displayed when $M=16$ and $L=256$. Approximation to a $L/p_0$ curve can be easily seen from the simulation results in the \textcolor{blue}{Fig.\ref{fig:TotalIteration}.}

\begin{figure*}[htb]
	\centering
	\subfigure[]
	{
		\includegraphics[width=0.38\textwidth]{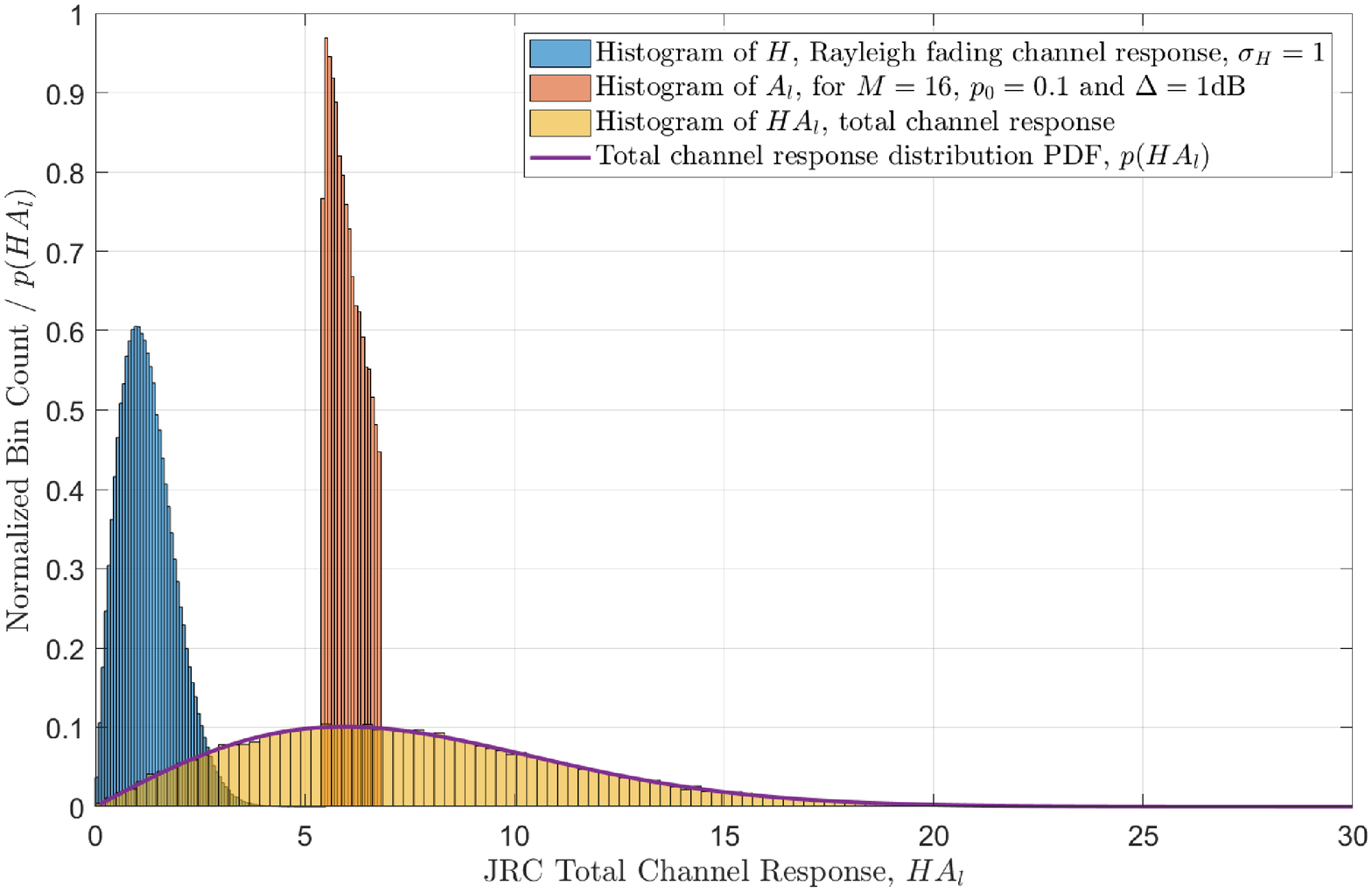}
		\label{fig:doubleRayCase3M16_1dB}
	}
	\subfigure[]
	{
		\includegraphics[width=0.38\textwidth]{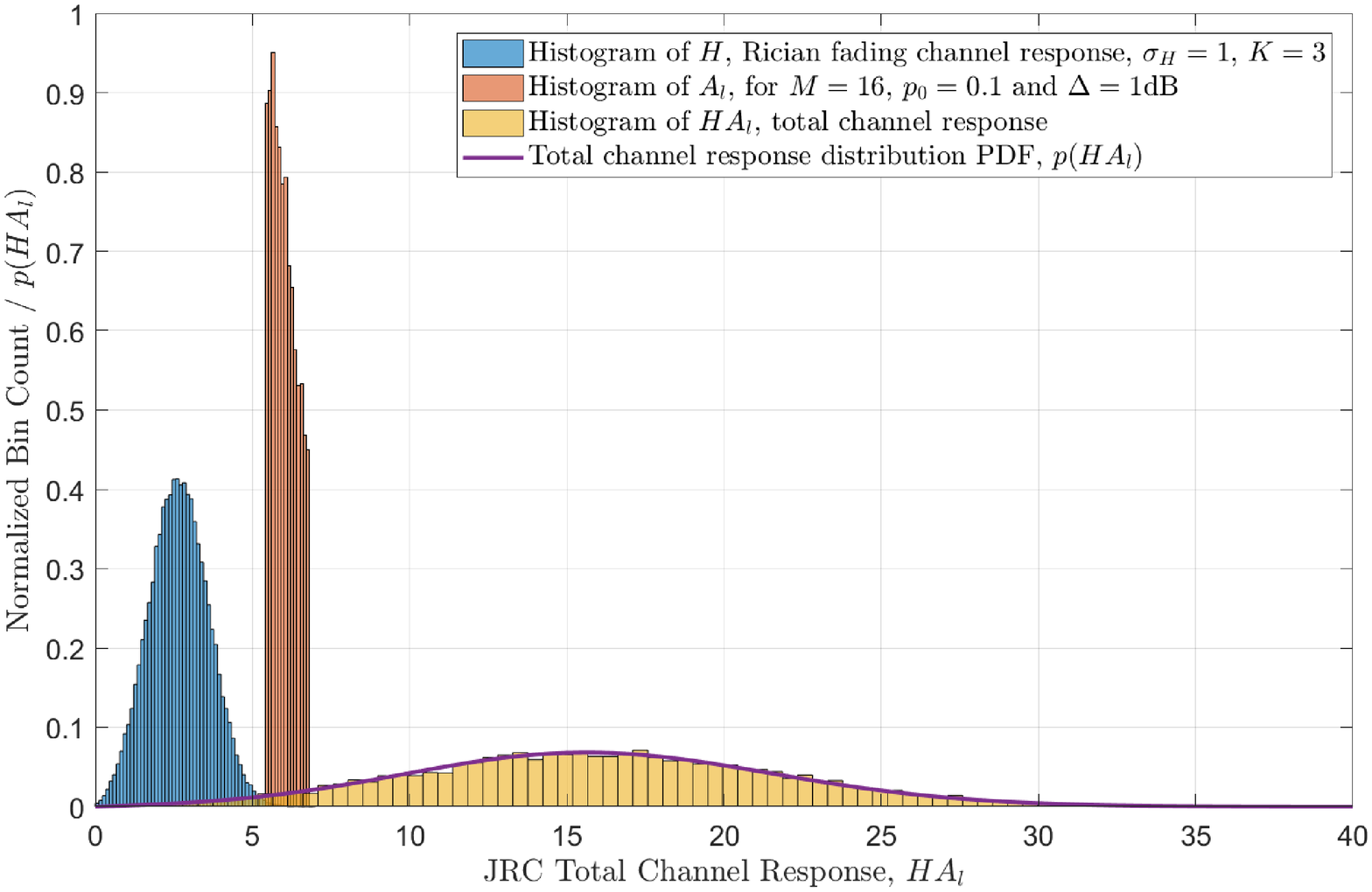}
		\label{fig:RayRiceCase3M16_1dB}
	}
	\subfigure[]
	{
		\includegraphics[width=0.38\textwidth]{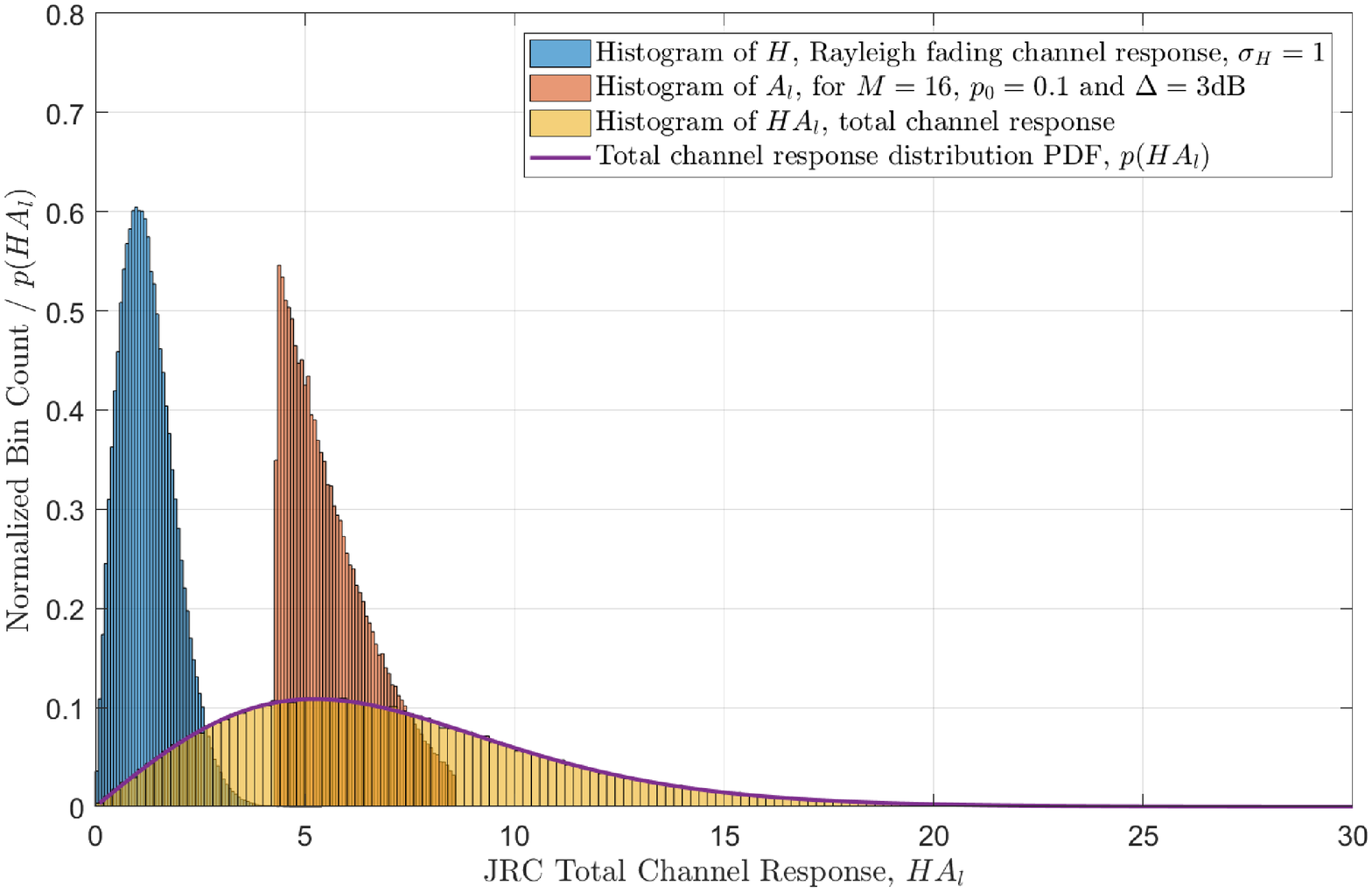}
		\label{fig:doubleRayCase3M16_3dB}
	}
	\subfigure[]
	{
		\includegraphics[width=0.38\textwidth]{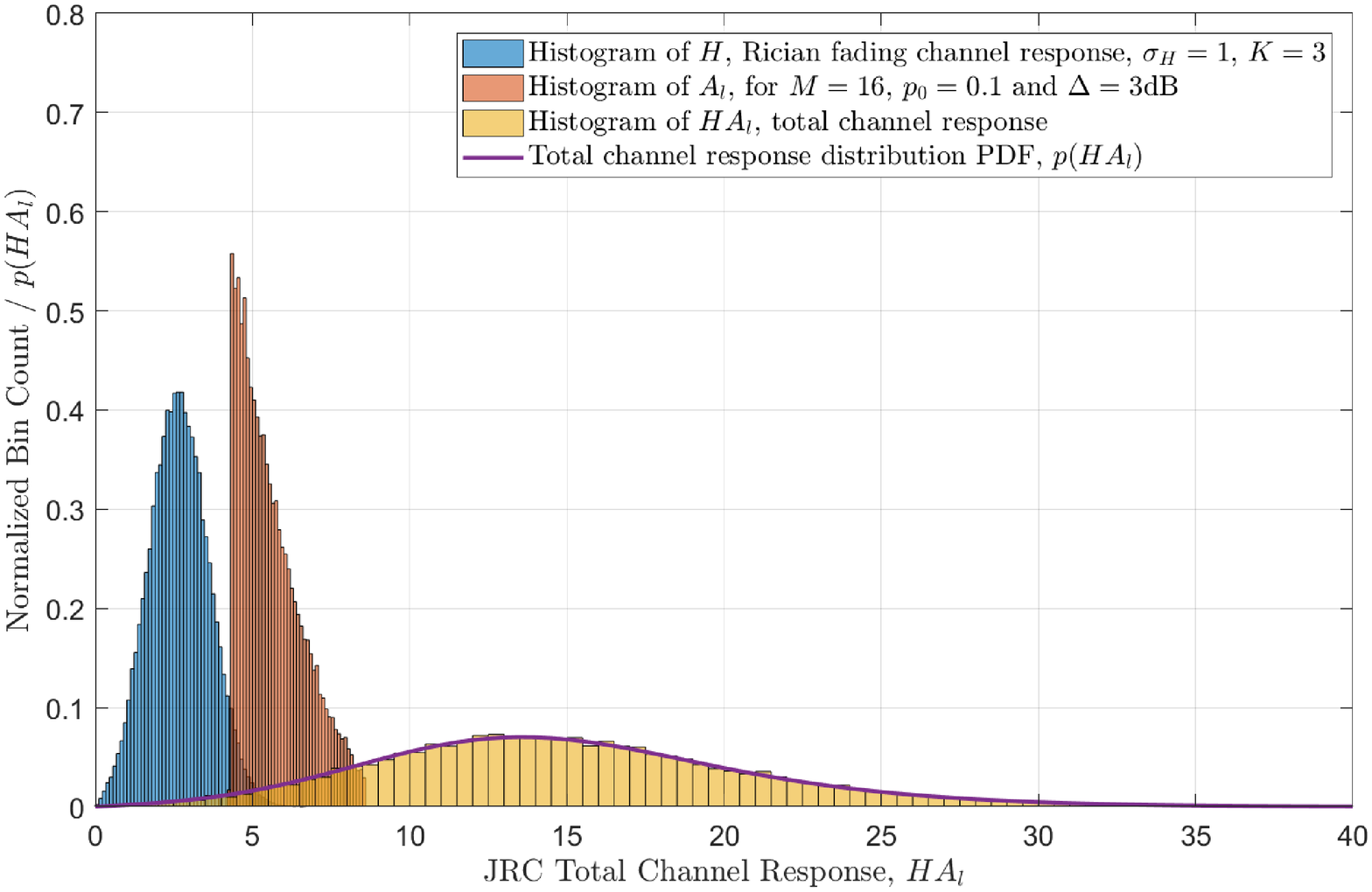}
		\label{fig:RayRiceCase3M16_3dB}
	}
	\subfigure[]
	{
		\includegraphics[width=0.38\textwidth]{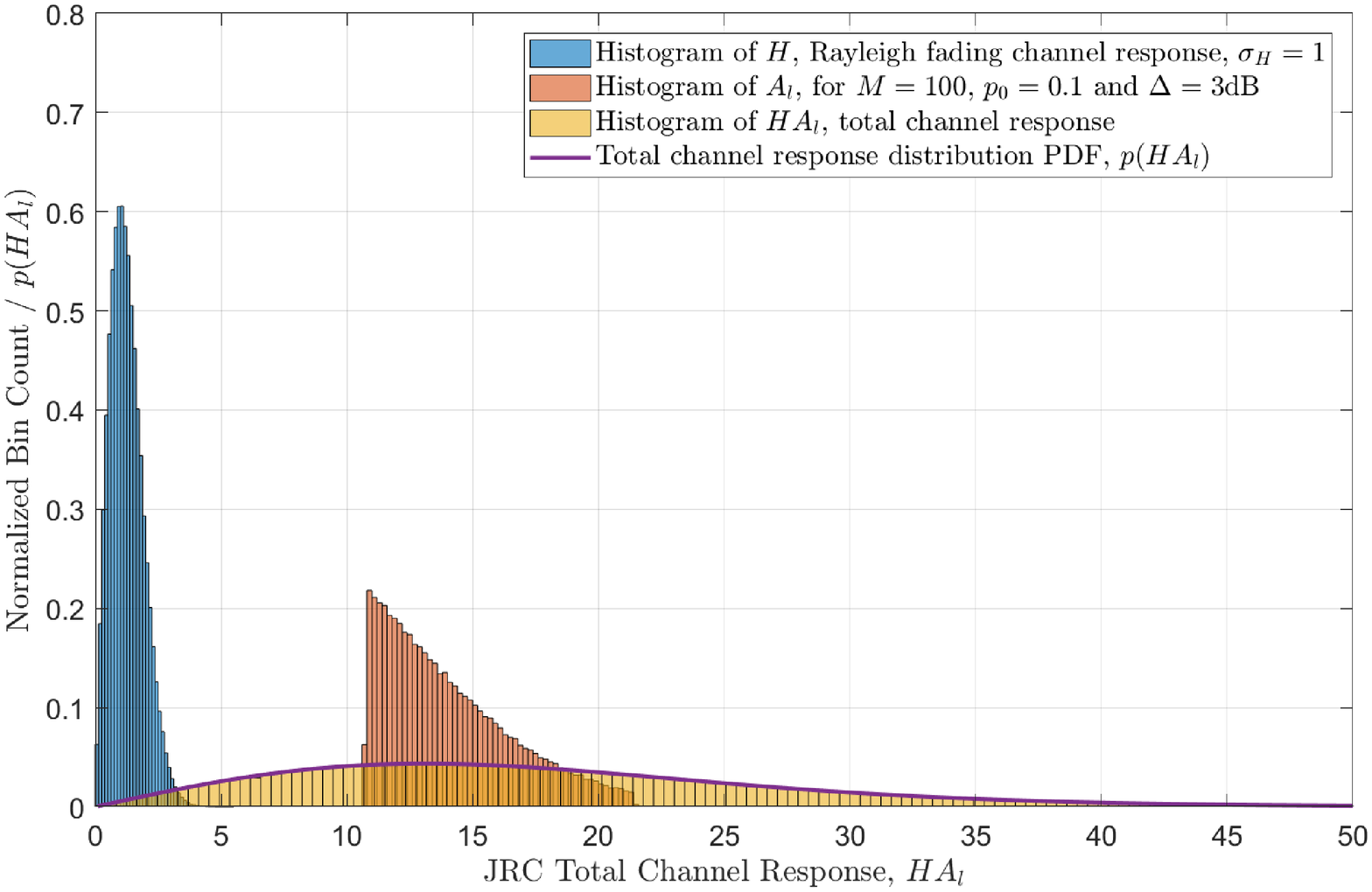}
		\label{fig:doubleRayCase3M100_3dB}
	}
	\subfigure[]
	{
		\includegraphics[width=0.38\textwidth]{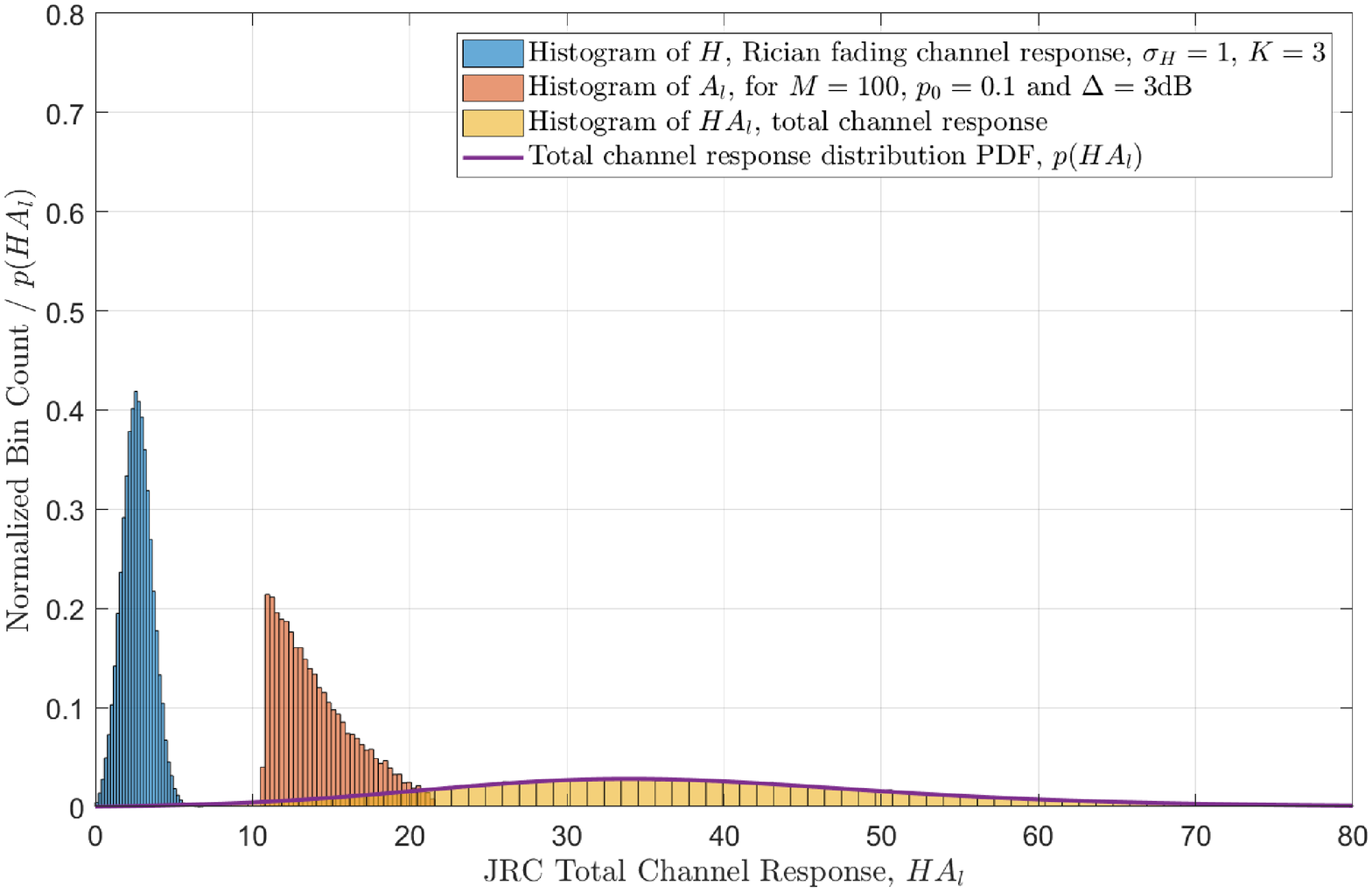}
		\label{fig:RayRiceCase3M100_3dB}
	}
	\caption{Histogram of $H$, $\sqrt{-M\ln(p_0)}/\varDelta\le A_l\le \sqrt{-M\ln(p_0)}\varDelta$ and $HA_l$, and probability density function $p_Z(z)$ in Eq.(\ref{LemmaDRayleigh}) \textcolor{blue}{when fading is Rayleigh (a,c,e) or Rician (b,d,f), for (a,b) when $M=16$, $\Delta=1dB$, (c,d) $M=16$, $\Delta=3dB$  and (e,f) $M=100$, $\Delta=3dB$ when $\sigma_H=1$, $p_0=0.1$ and $K=3$.}}
	\label{fig:doubleRayCase3}
\end{figure*}

\subsection{Amplitude Distribution under Rayleigh Fading}
Previously, we have proved that the unintentional modulation presents a Rayleigh distribution. Besides, numerical analysis in Sec.\ref{UnIntModAnalysis} is also proved that JRC systems configured with higher number of MIMO antennas shows very similar behavior with truncated Rayleigh distribution. Further, we have also investigated the channel distribution under Rayleigh \textcolor{blue}{ and Rician} fading conditions. Calculations are done for three cases.

\textcolor{blue}{While evaluating the pdf of amplitude distribution of the null direction fixed method under Rician fading channel, after $30$ summation steps, the resulting sum converges to a curve rapidly. Therefore, evaluations are terminated after $30$ summation iterations.}

\subsubsection{Case 1: With Delta constraint}
\textcolor{blue}{The} analysis is evaluated for $M=16$ and $M=100$, with the delta constraint for different values, i.e. $\sqrt{-M\ln(p_0)}/\varDelta\le A_l\le \sqrt{-M\ln(p_0)}\varDelta$ for $p_0=0.1$, $\sigma_H=1$ and $\Delta$ is $1dB$ and $3dB$. The histogram of the product distribution $HA_l$, and all the pdf curves in Fig.\ref{fig:doubleRayCase3} show exactly \textcolor{blue}{the} same behavior, \textcolor{blue}{hence} we can make sure that the Eq.(\ref{LemmaDRayleigh}) \textcolor{blue}{ and Eq.(\ref{LemmaDRayleigh2})} totally represent the \textcolor{blue}{corresponding} product distribution. For Rayleigh fading conditions, selecting smaller delta values results the product distribution more Rayleigh like shape, only with larger mean. \textcolor{blue}{ This feature is also repeated for Rician case and resulting distribution follows Rician like behavior.}

\subsubsection{Case 2: With Single Side constraint}
\begin{figure*}[htb]
	\centering
	\subfigure[]
	{
		\includegraphics[width=0.38\textwidth]{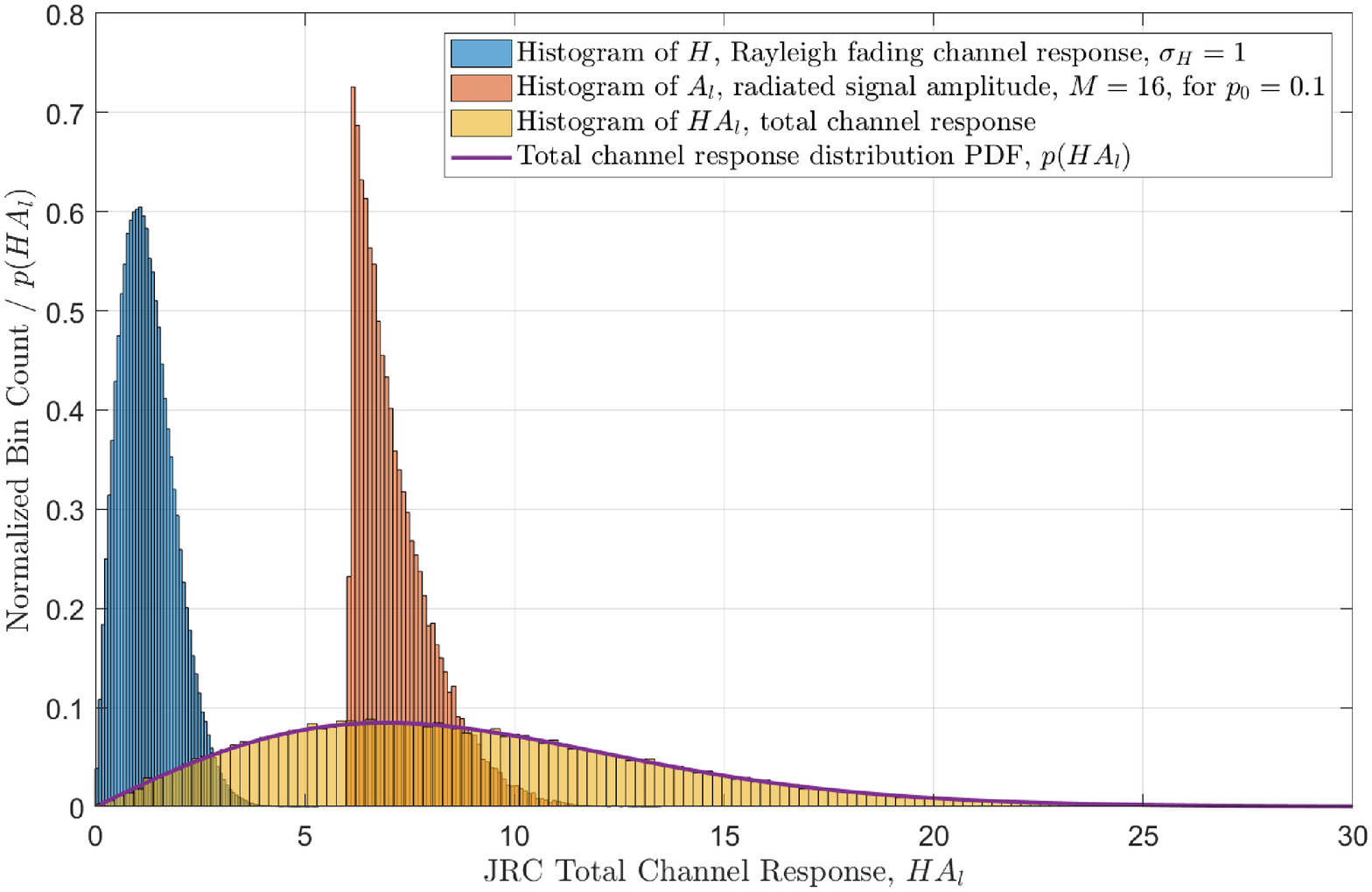}
		\label{fig:doubleRayCase2M16}
	}
	\subfigure[]
	{
		\includegraphics[width=0.38\textwidth]{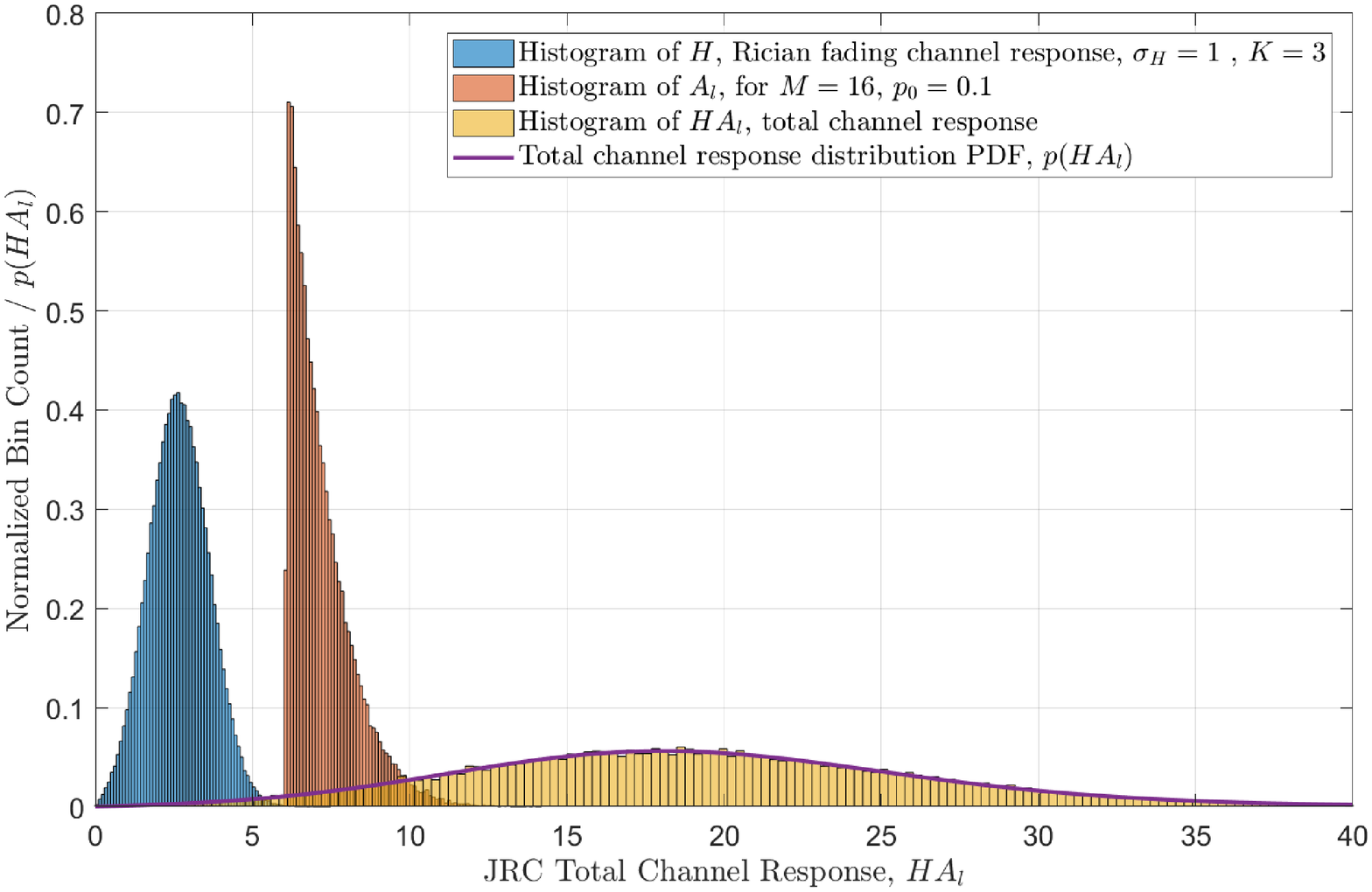}
		\label{fig:RayRiceCase2M16}
	}
	\subfigure[]
	{
		\includegraphics[width=0.38\textwidth]{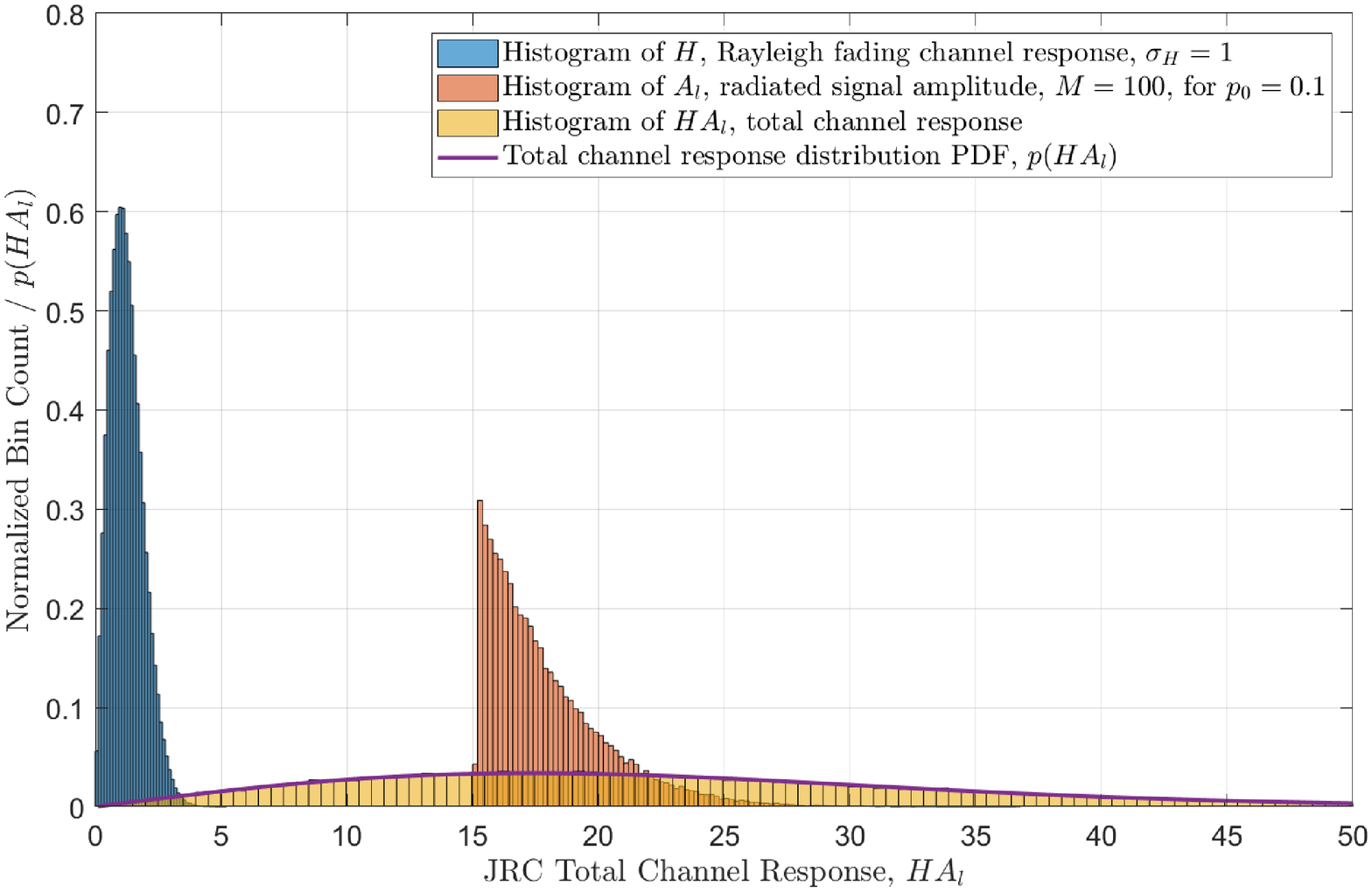}
		\label{fig:doubleRayCase2M100}
	}
	\subfigure[]
	{
		\includegraphics[width=0.38\textwidth]{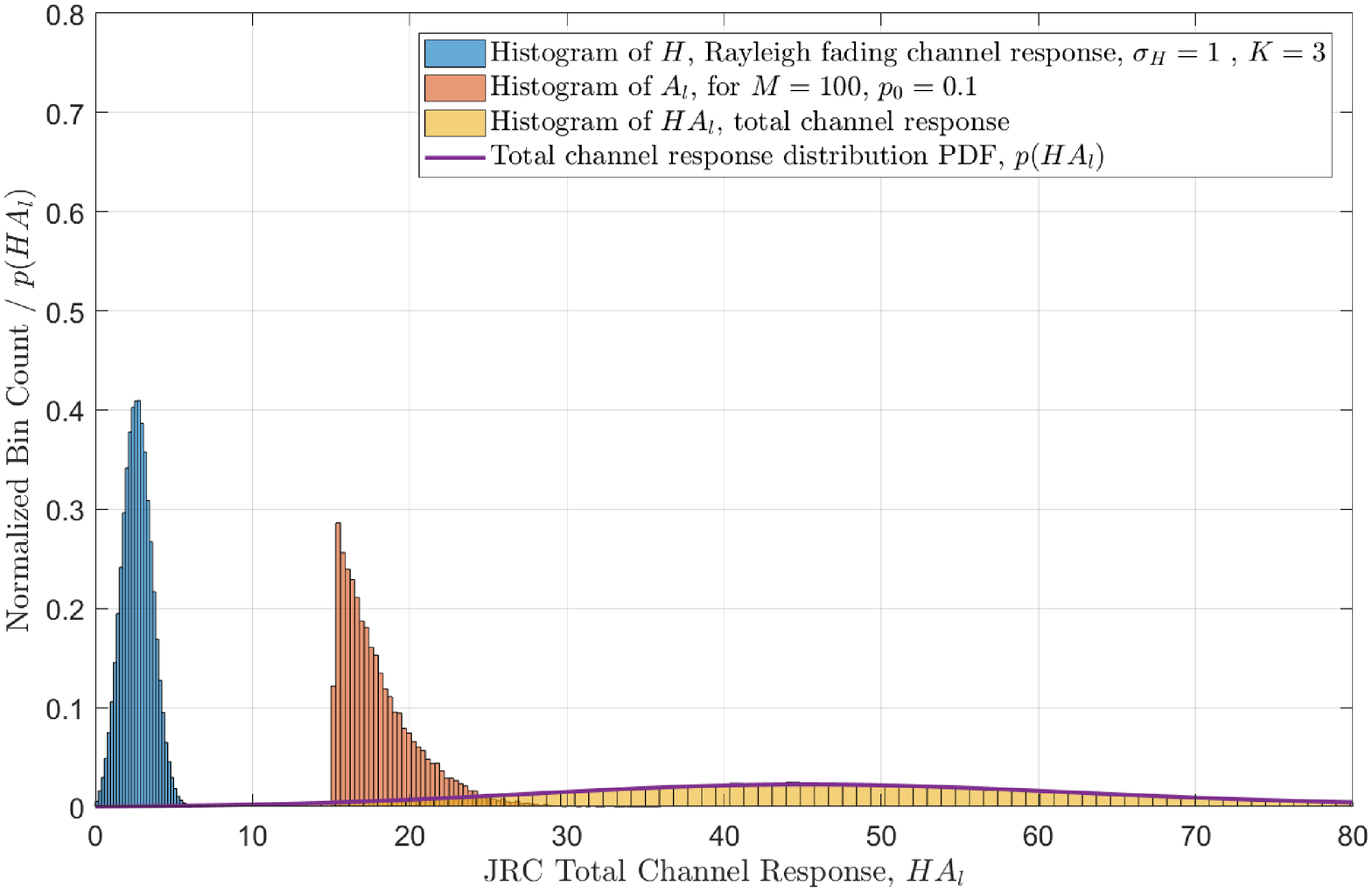}
		\label{fig:RayRiceCase2M100}
	}
	\caption{Histogram of $H$, $A_l\ge \sqrt{-M\ln(p_0)}$ and $HA_l$, and probability density function $p_Z(z)$ in Eq.(\ref{singleRaypdf}) \textcolor{blue}{when fading is Rayleigh (a,c) or Rician (b,d), for (a,b) when $M=16$ and (c,d) $M=100$ for $\sigma_H=1$, $p_0=0.1$ and $K=3$.}}
	\label{fig:doubleRayCase2}
\end{figure*}
The analysis is evaluated for $M=16$ and $M=100$, with the minimum constraint as $A_0=\sqrt{-M\ln(p_0)}$ when $p_0=0.1$. Then, the variance of the both real and the imaginary part of the complex channel gain for Rayleigh \textcolor{blue}{and Rician} channel, $\sigma_H$, is selected as $1$. \textcolor{blue}{For Rician cases, shape parameter $K$ is selected as $3$.} The histogram of the product distribution $HA_l$, and the pdf curves in Fig.\ref{fig:doubleRayCase2} show exactly same behavior, \textcolor{blue}{hence} we can make sure that the Eq.(\ref{singleRaypdf}) \textcolor{blue}{and Eq.(\ref{singleRaypdf2})} totally represent the product distribution. \textcolor{blue}{Also} Eq.(\ref{singleRaypdf}) \textcolor{blue}{ and Eq.(\ref{singleRaypdf2})} require a numerical calculation of generalized incomplete gamma function, $\Gamma(a,x;b)$. In \ref{NumericGenIncGamFunc}, a method for the numerical calculation which contains infinite summation is given. For our case, $a=0$ in $\Gamma(a,x;b)$, after $20$ step summation, the resulting sum converges to a curve rapidly. Determining a lower constraint results bigger mean value for the product distribution comparing with no constraint case \textcolor{blue}{in Fig.\ref{fig:doubleRayCase2}}.

\subsubsection{Case 3: Without any constraint}
\begin{figure*}[htb]
	\centering
	\subfigure[]
	{
		\includegraphics[width=0.38\textwidth]{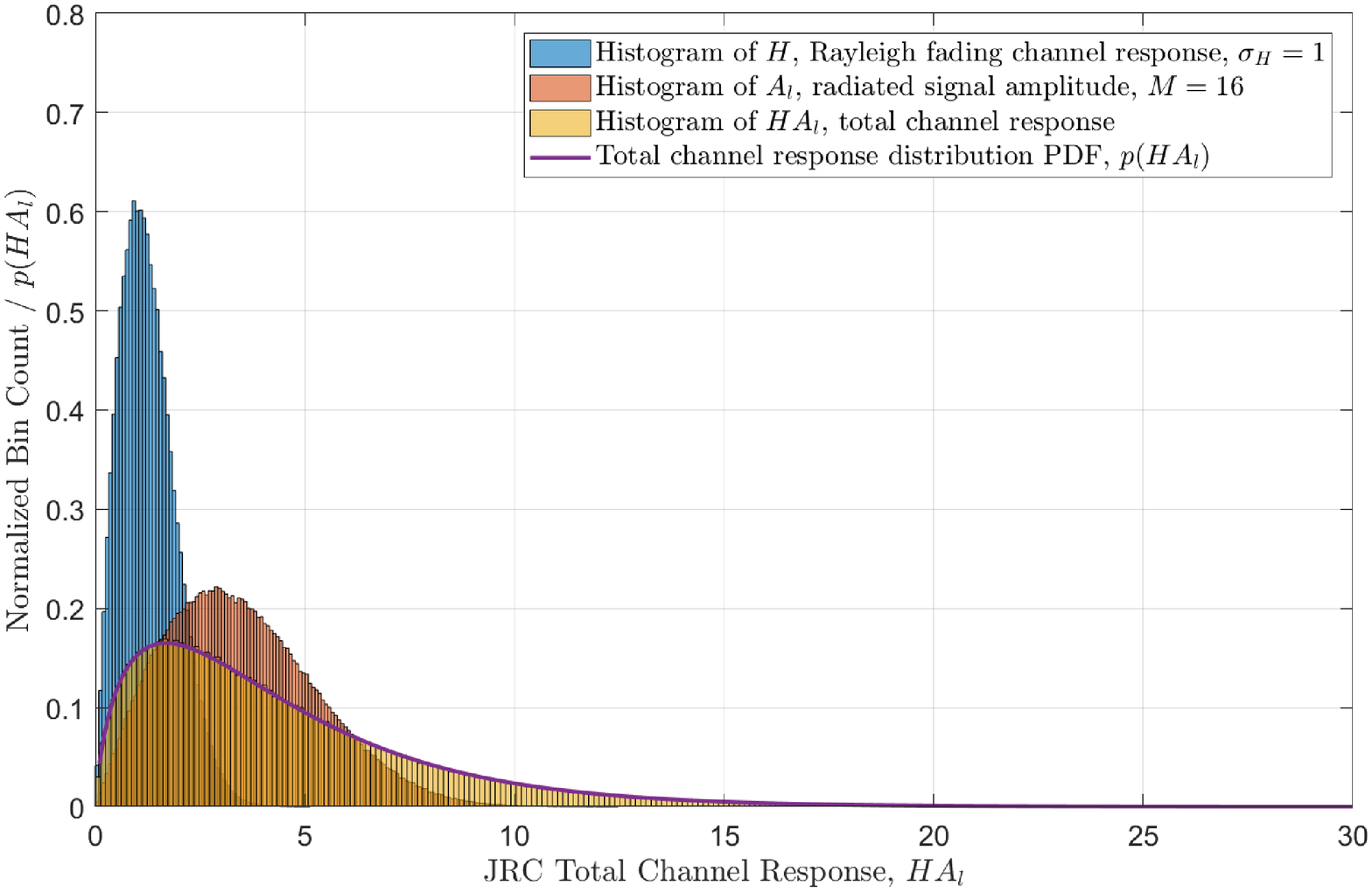}
		\label{fig:doubleRayCase1M16}
	}
	\subfigure[]
	{
		\includegraphics[width=0.38\textwidth]{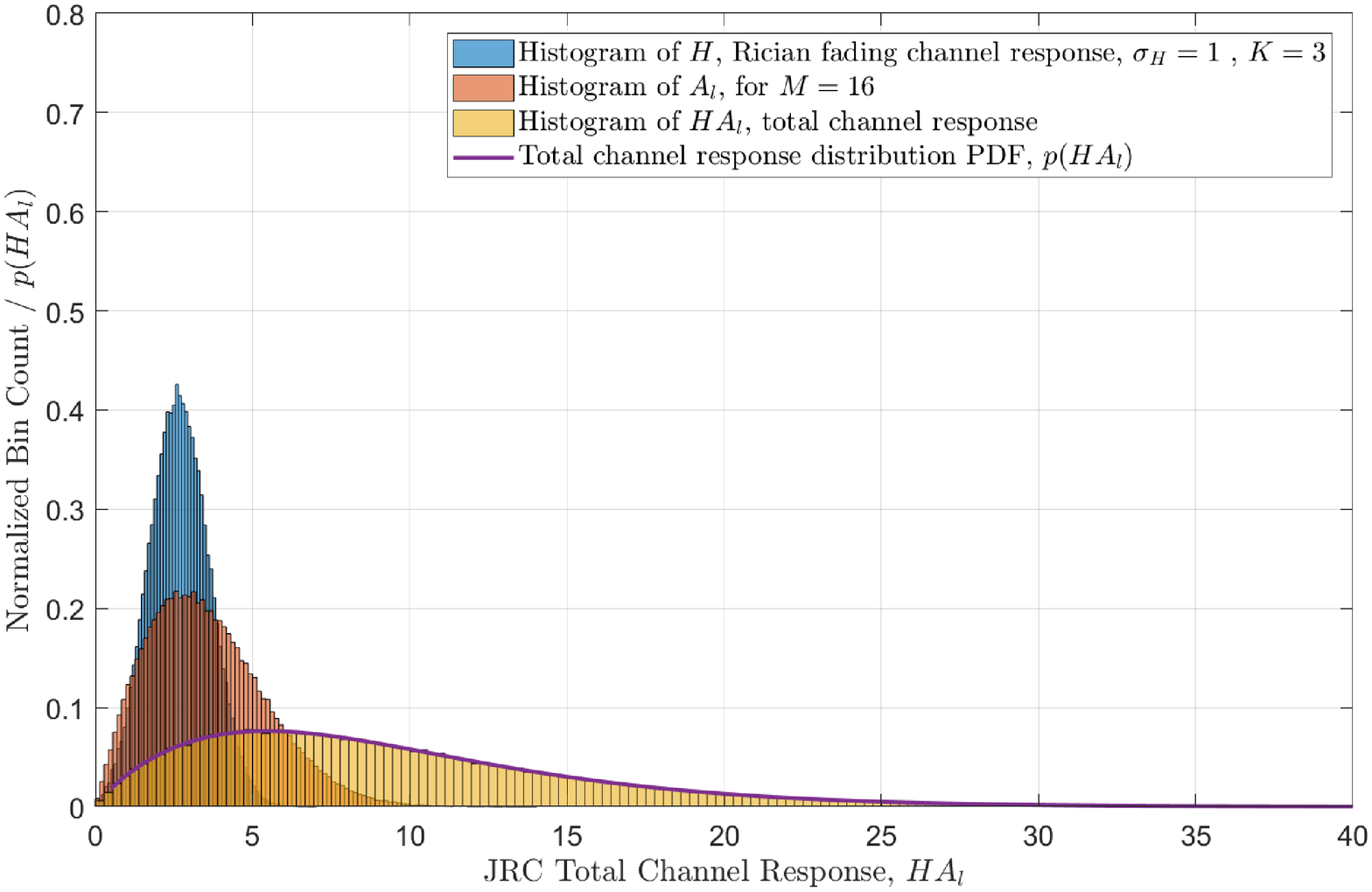}
		\label{fig:RayRiceCase1M16}
	}
	\subfigure[]
	{
		\includegraphics[width=0.38\textwidth]{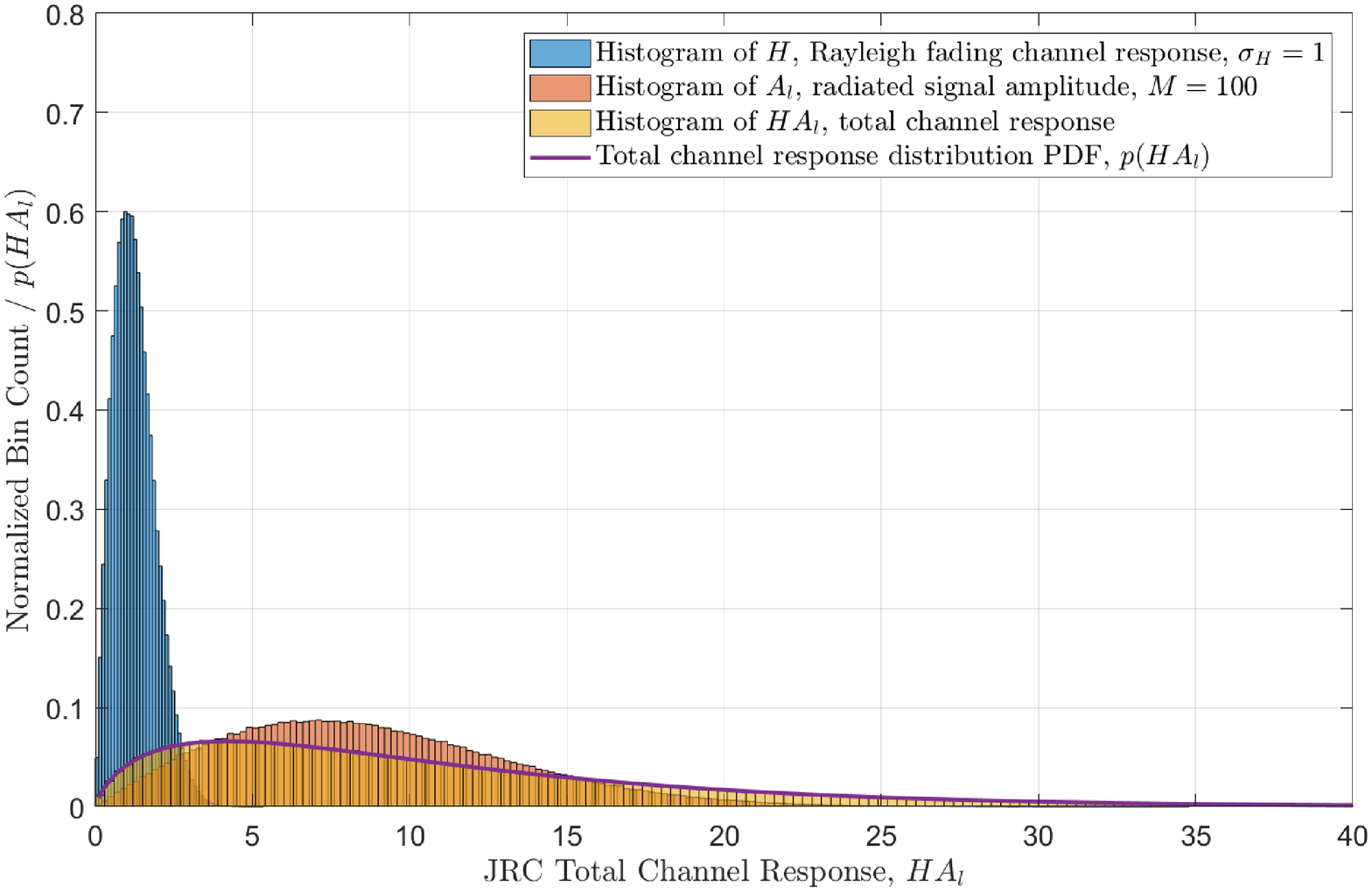}
		\label{fig:doubleRayCase1M100}
	}
	\subfigure[]
	{
		\includegraphics[width=0.38\textwidth]{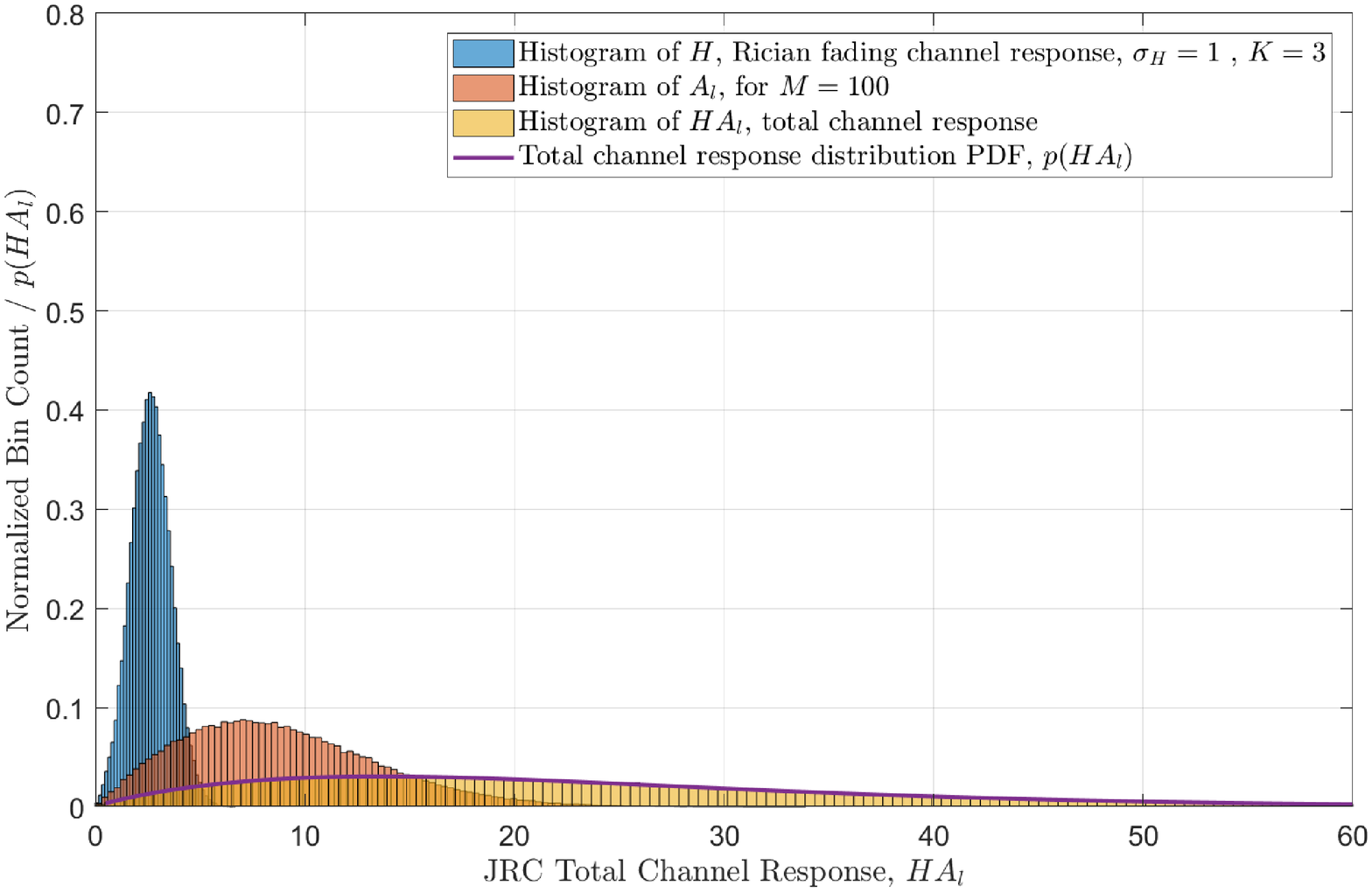}
		\label{fig:RayRiceCase1M100}
	}
	\caption{Histogram of $H$, $A_l$ and $HA_l$, and probability density function $p_Z(z)$ in Eq.(\ref{pdfDoubleRay}) \textcolor{blue}{when fading is Rayleigh (a,c) or Rician (b,d), for (a,b) when $M=16$ and (c,d) $M=100$ for $\sigma_H=1$ and $K=3$.}}
	\label{fig:doubleRayCase1}
\end{figure*}
The problem becomes easier without any constraint cases. \textcolor{blue}{For Rayleigh fading cases,} the distribution can be calculated as cascaded double Rayleigh channel. \textcolor{blue}{Similary, for Rician fading cases,} the distribution can be calculated as cascaded \textcolor{blue}{ Rician and} Rayleigh channel. The analysis is evaluated for $M=16$ and $M=100$. Then, mean channel gain for Rayleigh channel is selected as $\sigma_H\sqrt{\pi/2}$ which is $\sqrt{\pi/2}$ when $\sigma_H$ is $1$, where $\sigma_H$ is the variance of the both real and the imaginary part of the complex channel gain. The histogram of the product distribution $HA_l$, and the pdf curves in Fig.\ref{fig:doubleRayCase1} show exactly same behavior, \textcolor{blue}{hence} we can make sure that the Eq.(\ref{pdfDoubleRay}) totally represents the product distribution. Also, the product distribution for higher $M$ value displays more dispersive distribution as we have expected, hence it has bigger variance. 

\subsection{Coherent MIMO JRC Capacity Analysis} 
After determining the unintentional modulation distribution caused by the radiated signal towards communication direction, capacity formulas \textcolor{blue}{are} evaluated for three different cases.  

\subsubsection{Case 1: AWGN Channel} 
\begin{figure}[htb]
	\centering
	\includegraphics[width=0.45\textwidth]{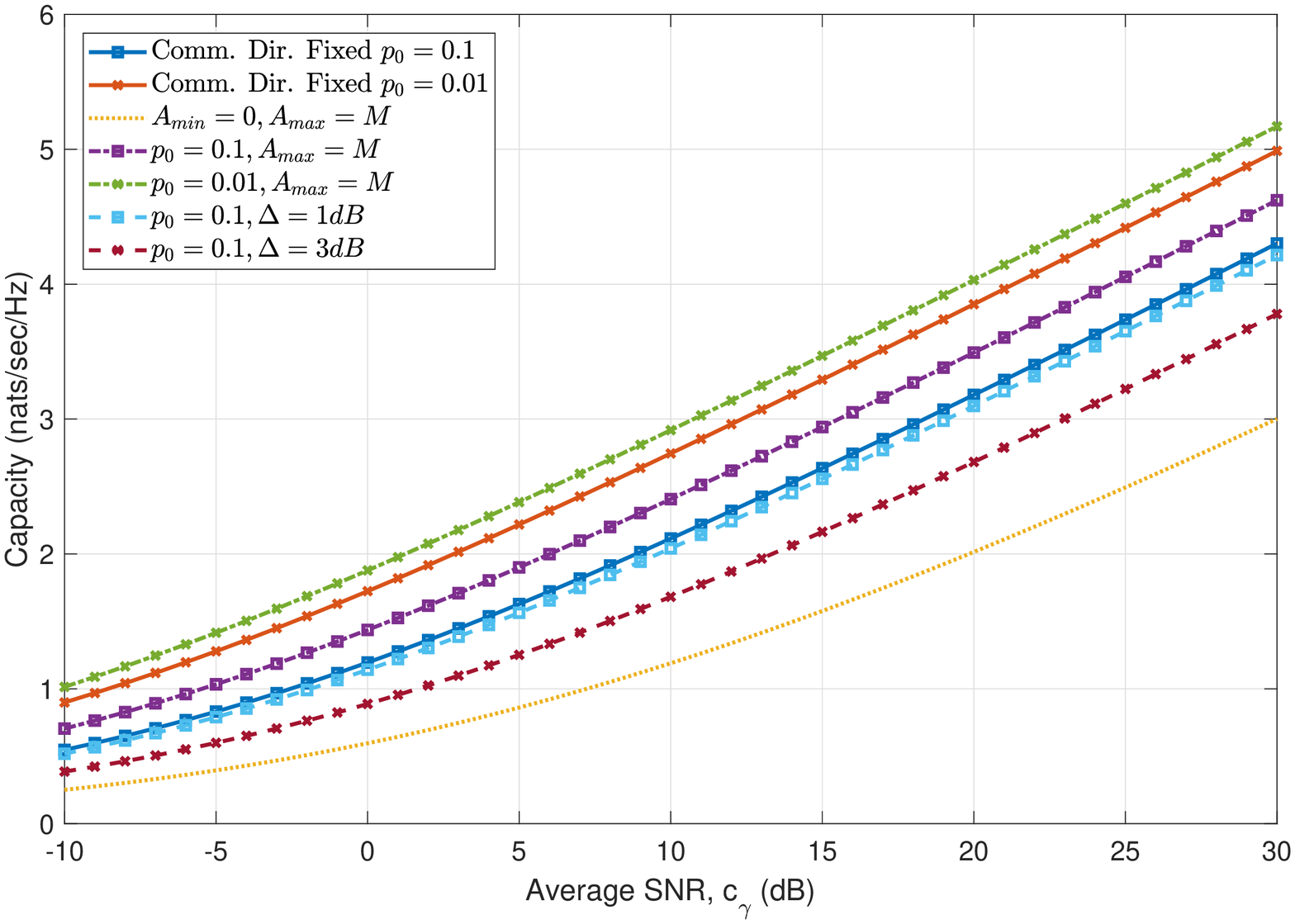}
	\caption{Capacity for JRC Capable Coherent MIMO Radar}
	\label{fig:Capacity}
\end{figure}
Calculations are done for $M=16$ using the Eq.(\ref{CsnrNull2}). Capacity of the communication direction fixed formula is evaluated when the radiated signal is constant at the level of $A_0=\sqrt{-16\ln(0.1)}=6.07$. Results are displayed in Fig.\ref{fig:Capacity}. As we have expected, applying no constraint, $A_{1}=0$ and $A_{2}=M$, presents the worst capacity, which is also the ergodic channel capacity of the Rayleigh fading channel with known CSI, $C_{RCSI}$. Between these capacity curves in Fig.\ref{fig:Capacity}, the techniques with a defined constraint with $\varDelta$ values, $6.07\varDelta\le A_l\le 6.07\varDelta$, are located. The capacity curves with higher $\varDelta$ values are getting closer to the no constraint case. More to that, increasing $A_0$ value with changing the probability $p_0$ always presents higher capacity due to its higher average power. On the other hand, smaller $p_0$ means higher search iterations, so higher search times \textcolor{blue}{during waveform generation process}. 

The capacity expressions in (\ref{CsnrNullspecialprop1}), (\ref{CsnrNullspecialprop2P1}) and (\ref{CsnrNullspecialprop2}) are calculated for three different $M$ values, as $8$, $16$ and $32$ when $SNR$ \textcolor{blue}{is} $10dB$. Then, the differences between approximated expressions and general capacity formula are computed in Table.(\ref{tab:differences}). If $M$ is greater than $8$, the difference converges to significantly smaller values. \textcolor{blue}{Therefore, selecting $M$ as $16$ is enough to show the channel capacity behaviors.}
\begin{table}[t]
	\caption{The differences between capacity expressions and the approximated formulas.}
	\label{tab:differences}
	\begin{center}
		\resizebox{.8\columnwidth}{!}{
			\begin{tabular}{c|c c c}
				\emph{Difference} & $M=8$ & $M=16$ & $M=32$ \\ \hline \\
				$|C_{RCSI}-C_{rp}(0,M)|$ & $8.7\times10^{-4}$ & $5.6\times10^{-7}$ & $2.1\times10^{-14}$ \\ \hline \\
				$|\bar{C}_{rp}(A_0,M)-C_{rp}(A_0,M)|$ & $6.2\times10^{-3}$ & $1.3\times10^{-6}$ & $3.2\times10^{-13}$ \\ \hline
			\end{tabular}
		}
	\end{center}
\end{table}
\subsubsection{Case 2: Slow Fading Rayleigh Channel} 
Before doing \textcolor{blue}{the corresponding} calculations, the major design parameter $p_{out}$ needs to be determined to reach outage capacity results. $p_{out}$ must be selected as \textcolor{blue}{a number which will maximize} the average rate $R_{out}$ in Eq.(\ref{OutRateSlow1}). For each communication method, we calculated the curves $R_{out}$ with respect to $p_{out}$ values. \textcolor{blue}{Capacity curves for Rayleigh fading is displayed} in Fig.\ref{fig:OutrateSlowFading} \textcolor{blue}{and for Rician fading in Fig.\ref{fig:OutrateSlowFadingRician}}. Outage probabilities for maximum rate values are displayed for each curves. These labeled $p_{out}$ values are used for the capacity with outage analysis \textcolor{blue}{are calculated} for \textcolor{blue}{the} slow fading case.
\begin{figure*}[htb]
	\centering
	\subfigure[]
	{
		\includegraphics[width=0.45\textwidth]{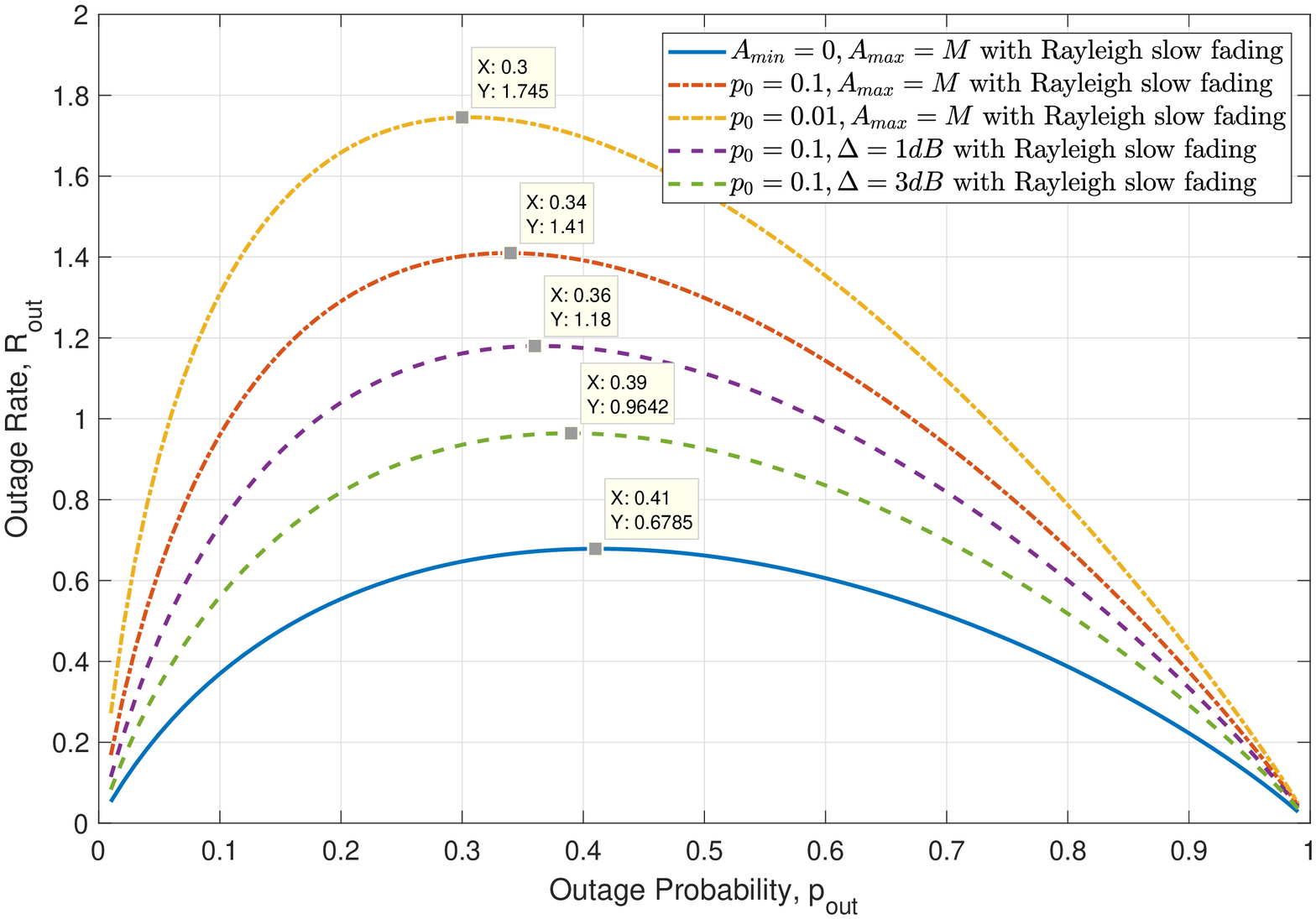}
		\label{fig:OutrateSlowFading}
	}
	\subfigure[]
	{
		\includegraphics[width=0.45\textwidth]{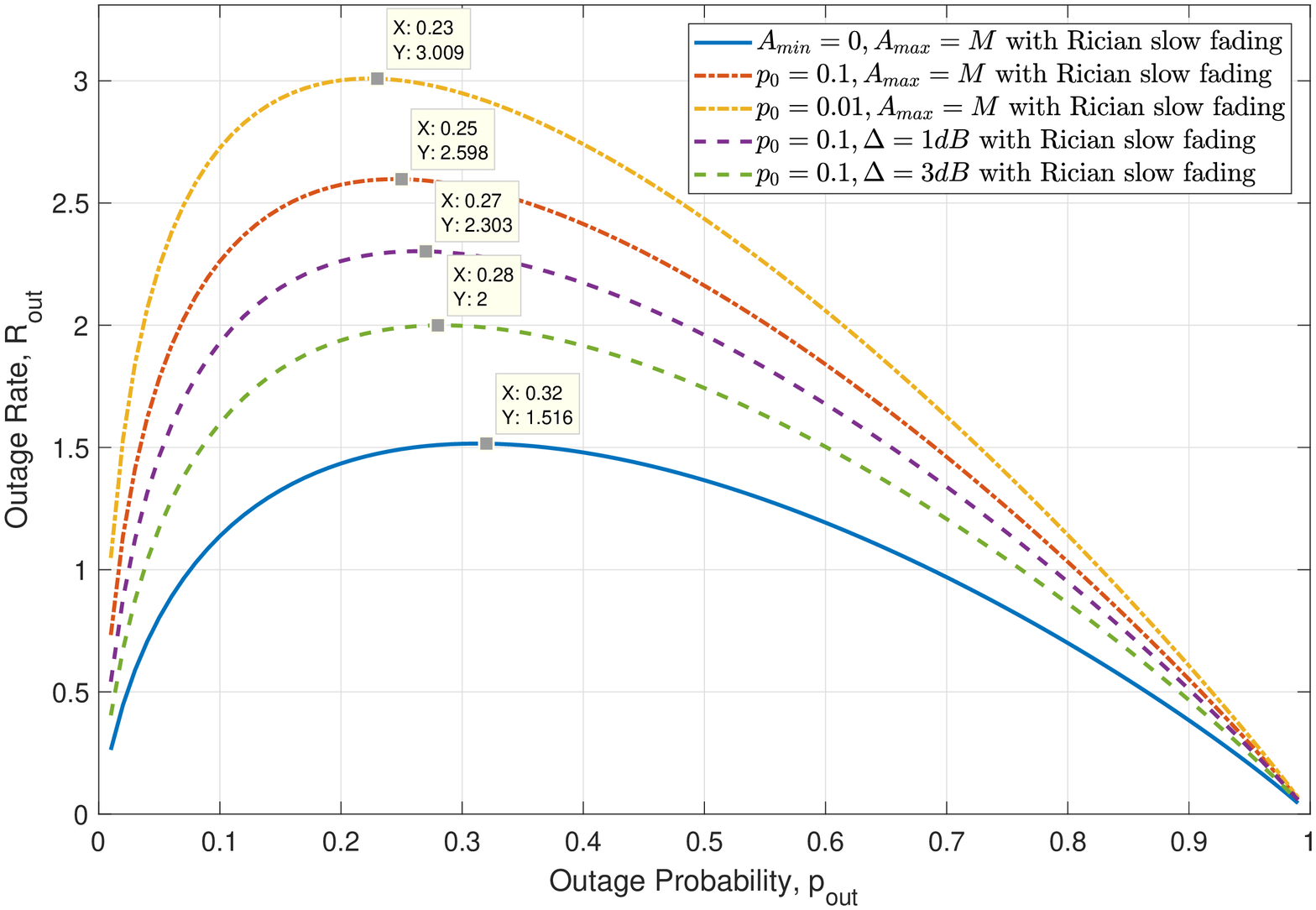}
		\label{fig:OutrateSlowFadingRician}
	}
	\caption{Average Rate with outage vs outage probability for JRC Capable Coherent MIMO Radar under \textcolor{blue}{Rayleigh (a) and Rician (b)} slow fading conditions}
	\label{fig:ALLOutrateSlowFading}
\end{figure*}

Capacity calculations are done for $M=16$, $\sigma_H^2=0.5$ \textcolor{blue}{and $K=3$} with using the Eq.(\ref{CsnrNull2DF}). Results are displayed in Fig.\ref{fig:ALLOutrateSlowFading}. Capacity of the communication direction fixed formula is evaluated when the radiated signal is constant at the level of $A_0=\sqrt{-16\ln(0.1)}=6.07$. Since, we \textcolor{blue}{have} consider\textcolor{blue}{ed} the system in outage when symbols are received below the minimum SNR limit over a long period of time. \textcolor{blue}{More to that, the} transmitter does not know instantaneous CSI. \textcolor{blue}{Then, the capacity with outage curves must display lower rate values.} Capacity curves in Fig.\ref{fig:ALLOutrateSlowFading} shows the same behavior with AWGN case, however \textcolor{blue}{Rayleigh fading case} presents lower capacity. \textcolor{blue}{On the other hand, Rician fading case displays more capacity due to its average power is derived by the shape parameter $K$.}

\begin{figure*}[htb]
	\centering
	\subfigure[]
	{
		\includegraphics[width=0.45\textwidth]{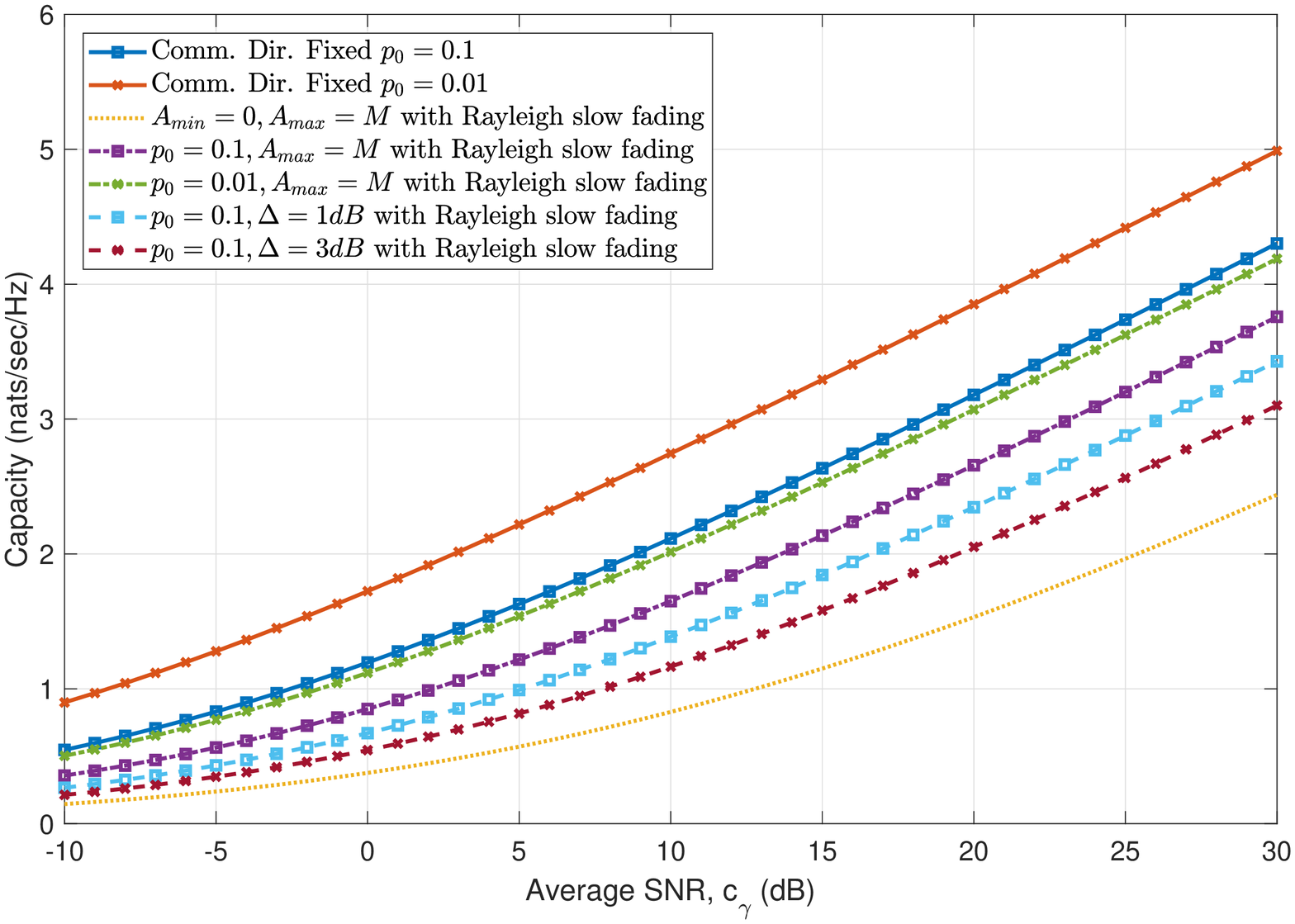}
		\label{fig:CapacitySlowFading}
	}
	\subfigure[]
	{
		\includegraphics[width=0.45\textwidth]{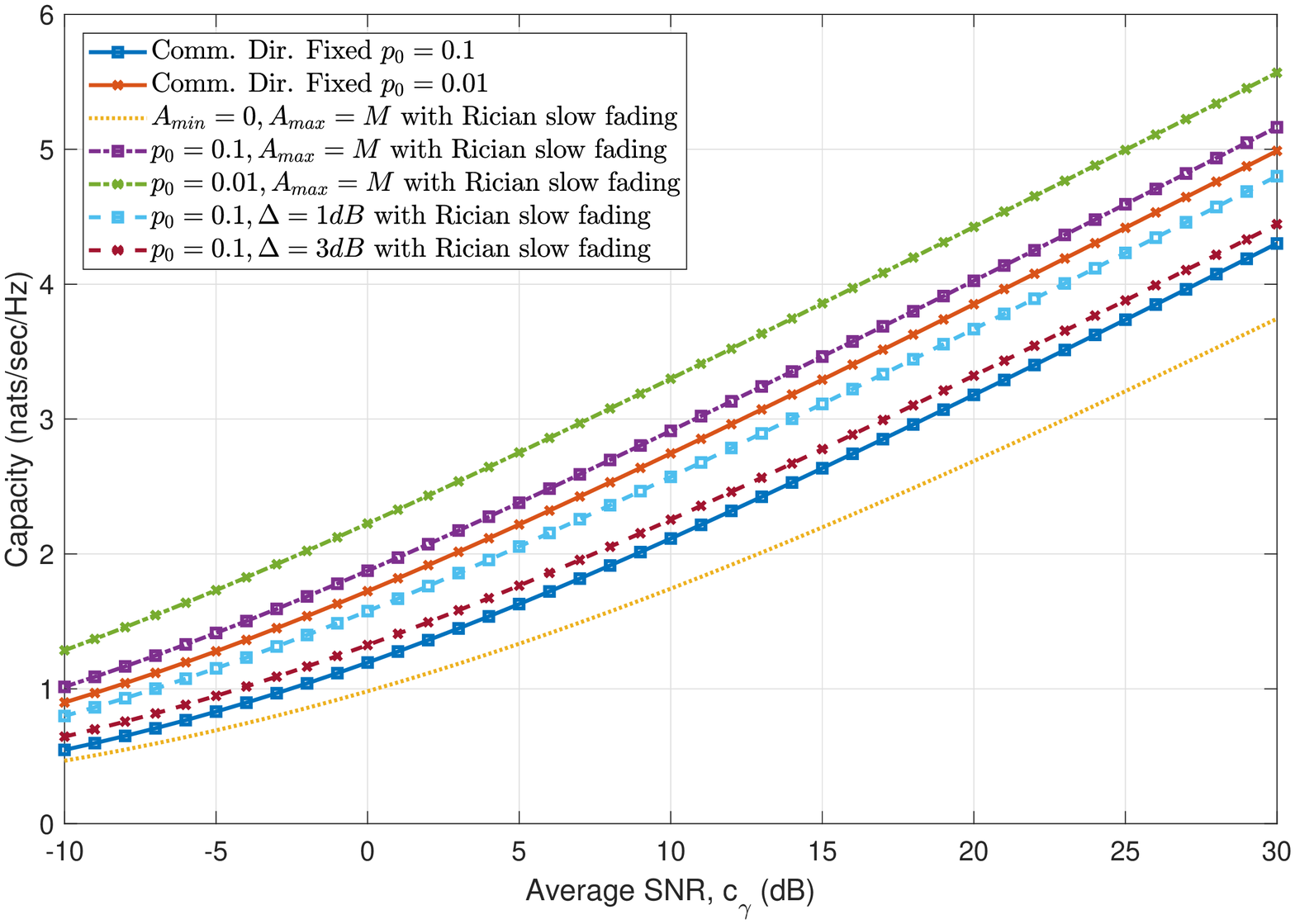}
		\label{fig:CapacitySlowFadingRician}
	}
	\caption{Capacity for JRC Capable Coherent MIMO Radar under \textcolor{blue}{Rayleigh (a) and Rician (b)} slow fading conditions.}
	\label{fig:ALLCapacitySlowFading}
\end{figure*}
\begin{figure*}[htb]
	\centering
	\subfigure[]
	{
		\includegraphics[width=0.45\textwidth]{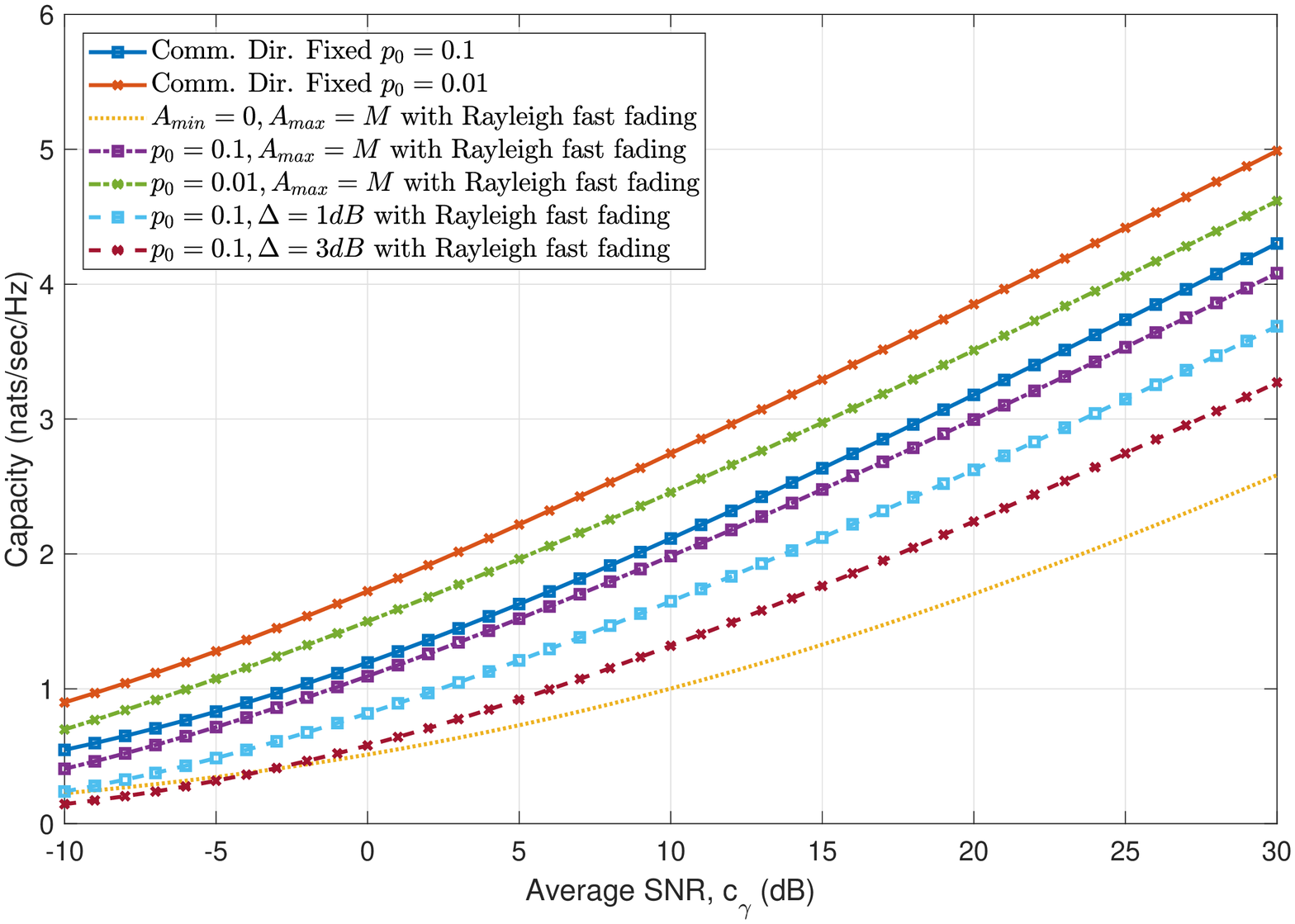}
		\label{fig:CapacityFastFading}
	}
	\subfigure[]
	{
		\includegraphics[width=0.45\textwidth]{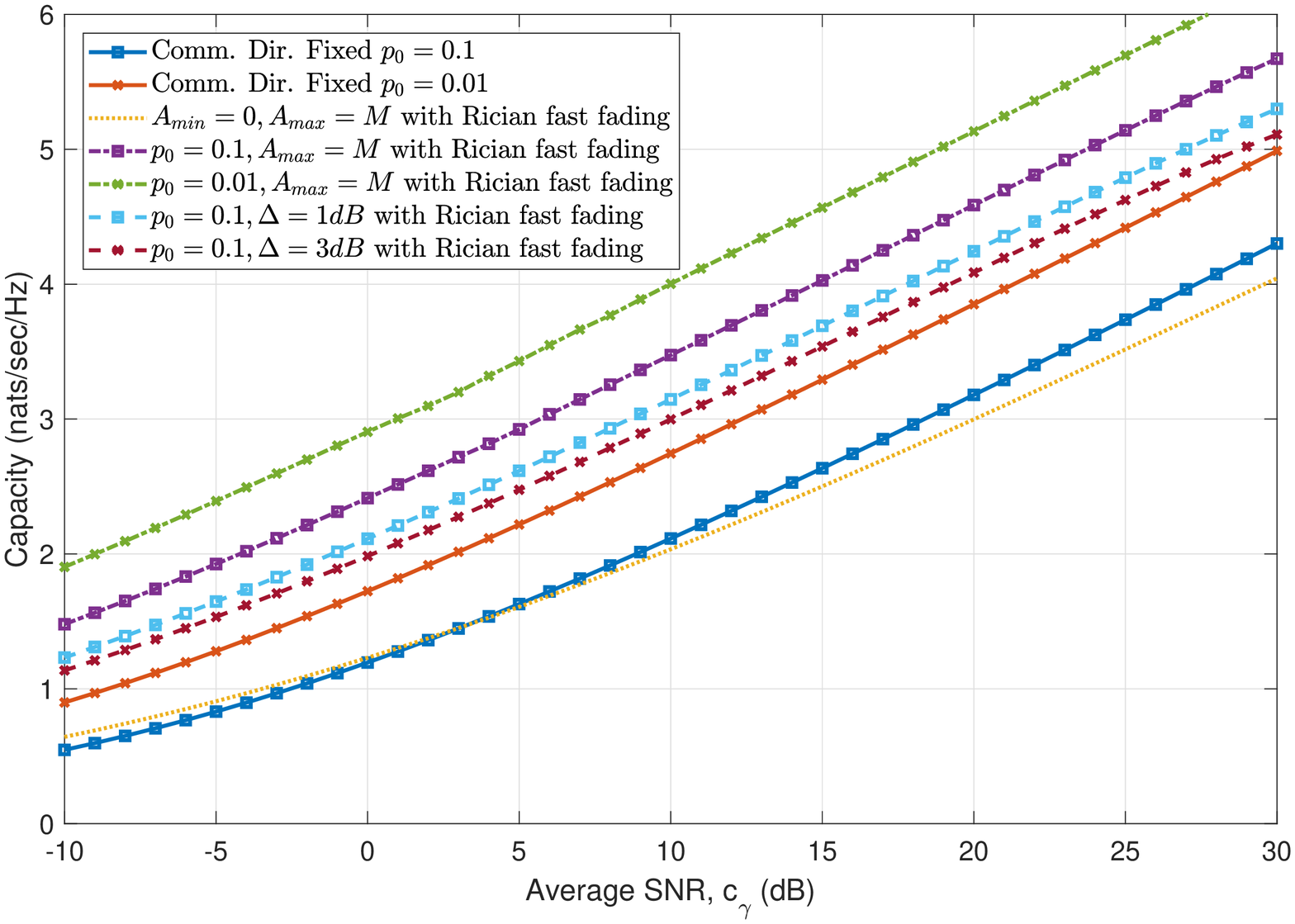}
		\label{fig:CapacityFastFadingRician}
	}
	\caption{Capacity for JRC Capable Coherent MIMO Radar under \textcolor{blue}{Rayleigh (a) and Rician (b)} fast fading conditions.}
	\label{fig:ALLCapacityFastFading}
\end{figure*}

\subsubsection{Case 3: Fast Fading Rayleigh Channel} 
In this case, we have assumed that channel has encountered all possible fades over a symbol period\textcolor{blue}{. C}hannel \textcolor{blue}{is} characterized by \textcolor{blue}{the} product distribution\textcolor{blue}{s} which \textcolor{blue}{are} introduced and analyzed \textcolor{blue}{in the} previous sections. This distribution has no truncated shape and it \textcolor{blue}{almost }spans over all values zero to infinity. This results relatively small degradation on the capacity for all methods under fast fading conditions as seen on Fig.\ref{fig:ALLCapacityFastFading}. Capacity calculations are done for $M=16$, $\sigma_H^2=0.5$ \textcolor{blue}{and $K=3$} using the Eq.(\ref{CtdrayFast}) and Eq.(\ref{CtdrayFast2}). Since \textcolor{blue}{all the capacity expressions are} not in closed form, we \textcolor{blue}{have} evaluate\textcolor{blue}{d} the integral with numerical methods using trapezoidal methods. \textcolor{blue}{These calculations are} evaluated with MATLAB. \textcolor{blue}{Specifically, capacity curves for Rician fast fading case displays highest capacity by comparing all the other capacity curves.} 

\section{Conclusion}
The increasing interest on connected devices causes crowded spectrum. To solve this problem, various efforts on developing crafty and compact technologies \textcolor{blue}{are proposed} as Joint Radar-Communication (JRC) systems. JRC ability can be established for coherent MIMO radar without disturbing the orthogonality and transmit beamforming requirements \textcolor{blue}{as in \cite{RealcMIMO}}. In this paper, the distribution characteristic of the unintentional modulation on signal amplitude is investigated. Then, the capacity expressions for different waveform generation methods under AWGN, slow and fast fading conditions are analyzed and evaluated. The results \textcolor{blue}{indicate} that, using the null direction fixed waveform generation method with single constraint as $A_l>A_0$ where $A_0$ is equal to $\sqrt{-M\ln(p_0)}$  displays the best capacity curves \textcolor{blue}{for under all channel conditions}. While reaching the capacity expressions, the pdf\textcolor{blue}{'s} of the product of a double truncated Rayleigh r.v. and a Rayleigh\textcolor{blue}{/Rician} distributed r.v. \textcolor{blue}{are} expressed for the first time in the literature. Lastly, the resulting expression\textcolor{blue}{s given in this paper are} also proven with \textcolor{blue}{the} numerical evaluation.

Proposed MIMO radar waveform generation process in \cite{RealcMIMO} under mobile scenarios may introduce computational complexity problems. To be more precise, the system must calculate the MIMO radar waveform in a very short period of time, since the direction of the receiver may change continuously due to the highly dynamic scenarios. The worst case may be that the MIMO waveform set changes from pulse to pulse. In order to reduce the computational cost, required number of iteration to generate MIMO waveform set can be decreased by loosing constraints. However, this will introduce some degradation on communication performance. This degradation is investigated in detail in \cite{RealcMIMO}. MIMO waveform generation algorithm in \cite{RealcMIMO} contains random process for each iteration. If the random permutation tables are generated off-line and stored to the system before the operation, the remaining process, i.e. summation and multiplication, will be fast and easy to implement on a system with a parallel processing capability. 

\appendix

\section{Proof of Amplitude Distribution of the Unintentional Modulation}
\label{ProofUnintAmpProof}

The Rayleigh distribution is the distribution of the sum of a large number of coplanar vectors with constant amplitudes and uniformly distributed phases \cite{RayleighBeckmann}, \cite{RayleighBeckmann2}. The sum in Eq.(\ref{sumofexponentials}) can be represented as a complex sum of random vectors as in Fig.\ref{fig:randomWalk}. This problem is very similar to the "random-walk problem" in mathematical statistics. Hence, if $\phi_m(l)$ presents a uniform distribution, $A_l$ will show Rayleigh distribution. Any phase distribution, $w_{\phi}(\phi)$, can be defined as a uniform distribution \cite{RayleighBeckmann2}, if it satisfies the condition below, 

\begin{figure}[t]
	\centering
	\includegraphics[width=0.3\textwidth]{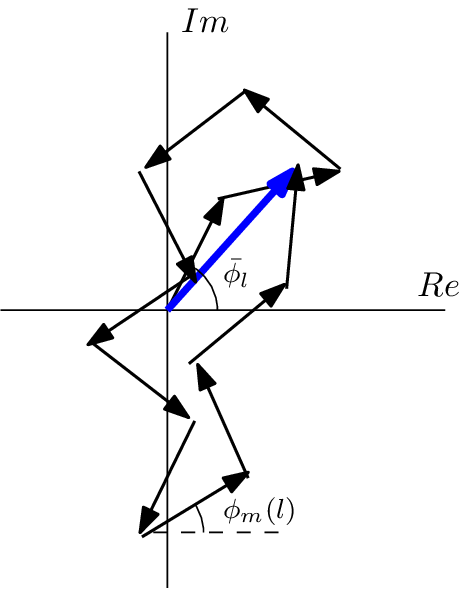}
	\caption{Random unit vector sum in the complex plane.}
	\label{fig:randomWalk}
\end{figure}

\begin{equation}
\label{uniformPr}
\sum_{k=-\infty}^{\infty}w_{\phi}(\phi+2k\pi)=\frac{1}{2\pi}+\epsilon(\phi),\qquad |\epsilon(\phi)|\ll\frac{1}{2\pi}.
\end{equation}

The first step is analyzing the phase component of $G_l(\theta_{c})$, which is $\phi_m(l)$. $\phi_m(l)$ can be written as a sum of two function as,
\begin{equation}
\begin{split}
\phi_m(l) =&\phi_{I_m(l)}+2\pi d\sin{\theta_{n}}I_m+2\pi d(\sin{\theta_{n}}\\
&-\sin{\theta_{c}})m-2\pi d\sin{\theta_{c}},
\end{split}
\end{equation}where $I_m\in\{I_{1},I_{2},...,I_{M}\}$ is the indexes of a random permutation sequence of $m\in\{1,2,...,M\}$ for each sub-pulse $l$ and $\phi_{I_{m}}$ is a phase component of the $m$-th antenna element and,

\begin{equation}
\label{fm}
\begin{split}
\tilde{h}(m,l)=&\phi_{I_m(l)}+2\pi d\sin{\theta_{n}}I_m(l)\\
=&\phi_{I_m(l)} + c_nI_m(l),\\
z(m)=&2\pi d(\sin{\theta_{n}}-\sin{\theta_{c}})m-2\pi d\sin{\theta_{c}}\\
=&mc_n-(m+1)c_c,
\end{split}
\end{equation}where $m=1,2,...,M$,  $l=1,2,...,L$, $c_n$ and $c_c$ are constants, equal to $2\pi d\sin{\theta_{n}}$, $2\pi d\sin{\theta_{c}}$, respectively.
Let $\phi_{m}$ is one time selected randomly from $(-\pi,\pi]$ and $M$ is a large number, $M\gg1$. $\tilde{h}(m,l)$ takes only $M$ different values, for each iteration these values are shuffled and summed with $z(m)$ which is a linear function. Hence, $\tilde{h}(m,l)$ takes values uniformly from the range $[-\pi+c_n,\pi+Mc_n]$. $z(m)$ is a linear discrete function and takes M different values from the range $[c_n+2c_c,(M+1)c_c]$. 

Let $\Phi$, $\tilde{H}$ and $Z$ are random variables which are selected from functions, $\phi_m(l)$, $\tilde{h}(m,l)$ and $z(m)$ respectively. Then, the distribution of $\phi_m(l)$ can be written as,
\begin{equation}
\label{dist1}
f_{\Phi}(\phi)=f_{\tilde{H}}(\tilde{h})*f_Z(z)=\int_{\mathbb{R}}f_{\tilde{H}}(\phi-z)f_Z(z)dz
\end{equation}where ($*$) is convolution operation and probability density function of $z(m)$, $f_Z(z)$, given as,
\begin{equation}
f_Z(z)=\sum_{m=1}^{M}p_Z(mc_n-(m+1)c_c)\delta(z-(mc_n-(m+1)c_c)).
\end{equation}The distribution of $h(m,l)$, $f_{\tilde{H}}(\tilde{h})$ can be expressed as,
\begin{equation}
f_{\tilde{H}}(\tilde{h})=\sum_{l=1}^{L}\sum_{m=1}^{M}p_{\tilde{H}}(\phi_{I_m(l)}+c_nI_m(l))\delta(\tilde{h}-(\phi_{I_m(l)}+c_nI_m(l))).
\end{equation}
Let each value of $z(m)$ is distinct, then $p_Z(mc_n-(m+1)c_c)=1/M$. Then, $\tilde{h}(m,l)$ is said to be each value of it is distinct due to $\phi_{m}$ is selected randomly, then $p_{\tilde{H}}(\phi_{I_m}(l)+c_nI_m(l))$ becomes $1/M$. The probability density function $f_\Phi(\phi)$ can be rewritten as,
\begin{equation}
\label{pdfDelta1}
\begin{split}
f_\Phi(\phi)=&\int_{\mathbb{R}}f_{\tilde{H}}(\phi-g)\sum_{m=1}^{M}p_{\tilde{H}}(mc_n-(m+1)c_c)\\
&\times\delta(g-mc_n-(m+1)c_c)dg\\
=&\frac{1}{M}\sum_{m=1}^{M}\int_{\mathbb{R}}f_{\tilde{H}}(\phi-g)\delta(g-mc_n-(m+1)c_c)dg
\end{split}
\end{equation}and rewriting Eq.(\ref{pdfDelta1}),
\begin{equation}
\label{pdfDelta11}
\begin{split}
f_\Phi(\phi)=&\frac{1}{M}\sum_{m=1}^{M}f_{\tilde{H}}(\phi-mc_n-(m+1)c_c)\\
=&\frac{1}{M}\sum_{l=1}^{L}\sum_{m=1}^{M}p_{\tilde{H}}(\phi_{I_m}(l)+c_nI_m(l))\\
&\times\delta(\phi-(\phi_{I_m}(l)+c_nI_m(l)+mc_n-(m+1)c_c))\\
=&\frac{1}{L}\frac{1}{M}\sum_{l=1}^{L}\sum_{m=1}^{M}\delta(\phi-(\phi_{I_m}(l)+c_nI_m(l)+mc_n-(m+1)c_c))\\
=&\frac{1}{LM}\sum_{l=1}^{L}\sum_{m=1}^{M}\delta(\phi-\phi_m(l)),
\end{split}
\end{equation}and since the $\phi_m$ is a phase component, it has periodic behavior as, $\phi_m=2\pi+\phi_m$. Therefore, $\delta(\phi-\phi_m(l))$ is equal to $\delta(\phi+2\pi-\phi_m(l))$. Therefore, we can expand dirac delta function with Fourier series expansion, (\ref{pdfDelta11}) can be rewritten as,
\begin{equation}
\label{pdfDelta2}
\begin{split}
f_\Phi&(\phi)=\frac{1}{LM}\sum_{l=1}^{L}\sum_{m=1}^{M}\bigg[\frac{1}{2\pi}+\frac{1}{\pi}\sum_{n=1}^{\infty}\bigg[\cos(n\phi_m(l))\cos(n\phi)\\
&+\sin(n\phi_m(l))\sin(n\phi)\bigg]\bigg]\\
=&\frac{1}{LM}\sum_{l=1}^{L}\sum_{m=1}^{M}\left[\frac{1}{2\pi}+\frac{1}{\pi}\sum_{n=1}^{\infty}\cos(n(\phi-\phi_m(l)))\right]
\end{split}
\end{equation}and after reordering Eq.(\ref{pdfDelta2}),
\begin{equation}
\label{pdfDelta22}
\begin{split}
f_\Phi&(\phi)=\frac{1}{2\pi}+\frac{1}{\pi LM}\sum_{n=1}^{\infty}\sum_{l=1}^{L}\sum_{m=1}^{M}\cos(n(\phi-\phi_m(l)))\\
=&\frac{1}{2\pi}+\frac{1}{\pi LM}\sum_{n=1}^{\infty}\big[\cos(n(\phi-\phi_1(1)))+\cos(n(\phi-\phi_2(1)))\\
&+...+\cos(n(\phi-\phi_M(L)))\big].
\end{split}
\end{equation}Since $M$ and $L$ are large numbers, we can always find a phase component couple $\{\phi_{m_1}(l_1),\phi_{m_2}(l_2)\}$ for $m_1,m_2\in\{1,...,M\}$ and $l_1,l_2\in\{1,...,L\}$, with a relation of,
\begin{equation}
\cos(n(\phi-\phi_{m_1}(l_1)))+\cos(n(\phi-\phi_{m_2}(l_2)))\approx 0,
\end{equation}when $\phi_{m_1}(l_1)\approx\phi_{m_2}(l_2)+\pi$. In order to prove this, $\cos(n(\phi-\phi_{m_1}(l_1)))$ can be written in terms of $\phi_{m_2}(l_2)$ as,
\begin{equation}
\label{relationPhi}
\begin{split}
\cos&(n(\phi-\phi_{m_1}(l_1)))=\cos(n\phi)\cos(n\phi_{m_1}(l_1))\\
&+\sin(n\phi)\sin(n\phi_{m_1}(l_1))\\
\approxeq&\cos(n\phi)\cos(n\phi_{m_2}(l_2)+n\pi)\\
&+\sin(n\phi)\sin(n\phi_{m_2}(l_2)+n\pi)\\
\approxeq&-\cos(n\phi)\cos(n\phi_{m_2}(l_2))-\sin(n\phi)\sin(n\phi_{m_2}(l_2))\\
\approxeq&-\cos(n(\phi-\phi_{m_2}(l_2))).
\end{split}
\end{equation}Hence, (\ref{pdfDelta22}) can be rewritten as using the relation (\ref{relationPhi}),
\begin{equation}
\label{pdfDelta3}
\begin{split}
f_\Phi&(\phi)=\frac{1}{2\pi}+\frac{1}{\pi LM}\\
&\times\sum_{n=1}^{\infty}\bigg[\cos(n(\phi-\phi_1(1)))+\cos(n(\phi-\phi_1(1)-\pi))+...\\
&+\cos(n(\phi-\phi_{m_1}(l_1)))+\cos(n(\phi-\phi_{m_1}(l_1)-\pi))+...\\
&+\cos(n(\phi-\phi_{M}(L)))+\cos(n(\phi-\phi_{M}(L)-\pi))\bigg].
\end{split}
\end{equation}Therefore, when $M\to\infty$, the probability density function $f_\Phi(\phi)$ approximates uniform distribution, $U(-\pi,\pi)$, as,
\begin{equation}
\label{pdfDelta4}
\begin{split}
\lim\limits_{M\to\infty}f_\Phi(\phi)\approxeq&\frac{1}{2\pi}+\frac{1}{\pi LM}\sum_{n=1}^{\infty}\left[0+...+0\right]\\
\approxeq&\frac{1}{2\pi}.
\end{split}
\end{equation}Since the phase distribution approximates uniform distribution, the distribution of sum of exponentials will follows Rayleigh distribution. 

\section{Proof of Product Distribution of Truncated Rayleigh and Rayleigh Random Variables}
\label{ProofProductDistTrunRay}

There are two methods to reach product distribution. First is Mellin Convolution \cite{BookProbMath} which is defined as,
\begin{equation}
\label{MellinConvolution}
p_{X}(x)=\int_{-\infty}^{\infty}p_{X_1}(x_1)p_{X_2}\left(\frac{x}{x_1}\right)\frac{1}{|x_1|}dx_1.
\end{equation}where $X=X_1X_2$. Then, we can reach the pdf of the r.v. $X$ replacing the pdf functions with double truncated Rayleigh and Rayleigh distributions,
\begin{equation}
\begin{split}
p_{X}(x)&=\int_{X^{min}_1}^{X^{max}_1}\frac{x}{x_1\sigma_{X_2}^2}e^{\frac{-x^2}{x_1^22\sigma_{X_2}^2}}\frac{x_1}{\alpha(\beta_1,\beta_2)\sigma_{X_1}^2}e^{\frac{-x_1^2}{2\sigma_{X_1}^2}}\frac{1}{x_1}dx_1\\
&=\frac{x}{\alpha(\beta_1,\beta_2)\sigma_{X_1}^2\sigma_{X_2}^2}\int_{X^{min}_1}^{X^{max}_1}\frac{1}{x_1}e^{\left(-\frac{x^2}{x_1^22\sigma_{X_2}^2}-\frac{-x_1^2}{2\sigma_{X_1}^2}\right)}dx_1
\end{split}
\end{equation}where $\beta_1=\frac{(X^{min}_1)^2}{2\sigma_{X_1}^2}$, $\beta_2=\frac{(X^{max}_1)^2}{2\sigma_{X_1}^2}$, and by substituting $v=\frac{x_1^2}{2\sigma_{X_1}^2}$ using the transformation theorem we have,
\begin{equation}
\label{mellinconvint1}
\begin{split}
p_{X}(x)&=\frac{x}{\alpha(\beta_1,\beta_2)2\sigma_{X}^2}\int_{\beta_1}^{\beta_2} \frac{1}{v} e^{\left(-v-\left(\frac{x^2}{4\sigma_{X}^2}\right)\frac{1}{v}\right)}dv\\
&=\frac{x}{\alpha(\beta_1,\beta_2)2\sigma_{X}^2}\bigg[\int_{\beta_1}^{\infty} \frac{1}{v} e^{\left(-v-\left(\frac{x^2}{4\sigma_{X}^2}\right)\frac{1}{v}\right)}dv\\
&-\int_{\beta_2}^{\infty} \frac{1}{v} e^{\left(-v-\left(\frac{x^2}{4\sigma_{X}^2}\right)\frac{1}{v}\right)}dv\bigg]\\
&=\frac{x\left[\Gamma\left(0,\beta_1;\frac{x^2}{4{\sigma_X}^2}\right) - \Gamma\left(0,\beta_2;\frac{x^2}{4{\sigma_X}^2}\right)\right]}{\alpha(\beta_1,\beta_2)2\sigma_{X}^2},
\end{split}
\end{equation}where $\Gamma(a,x;b)$ is the generalized incomplete gamma function \cite{GenIncGamma} as $\Gamma(a,x;b)=\int_{x}^{\infty}t^{a-1}e^{-t-bt^{-1}}dt$. 

Second method uses Mellin transform to reach product distribution. Mellin transform of the product of two independent r.v. from different distributions is equal to the product of their Mellin transforms, $\mathcal{M}_{X_1X_2}=\mathcal{M}_{X_1}\mathcal{M}_{X_2}$. Mellin transform of a distribution $p_X(x)$ can be given as \cite{IntegTrans},
\begin{equation}
\mathcal{M}{p_{X}(x)}=\psi(s)=\int_{0}^{\infty}x^{s-1}p_{X}(x)dx,
\end{equation}and the inverse Mellin transform is,
\begin{equation}
\mathcal{M}^{-1}{\psi(s)}=p_{X}(x)=\frac{1}{2\pi j}\oint_{c-j\infty}^{c+j\infty}x^{-s}\psi(s)ds.
\end{equation}where $j=\sqrt{-1}$ and the integration is along any path $Re(s)=c$, such that $\psi(s)$ exists. Then, Mellin transform of the r.v. $X_1$ can be expressed as,
\begin{equation}
\begin{split}
&\mathcal{M}{p_{X_1}(x_1)}=\frac{1}{\alpha(\beta_1,\beta_2)}\int_{X^{min}_1}^{X^{max}_1}x^{s-1}\frac{x_1}{\sigma_{X_1}^2}e^{\frac{-x_1^2}{2\sigma_{X_1}^2}}dx_1\\
&=\frac{1}{\alpha(\beta_1,\beta_2)}\bigg[\int_{X^{min}_1}^{\infty}x^{s-1}\frac{x_1}{\sigma_{X_1}^2}e^{\frac{-x_1^2}{2\sigma_{X_1}^2}}dx_1-\int_{X^{max}_1}^{\infty}x^{s-1}\frac{x_1}{\sigma_{X_1}^2}e^{\frac{-x_1^2}{2\sigma_{X_1}^2}}dx_1\bigg]\\
&=\frac{1}{\alpha(\beta_1,\beta_2)}\bigg[\left(\frac{1}{2\sigma_{X_1}^2}\right)^{\frac{1-s}{2}}\Gamma\left(\frac{s+1}{2},\beta_1\right)-\left(\frac{1}{2\sigma_{X_1}^2}\right)^{\frac{1-s}{2}}\Gamma\left(\frac{s+1}{2},\beta_2\right)\bigg]
\end{split}
\end{equation}
\begin{equation}
\mathcal{M}{p_{X_2}(x_2)}=\int_{0}^{\infty}x_2^{s-1}\frac{x_2}{\sigma_{X_2}^2}e^{\frac{-x_2^2}{2\sigma_{X_2}^2}}dx_2=\left(\frac{1}{2\sigma_{X_2}^2}\right)^{\frac{1-s}{2}}\Gamma\left(\frac{s+1}{2}\right)
\end{equation}where $\Gamma(.)$ denotes the gamma function \cite{BookHandMath} and  $\Gamma(a,x)$ is the incomplete gamma integral \cite{GenIncGamma}, $\Gamma(a,x)=\int_{x}^{\infty}t^{a-1}e^{-t}dt$. Hence, Mellin transform of the product becomes,
\begin{equation}
\mathcal{M}{p_{X}(x)}=\psi(s)=\frac{\left(\frac{1}{4\sigma_{X}^2}\right)^{\frac{1-s}{2}}\Gamma\left(\frac{s+1}{2}\right)\left[\Gamma\left(\frac{s+1}{2},\beta_1\right) -\Gamma\left(\frac{s+1}{2},\beta_2\right)\right]}{\alpha(\beta_1,\beta_2)}.
\end{equation}Then, by taking the inverse Mellin transform,
\begin{equation}
\begin{split}
&p_{X}(x)=\mathcal{M}^{-1}\left[\psi(s)\right]\\
&=\frac{\frac{1}{\alpha(\beta_1,\beta_2)}}{2\pi j}\bigg[\oint\limits_{\mathrm{C}}x^{-s}\left(\frac{1}{4\sigma_{X}^2}\right)^{\frac{1-s}{2}}\Gamma\left(\frac{s+1}{2}\right)\Gamma\left(\frac{s+1}{2},\beta_1\right)ds\\
& - \oint\limits_{\mathrm{C}}x^{-s}\left(\frac{1}{4\sigma_{X}^2}\right)^{\frac{1-s}{2}}\Gamma\left(\frac{s+1}{2}\right)\Gamma\left(\frac{s+1}{2},\beta_2\right)ds\bigg],
\end{split}
\end{equation}product distribution can be defined with a contour integral above. If we replace $(s+1)/2$ by $v$ in the above equation, this	does not affect the path of integration and the integral becomes,
\begin{equation}
\begin{split}
p_{X}(x)&=\frac{\frac{1}{\alpha(\beta_1,\beta_2)}}{2\pi j}\bigg[\oint\limits_{\mathrm{C}}x^{1-2v}\left(\frac{1}{4\sigma_{X}^2}\right)^{1-v}\Gamma\left(v\right)\Gamma\left(v,\beta_1\right)2dv\\
& - \oint\limits_{\mathrm{C}}x^{1-2v}\left(\frac{1}{4\sigma_{X}^2}\right)^{1-v}\Gamma\left(v\right)\Gamma\left(v,\beta_2\right)2dv\bigg]\\
&=\frac{\frac{1}{\alpha(\beta_1,\beta_2)}}{2\pi j}\left(\frac{x}{2\sigma_{X}^2}\right) \bigg[\oint\limits_{\mathrm{C}}\left(\frac{x^2}{4\sigma_{X}^2}\right)^{-v}\Gamma\left(v\right)\Gamma\left(v,\beta_1\right)dv\\
& - \oint\limits_{\mathrm{C}}\left(\frac{x^2}{4\sigma_{X}^2}\right)^{-v}\Gamma\left(v\right)\Gamma\left(v,\beta_2\right)dv\bigg].
\end{split}
\end{equation}Generalized incomplete gamma functions have inverse Mellin transform representations from \cite[Eq.~7.1]{ExtendGenIncGamma} which can be given by,
\begin{equation}
\label{GenincompToIncomp}
\Gamma(a,x;b)=\frac{1}{2\pi j}\oint_{c-j\infty}^{c+j\infty}\Gamma\left(s\right)\Gamma\left(a+s,x\right)b^{-s}ds.
\end{equation}After applying the Eq.(\ref{GenincompToIncomp}), $p_{X}(x)$ can be rewritten as,
\begin{equation}
p_{X}(x)=\frac{x\left[\Gamma\left(0,\beta_1;\frac{x^2}{4{\sigma_X}^2}\right) - \Gamma\left(0,\beta_2;\frac{x^2}{4{\sigma_X}^2}\right)\right]}{\alpha(\beta_1,\beta_2)2\sigma_{X}^2}.
\end{equation}

\section{Proof of Product Distribution of Truncated Rayleigh and Rician Random Variables}
\label{ProofProductDistTrunRayAndRice}
\textcolor{blue}{
Mellin Convolution in Eq.\ref{MellinConvolution} is also used for the derivation of the joint pdf of the product distribution of truncated Rayleigh and Rician r.v's. We can reach the pdf of the r.v. $X$ replacing the pdf functions with double truncated Rayleigh and Rician distributions,}
\textcolor{blue}{
\begin{equation}
\begin{split}
p_{X}(x)&=\int_{X^{min}_1}^{X^{max}_1}\frac{xe^{\left(\frac{-x^2}{x_1^22\sigma_{X_2}^2}+\frac{-\mu^2}{2\sigma_{X_2}}\right)}}{x_1\sigma_{X_2}^2}I_0\left(\frac{\mu x}{x_1\sigma_{X_2}}\right)\frac{x_1e^{\frac{-x_1^2}{2\sigma_{X_1}^2}}}{\alpha(\beta_1,\beta_2)\sigma_{X_1}^2}\frac{1}{x_1}dx_1\\
&=\frac{xe^{\frac{-\mu^2}{2\sigma_{X_2}}}}{\alpha(\beta_1,\beta_2)\sigma_{X_1}^2\sigma_{X_2}^2}\int_{X^{min}_1}^{X^{max}_1}\frac{1}{x_1}e^{\left(-\frac{x^2}{x_1^22\sigma_{X_2}^2}-\frac{-x_1^2}{2\sigma_{X_1}^2}\right)}I_0\left(\frac{\mu x}{x_1\sigma_{X_2}}\right)dx_1
\end{split}
\end{equation}where $\beta_1=\frac{(X^{min}_1)^2}{2\sigma_{X_1}^2}$, $\beta_2=\frac{(X^{max}_1)^2}{2\sigma_{X_1}^2}$. After applying the series representation of the $I_0(h)$ which is $\sum_{k=0}^{\infty}\frac{\left(\frac{h^2}{4}\right)^k}{k!\Gamma(k+1)}$, $p_X(x)$ becomes,}
\textcolor{blue}{
\begin{equation}
\begin{split}
p_{X}(x)=&\frac{xe^{\frac{-\mu^2}{2\sigma_{X_2}}}}{\alpha(\beta_1,\beta_2)2\sigma_{X}^2}\int_{X^{min}_1}^{X^{max}_1}\frac{1}{x_1}e^{\left(-\frac{x^2}{x_1^22\sigma_{X_2}^2}-\frac{x_1^2}{2\sigma_{X_1}^2}\right)}\sum_{i=0}^{\infty}\frac{\left(\frac{\mu^2x^2}{4x_1^2\sigma_{X_2}^2}\right)^i}{i!\Gamma(i+1)}dx_1\\
=&\frac{xe^{\frac{-\mu^2}{2\sigma_{X_2}}}}{\alpha(\beta_1,\beta_2)2\sigma_{X}^2}\sum_{i=0}^{\infty}\frac{\left(\frac{\mu^2x^2}{4\sigma_{X_2}^2}\right)^i}{i!i!}\int_{X^{min}_1}^{X^{max}_1}\frac{1}{x_1^{(2i+1)}}e^{\left(-\frac{x^2}{x_1^22\sigma_{X_2}^2}-\frac{x_1^2}{2\sigma_{X_1}^2}\right)}dx_1
\end{split}
\end{equation}and by substituting $v=\frac{x_1^2}{2\sigma_{X_1}^2}$ using the transformation theorem we have,}
\textcolor{blue}{
\begin{equation}
\label{mellinconvint12}
\begin{split}
p_{X}(x)&=\frac{xe^{\frac{-\mu^2}{2\sigma_{X_2}}}}{\alpha(\beta_1,\beta_2)2\sigma_{X}^2}\sum_{i=0}^{\infty}\frac{\left(\frac{x\mu}{2}\right)^{2i}\left(\frac{1}{2\sigma_{X_2}^2\sigma_{X}^2}\right)^i}{i!i!}\int_{\beta_1}^{\beta_2} \frac{e^{\left(-v-\left(\frac{x^2}{4\sigma_{X}^2}\right)\frac{1}{v}\right)}}{v^{i+1}}dv\\
&=\frac{xe^{\frac{-\mu^2}{2\sigma_{X_2}}}}{\alpha(\beta_1,\beta_2)2\sigma_{X}^2}\sum_{i=0}^{\infty}\frac{\left(\frac{x\mu}{2}\right)^{2i}}{i!i!\left(2\sigma_{X_2}^2\sigma_{X}^2\right)^i}\bigg[\int_{\beta_1}^{\infty} \frac{e^{\left(-v-\left(\frac{x^2}{4\sigma_{X}^2}\right)\frac{1}{v}\right)}}{v^{i+1}}dv\\
&-\int_{\beta_2}^{\infty} \frac{e^{\left(-v-\left(\frac{x^2}{4\sigma_{X}^2}\right)\frac{1}{v}\right)}}{v^{i+1}}dv\bigg]\\
&=\frac{xe^{\frac{-\mu^2}{2\sigma_{X_2}^2}}}{\alpha(\beta_1,\beta_2)2{\sigma_X}^2}\sum_{i=0}^{\infty}\frac{\Gamma\left(-i,\beta_1;\frac{x^2}{4{\sigma_X}^2}\right)-\Gamma\left(-i,\beta_2;\frac{x^2}{4{\sigma_X}^2}\right)}{i!i!(2\sigma_{X_2}^2\sigma_{X}^2)^{i}\left(\frac{x\mu}{2}\right)^{-2i}},
\end{split}
\end{equation}where $\Gamma(a,x;b)$ is the generalized incomplete gamma function \cite{GenIncGamma} as $\Gamma(a,x;b)=\int_{x}^{\infty}t^{a-1}e^{-t-bt^{-1}}dt$.}

\section{A Numerical Method for the Calculation of Generalized Incomplete Gamma Functions}
\label{NumericGenIncGamFunc}

In \cite{GenIncGammaBessel}, authors present a numerical method to compute generalized incomplete gamma functions based on Exponential Integrals with a relative precision of $10^{-6}$ and an absolute precision of $10^{-25}$. This approach is also given in \cite{AsymGenIncGamFnc}. This method can be written as,

If $a$, in $\Gamma(a,x;b)$, is an integer, $a=q\le 0$:

\begin{equation}
\Gamma(q,x;b)=\sum_{n=0}^{\infty}\frac{(-b/x)^n}{n!}x^qE_{-q+n+1}(x),\quad x\ge\sqrt{b}, 
\end{equation}and,
\begin{equation}
\Gamma(q,x;b)=2b^{\frac{q}{2}}K_q\left(2\sqrt{b}\right)-\sum_{n=0}^{\infty}\frac{-x^n}{n!}x^qE_{q+n+1}(x),\quad 0\le x\le\sqrt{b},
\end{equation}where $E_{n}(x)=\int_{1}^{\infty}t^{-n}e^{-xt}dt,\quad x>0, n=0,1,...$

\end{document}